\documentclass[11pt]{article}
\usepackage{subfig,amsmath,amssymb,amsthm,graphicx,tikz,bbm,bm,geometry,setspace,xcolor,wrapfig,arydshln,pgf,caption, multirow, threeparttable, pdflscape}
\usepackage{hyperref}
\hypersetup{
    colorlinks=true,
    linkcolor=red,
	citecolor=blue, 	
 }
\usepackage[authoryear,longnamesfirst]{natbib}
\usepackage[utf8]{inputenc}
\usepackage{booktabs,dcolumn, epstopdf, chngcntr, multirow}
\usepackage{subfiles}
\counterwithout{figure}{section}
\usetikzlibrary{patterns,arrows,decorations.pathreplacing, automata}
\tikzset{line/.style={-latex'}}
\usetikzlibrary{fadings}
\newtheorem{theorem}{Theorem}[section]
\newtheorem{lemma}{Lemma}[section]

\newtheorem{corollary}{Corollary}[section]
\newtheorem{definition}{Definition}
\newtheorem{example}{Example}

\newtheorem{assumption}{Assumption}[subsection]
\DeclareMathOperator*{\argmax}{arg\,max}

\DeclareMathOperator*{\supp}{supp}
\DeclareMathOperator*{\sign}{sign}

\colorlet{xred}{red!70!black}
\begin{document}
\title{\textbf{A Semiparametric Network Formation Model} \\ \textbf{with Unobserved Linear Heterogeneity}\footnote{First version: November, 2016.  A previous version of this paper was titled: ``A Semiparametric Network Formation Model with Multiple Linear Fixed Effects."}}
\author{Luis E. Candelaria\thanks{Department of Economics, University of Warwick, Coventry, U.K. Email: \texttt{L.Candelaria@warwick.ac.uk}. 
\newline
I am deeply grateful to Federico Bugni, Shakeeb Khan, Arnaud Maurel, and Matthew Masten for their excellent guidance, constant encouragement, and helpful discussions. 
I also thank Irene Botosaru, \'{A}ureo de Paula, Andreas Dzemski, Cristina Gualdani, Bryan Graham, Bo Honor\'{e}, Arthur Lewbel, Thierry Magnac, Chris Muris, James Powell, Adam Rosen, Takuya Ura, Martin Weidner, 
and seminar participants at Aarhus, Cambridge, Duke, Gothenburg,  LSE, Surrey, Syracuse, TSE, UCL, UNC Chapel Hill, Vanderbilt, Warwick,  2018 ES Winter Meeting in Philadelphia, 2019 Panel Data Workshop at the University of Amsterdam, 2019 Royal Economic Society at the University of Warwick, for their comments.
}}
\date{\today}
\maketitle
\begin{abstract}
This paper analyzes a semiparametric model of network formation in the presence of unobserved agent-specific heterogeneity. 
The objective is to identify and estimate the preference parameters associated with homophily on observed attributes when the distributions of the unobserved factors are not parametrically specified.
This paper offers two main contributions to the literature on network formation. 
First, it establishes a new point identification result for the vector of parameters that relies on the existence of a special regressor. 
The identification proof is constructive and characterizes a closed-form for the parameter of interest.
Second, it introduces a simple two-step semiparametric estimator for the vector of parameters with a first-step kernel estimator.
The estimator is computationally tractable and can be applied to both dense and sparse networks.
Moreover, I show that the estimator is consistent and has a limiting normal distribution as the number of individuals in the network increases.  
Monte Carlo experiments demonstrate that the estimator performs well in finite samples and in networks with different levels of sparsity.	
\begin{description}
	\item \emph{Keywords:} Network formation, Unobserved heterogeneity, Semiparametrics, Special regressor, Inverse weighting. 
\end{description}
\end{abstract}

\newpage
\section{Introduction}\label{S:Introduction}

People tend to connect with individuals with whom they share similar observed attributes. 
This observation is known as homophily and it is one of the main objects of study in the literature of social networks (\citealt{mcpherson2001birds}). 
However, few have investigated the role of homophily when individuals have preferences for unobserved attributes. 
Proper policy evaluation requires us to distinguish between the contributions of observed and unobserved attributes, since they have different policy implications. 
For example, students might form friendships based on their similarities on observed socioeconomic attributes as well as on their preferences for high levels of unobserved ability. 
While socioeconomic attributes can be influenced by a given policy intervention, preferences for ability are harder to change via targeted policies.
In this paper, I study the identification and estimation of the preference parameters associated with the observed attributes in a model of network formation that accounts for valuations on unobserved agent-specific factors. 
The identification and estimation strategies that I develop do not depend on distributional assumptions of the unobserved random components. 

In particular, I consider a semiparametric model of network formation with unobserved agent-specific heterogeneity. 
Specifically, two distinct agents $i$ and $j$ form an undirected link according to the following network formation equation:\footnote{A link between two agents is undirected if the connection is reciprocal. In other words, two agents are either connected or they are not. It excludes the case where one agent is related to another without the second being related to the first.}
\begin{equation}
	\label{eq:NFE}
	D_{ij} = \bm{1}\left[ g_0(Z_{i}, Z_{j})' \beta_0 + A_{i} + A_{j}-U_{ij} \geq 0 \right], 
\end{equation}
where $\mathbf{1} \left[ \cdot \right]$ is the indicator function, $D_{ij}$ is a binary outcome variable that takes a value equal to $1$ if agents $i$ and $j$ form a link and $0$ otherwise, 
$Z_{i}$ is a  vector of individual-specific and observed attributes, 
$g_0$ is a measurable function that is assumed to be known, nonlinear, finite, and symmetric on its arguments, 
$\beta_{0}$ is a vector of unknown parameters, 
$A_{i}$ and $A_{j}$ are unobserved and agent-specific random variables, 
and $U_{ij}$ is an unobserved and link-specific disturbance term.

Intuitively, equation  \eqref{eq:NFE} says that an undirected link between two agents is formed if the net benefit of the link between agents $i$ and $j$ is nonnegative. 
The components in equation  \eqref{eq:NFE} can be classified into three different categories. 
The first class, given by the vector of exogenous attributes $g_0(Z_{i}, Z_{j})$, captures the agents' preferences for establishing a link based on observed characteristics. 
For instance, this component is known as homophily on observed attributes when it captures preferences for sharing similar traits.
The second class, formed by the agent-specific and unobserved factors $A_{i}$ and $A_{j}$, captures the individual preferences for establishing connections based on agent-specific unobserved traits.
Finally, the third class, given by a link-specific disturbance term $U_{ij}$, captures the exogenous factors that influence the decision to form a specific link. 
The components in the last two categories are known to the agents but unobserved to the researcher.

The agent-specific factors in equation \eqref{eq:NFE} allow for unobserved heterogeneity across the individuals' decisions; 
this property enables the model to predict network structures with individual connections that are heterogeneous.
Moreover, under an unrestricted distribution of the unobserved agent-specific factors, 
these components could exhibit flexible dependence with the observed attributes. 

This paper offers two main contributions to the literature on network formation. 
The first contribution is to propose a new point identification strategy to identify the vector of coefficients in a semiparametric network formation model with unobserved agent-specific factors.
The point identification result is, to the best of my knowledge, the first generalization of a special regressor to analyze a network formation model (\citealt{lewbel:1998} and \citealt{lewbel:2000}).
This result depends on the existence of a special regressor and is obtained by weighting each linking decision in the network by the inverse of the conditional density of the special regressor given the observed attributes.
In section \ref{S:SpecialRegressor}, I provide sufficient conditions to point identify the vector of coefficients.
In section \ref{S:Identification:Bounds}, I provide a second point identification result that does not assume the existence of a special regressor. 
This result requires that at least one covariate has full support and consists in finding a sufficient statistic for the unobserved heterogeneity in equation \eqref{eq:NFE} at the tails of the distribution of the observed covariate with full support.

As a second contribution, I use the point identification result in section \ref{S:SpecialRegressor} to introduce a two-step semiparametric estimator of the vector of coefficients with a first-stage kernel estimator.
As an appealing property, this estimator has a closed form and is computationally tractable.
In section \ref{S:Inference: Special Regressor}, I provide sufficient conditions to show that the estimator is consistent, and it has a limiting normal distribution.
I perform inference in a setting where only one network with a large number of agents is observed in the data. 
Furthermore, I propose an adaptive inference approach to adjust for varying rates of convergence due to different levels of sparsity in the network (see, e.g., \citealt{andrews/schafgans:1998} and \citealt{khan/tamer:2010}).

In the rest of this section, I relate my results to the existing literature.

This paper is most closely related to the literature that studies dyadic network formation models with unobserved heterogeneity, (see, e.g., \citealt{graham:2017}, and \citealt{graham:2019a, graham:2019b} for additional surveys).
Within this literature, the studies by \cite{charbonneau:2017,jochmans:2017, jochmans:2018, dzemski:2019}, and \cite{yan/jiang/fienberg/leng:2019}
have analyzed the formation of a directed network.\footnote{\citet{charbonneau:2017} and \citet{jochmans:2017} study a two-way gravity model, which can be rationalized as a bipartite network with directed links.} 
Their methodologies differ substantially from the one proposed here since they follow a parametric conditional maximum likelihood approach to estimate the vector of coefficients $\beta_0$.
In contrast, I study the formation of an undirected network and follow a semiparametric approach.

This paper builds on the seminal work by \citet{graham:2017}, which aims to detect preferences for homophily in an undirected network model with agent heterogeneity. 
\citet[p. 1040]{graham:2017} introduces a Tetrad Logit Estimator with identification and asymptotic properties that depend on the link-specific disturbance terms following a logistic distribution. 
The point identification and estimation results presented below relax this requirement and can be applied to models where the distribution of $U_{ij}$ is not parametrically specified.

Since the initial draft of this paper was circulated, 
recent studies have appeared analyzing semiparametric or nonparametric variations of a dyadic network formation model with unobserved heterogeneity; 
these include papers by \citet{toth:2017,gao:2020}, and \citet{zeleneev:2020}.

Similarly to this paper, \cite{toth:2017} studies a dyadic network formation model in which the distribution of $U_{ij}$ is unknown.
However, the author uses a different identification strategy.
In particular, his strategy relies on assuming that each component in the vector of observed attributes $Z_i$ is continuously distributed which is then used to propose an identification strategy similar to the maximum rank by \citet{han:1987}.
An estimator for $\beta_0$ is then defined as the maximizer of a U process of order 4, with a nonparametric first-step estimator.\footnote{\cite{toth:2017} also proposes a variation of his estimation strategy which requires maximizing a U-process of order 2, with a nonparametric first-step estimator. This moditication improves the computational tractability of his method.} 

\citet{gao:2020} studies the identification of a dyadic network model with a nonparametric functional form for the preferences on homophily and an unknown cumulative distribution for $U_{ij}$.\footnote{ \citet{gao:2020} also provides several interesting extensions on the functional form of the unobserved heterogeneity; for reference, see \citet[p. 5]{gao:2020} and \citet[p. 6]{zeleneev:2020}. Those extensions are beyond the scope of this paper and left for future research.}  
He identifies the nonparametric homophily function by introducing a novel identification strategy that imposes an interquartile-range normalization and a location normalization of one of the quantiles as stochastic restrictions on the distribution of $U_{ij}$.

Finally, \citet{zeleneev:2020} studies the identification and estimation of a dyadic network formation model with a nonparametric structure of the unobserved heterogeneity.
This framework allows him to account for latent homophily on the unobserved attributes.
The author's identification analysis is based on introducing a pseudo-distance between a pair of agents $i$ and $j$, 
which allows him to recover groups of agents with the same levels of agent-specific unobserved heterogeneity. 
After conditioning on the matched agents with similar unobserved heterogeneity, the identification of the vector of coefficients proceeds from a pairwise difference strategy.
The estimation procedure follows the same logic as the identification strategy.

Contrary to previous studies, the identification strategy proposed here is based on the existence of a special regressor (see, e.g., \citet{lewbel:1998} and \citet{lewbel:2012} for a survey).
This paper, to the best of my knowledge, represents the first effort in the econometric literature to introduce a special regressor to analyze a network formation model. 
The vector of parameters $\beta_0$ is point identified after introducing a transformation that consists in weighting the linking decisions $D_{ij}$ by the inverse of the conditional density of the special regressor given the observed attributes. 
This transformation utilizes features of the distributions of observables and does not represent a stochastic restriction on the distribution of $U_{ij}$.
Therefore it is not nested in any existing work. As a restriction on the distribution of $U_{ij}$, I normalize to zero the conditional mean of the link-specific disturbance terms given the observed attributes.\footnote{In further research I will explore the informational content of the special regressor in a network formation model given a quantile or median restriction.}
In Section \ref{S:SpecialRegressor}, I provide a detailed discussion on the sufficient conditions needed to point identify $\beta_0$ via the existence of a special regressor.
 
The second point identification result introduced in section \ref{S:Identification:Bounds} is based on a sufficient statistic argument at the tails of the distribution of a covariate with full support.
The identification strategy shows that within- and across-individuals variation in the linking decisions can be used as a sufficient statistic to differentiate out the unobserved agent-specific factors in some sets of sufficient variations of the covariate with full support.
The existence of only one continuous attribute with large support in $Z_i$ is sufficient to show this result.
The latter assumption is satisfied by many real network datasets, and hence it is empirically relevant.\footnote{For example, in the  National Longitudinal Study of Adolescent to Adult Health (Add Health) dataset, household income is a continuous variable that can be demeaned and standardized to satisfy the support condition.}
The resulting semiparametric estimator is solved in one step, and it is defined as the maximizer of a U-process of order 4 with a trimming sequence.

In Section \ref{S:Inference: Special Regressor}, I introduce a two-step semiparametric estimator for $\beta_0$ based on the identification result that requires the existence of a special regressor.
The estimator has an analytic form similar to the least-squares, and it uses a first-step kernel estimator to weight the linking decisions $D_{ij}$ 
by the inverse of the conditional density of the special regressor.
In a recent paper, \citet{graham/niu/powell:2019} have studied the nonparametric estimation of density functions with dyadic data. 
I follow their findings to perform the first-step kernel estimation.
In theorems \ref{Theo:Inf:Consistency} and \ref{Theo:Inf:AN}, I show that the semiparametric estimator for $\beta_0$ is consistent and has limiting normal distribution.

Finally, the network formation model that I analyze is related to the literature on empirical games.
Specifically, the model in equation \eqref{eq:NFE} can be derived as a stable outcome in a static game. 
Papers that study the strategic formation of a network as a static game include 
\citet{goldsmith-pinkham/imbens:2013,leung:2015a,leung:2015,menzel:2015,miyauchi:2016, boucher/mourifie:2017, dePaula/Tamer:2017, mele2017structural, candelaria/ura:2018, sheng:2018, gualdani:2020}, and 
\citet{ridder/sheng:2020}.
The authors study network formation models that account for network externalities. 
Network externalities generate interdependencies in the linking decisions that depend on the structure of the network. 
The identification and estimation methods used in these papers differ substantially from the ones proposed here as they restrict the presence and distribution of the unobserved agent-specific heterogeneity.

The rest of the paper is organized as follows. 
Section \ref{S:Model} introduces the network formation model and motivates it as a stable outcome of a random utility model with transferable utilities.
Section \ref{S:Identification} provides the main identification results of the paper.
Section \ref{S:Inference: Special Regressor} introduces the semiparametric estimator and proves the main asymptotic results.
Section \ref{S:Simulations} reports simulation evidence and 
section \ref{S:Conclusion} concludes.
The appendix collects the proofs of various lemmas and theorems.


\section{Network formation model}\label{S:Model}

A network is an ordered pair $(\mathcal{N}_{n}, \bm{D}_{n})$ formed by a set of $n$ agents denoted by $\mathcal{N}_{n}= \left\{1, \cdots, n \right\}$ and an $n\times n$ adjacency matrix $\bm{D}_{n}$, which represents the links between the agents in $\mathcal{N}_{n}$. 
Let $D_{ij}$ denote the $(i,j)$th entry of the matrix $\bm{D}_{n}$. 
I assume the network is undirected and unweighted. A network is undirected if the adjacency matrix is symmetric, i.e., $D_{ij} = D_{ji}$. 
A network is unweighted if any $(i,j)$th entry of the adjacency matrix takes one of two values, where the values are normalized to be 0 and 1. 
In other words, $D_{ij} \in \left\{0,1\right\}$, where $D_{ij} =1$ if the agents $i$ and $j$ share a link and $D_{ij} =0$ otherwise. 
Furthermore, I normalize the value of self-ties to zero, that is, $D_{ii}=0$ for any agent $i$. 

\begin{example}[Friendships network]
 A network of best friends is an example of an undirected and unweighted network. 
 Two agents are considered to be best friends if and only if both agents report each other as friends. 
 In this case, $D_{ij} = D_{ji}=1$. Also, this example rules out the scenario of an agent reporting herself as her best friend.  
\end{example}

Each agent $i \in \mathcal{N}_{n}$ is endowed with a $K+1$-dimensional vector of observed attributes $Z_{i}$ and an unobserved scalar component term $A_{i}$. 
Common examples of observed attributes that could explain the formation of a friendships network among high school students are age,  gender, ethnicity, religion, and the students' interest in extracurricular activities.
The component $A_{i}$ captures individual $i$'s preferences for establishing a link based on unobserved and agent-specific attributes. 
The unobserved component $U_{ij}$ captures exogenous stochastic factors that influence the pair-specific decision to establish a link between agents $i$ and $j$.

Given the vectors of observed attributes $Z_{i}$ and $Z_{j}$ for $i\neq j$, let  $\bar{Z}_{ij}=g_0(Z_{i}, Z_{j})$ be a $K+1$-dimensional vector of pair-specific attributes. 
The function $g_0$ is assumed to be a known measurable function that is nonlinear and finite.\footnote{The intuition behind the requirement that $g_0$ is a nonlinear function is similar to the logic for the identification of the vector of coefficients in a linear panel data model with fixed effects. A specific feature of those models is that only the coefficients associated with time-varying variables are identified. The identification strategies proposed in section \ref{S:Identification} use the pairwise variation in $\bar{Z}_{ij}$ to identify $\beta_0$. The assumption that $g_0$ is nonlinear rules out the case that the pairwise variation is equal to the vector of zeroes, and hence, $\beta_0$ is not identified.} 
Given the undirected nature of the network, $g_0$ is assumed to be symmetric on its terms. 
The specification of $g_0$ varies according to the empirical application and is chosen by the researcher to capture homophily or heterophily effects. 
For example, suppose that $Z_{i}$ is a scalar random variable that represents agent $i$'s gender, 
then $\bar{Z}_{ij}$ could be defined as $\mathbf{1} \left[ Z_{i} = Z_{j}\right]$ to capture the preferences for homophily. 
Under this specification, $\bar{Z}_{ij}$ equals $1$ if agents $i$ and $j$ share the same gender and $0$ otherwise.

The network formation model described in equation \eqref{eq:NFE} can be obtained as a stable outcome of a random utility model with transferable utilities. 
In particular, let $\bar{u}_{ij}(\bar{Z}_{ij}, A_{j}, U_{ij})$ denote individual $i$'s latent valuation of establishing a link with $j$ given 
their shared-observed attributes $\bar{Z}_{ij}$, agent $j's$ unobserved type $A_j$, and their common unobserved factor $U_{ij}$.  
It follows that the joint net benefit of adding the link $\{i,j\}$ to the network $\bm{D}_{n}$ is
\begin{equation}
\label{eq:transferableutility}
\bar{u}_{ij}(\bar{Z}_{ij}, A_{j}, U_{ij}) + \bar{u}_{ji}(\bar{Z}_{ij}, A_{i}, U_{ij}) = \bar{Z}_{ij}'\beta_{0} + A_{i} + A_{j} - U_{ij}.
\end{equation}

Notice that the joint net benefit accounts for the preferences based on the observed attributes $\bar{Z}_{ij}'\beta_{0}$, as well as 
preferences for association based on agent-specific factors $A_{i} + A_{j}$, 
and for exogenous factors affecting the decision to establish a link $U_{ij}$.

Equation \eqref{eq:transferableutility} implies that two distinct individuals $i$ and $j$ in $\mathcal{N}_n$ 
only have valuations for their own observed attributes and agent-specific factors. 
To clarify, in the link formation decision for dyad $\{i,j\}$, 
the individuals do not take into account either observed and unobserved attributes of other individuals in the network, 
or general features of the network other than the dyad $\{i,j\}$.
These effects are known as network externalities 
(see, e.g., \citealt{chandrasekhar2014tractable, leung:2015,mele2017structural, menzel:2015,badev:2018, sheng:2018, ridder/sheng:2020}).
Some examples of these effects are preferences for reciprocity, transitive triads, or high network degree.
I leave this extension for future research. 

Next, I introduce the definition of stability.
\begin{definition}[Stability]
A network $\bm{D}_{n}$ is stable with transfers if for any distinct $i,j \in \mathcal{N}_n$:
\begin{enumerate}
\vspace{-1em}
\item for all $D_{ij} =1$, $ \bar{u}_{ij}(\bar{Z}_{ij}, A_{j}, U_{ij}) + \bar{u}_{ji}(\bar{Z}_{ij}, A_{i}, U_{ij}) \geq 0 $;
\item for all $D_{ij} =0$, $ \bar{u}_{ij}(\bar{Z}_{ij}, A_{j}, U_{ij}) + \bar{u}_{ji}(\bar{Z}_{ij}, A_{i}, U_{ij}) <  0 $.
\end{enumerate}
\end{definition}

Notice that this definition adapts the pairwise stability in \citet{jackson1996strategic} to allow for transferable utilities. 
Intuitively, this condition states that a link within dyad $\{i,j\}$ is established if the net benefit of that connection is nonnegative.   
For a generalization to nontransferable utilities, see \citet{gao/li/Sheng:2020}.

\subsection{Notation}

The following notation will be maintained  in the rest of the paper.
I will assume that the vector of observed covariates 
$Z_{i}=(v_{i}, X_{i}')'$
is comprised of a scalar random variable 
$v_i \in \mathbb{R}$
and a $K$-dimensional random vector 
$X_{i} \in \mathbb{R}^{K}$. 
Similarly, let
\begin{eqnarray*}
	\bar{Z}_{ij} &=& \left( g_0(v_i, v_j) , g_0(X_i, X_j)' \right)' 
	= (v_{ij}, W_{ij}')'	
\end{eqnarray*}
denote  the observed covariates at dyad level, and let $\beta_0 = (1, \theta_0')'$.

I will denote the distinct profiles of observed attributes for all the agents in the network as
$\bm{Z}_{n}=\{Z_{i}: i\in \mathcal{N}_{n}\}$, 
$\bm{v}_{n}=\{v_{i}: i\in \mathcal{N}_{n}\}$, and $\bm{X}_n = \left\{ X_{i}: i\in \mathcal{N}_n\right\}$ .
Similarly, let 
$\bm{A}_n= \left\{ A_{i}: i\in \mathcal{N}_n\right\}$ denote the profile of unobserved attributes.
Moreover, 
let $\bm{Z}_{-ij} = \{Z_{k}: k\neq i,j\}$, and $\bm{A}_{-ij} = \{A_{k}: k\neq i,j\}$
denote the collection of observed and unobserved attributes for all agents in the network other than agents $i$ and $j$.

The identification and estimation strategies introduced in sections \ref{S:Identification} and \ref{S:Inference: Special Regressor} use 
the information contained in subnetworks formed by groups of four distinct agents  $\{i_1,i_2,j_1,j_2\}$, also known as tetrads.
The following notation is used to describe attributes at the tetrad level.
Given a network of size $n$, there is a total of 
\[
  m_{n} = 4! \binom{n}{4}	
\]
ordered tetrads with distinct indices $i_1,i_2,j_1,j_2\in \mathcal{N}_{n}$.  
Let $\sigma$ be a function that maps these tetrads to the index set 
$\mathcal{N}_{m_{n}} = \left\{ 1, \cdots, m_{n}\right\}$. 
Thus, each tetrad with distinct indices $\left\{ i_1,i_2,j_1,j_2 \right\}$ corresponds to a unique $\sigma\left( \left\{ i_1,i_2, j_1,j_2\right\}\right) \in  \mathcal{N}_{m_{n}}$.

Given any $\sigma(\{i_1,i_2,  j_1,j_2\}) \in \mathcal{N}_{m_n}$, let
$ v_{\sigma} = \left\{ v_{i_1}, v_{j_1}, v_{i_2}, v_{j_2}\right\}$,
$ X_{\sigma} = \left\{ X_{i_1}, X_{j_1}, X_{i_2}, X_{j_2}\right\}$, 
and 
$ A_{\sigma}= \left\{ A_{i_1}, A_{j_1}, A_{i_2}, A_{j_2}\right\}$.

Moreover, define the pairwise variations across observed attributes and linking decisions as follows
\begin{eqnarray*}
	\tilde{v}_{\sigma} &=& \tilde{v}_{i_1i_2,j_1j_2} = (v_{i_1j_1}-v_{i_1j_2})-(v_{i_2j_1}-v_{i_2j_2}) \\
	\tilde{W}_{\sigma} &=& \tilde{W}_{i_1i_2,j_1j_2} = (W_{i_1j_1}-W_{i_1j_2})-(W_{i_2j_1}-W_{i_2j_2}) \\
	\tilde{D}_{\sigma} &=& \tilde{D}_{i_1i_2,j_1j_2} =  (D_{i_1j_1}-D_{i_1j_2})-(D_{i_2j_1}-D_{i_2j_2}). 
\end{eqnarray*}

Finally, given any fixed tetrad $\sigma(\{i_1,i_2,  j_1,j_2\}) \in \mathcal{N}_{m_n}$, let 
$\omega_{l_1l_2}= \left( v_{l_1l_2},X_{l_1}, X_{l_2}, A_{l_1}, A_{l_2} \right)$ 
denote the profile of attributes at dyad-level 
and 
$p_n(\omega_{l_1l_2}) = P\left[ D_{l_1l_2} =1 \mid \omega_{l_1l_2} \right]$
denote the probability that a link is created for any dyad $(l_1,l_2) \in \left\{ (i_1, j_1), (i_1, j_2), (i_2, j_1), (i_2,j_2) \right\}$.

\section{Identification}\label{S:Identification}
This section introduces the main identification results for the semiparametric network formation model with unobserved agent-specific factors.
In particular, section \ref{S:SpecialRegressor} presents the main point identification result when a special regressor is available.
Section \ref{S:Identification:Bounds} introduces a second point identification result when a covariate with full support is available.

\subsection{Point Identification Result: Special Regressor}\label{S:SpecialRegressor}

Using the notation introduced in section \ref{S:Model}, the rest of the paper considers the following representation for the network
formation model specified by equation \eqref{eq:NFE}.
In particular, agents $i$ and $j$ in $\mathcal{N}_n$ with $i\neq j$ will form an undirected link according to the following equation
\begin{equation}\label{eq:NFESpecReg}
D_{ij} 
= 
\bm{1}
\left[ 
	v_{ij}+ W_{ij}' \theta_0 + A_{i} + A_{j}-U_{ij} \geq 0 
\right], 
\end{equation}
where the coefficient associated with $v_{ij}$ has been normalized to 1 and $\theta_0$ is a $K$-dimensional vector of coefficients.
Given that the network of interest is undirected, $U_{ij}$ is assumed to be symmetric, i.e., $U_{ij} = U_{ji}$.
The vector $\theta_0$ represents the main parameter of interest.

Assumptions \ref{Ass:SR00:iidsampling}-\ref{Ass:SR04:identif} 
will 
specify the underlying structure for the network formation model in equation \eqref{eq:NFESpecReg}, which  
will be used to show the main identification result for $\theta_0$.

\begin{assumption}\label{Ass:SR00:iidsampling}
The random sequence $\{ Z_{i}, A_{i}\}_{i=1}^{n}$ is independent and identically distributed. 
\end{assumption}

Assumption \ref{Ass:SR00:iidsampling} describes the sampling process, and it is widely used to describe network data 
(see, e.g., \citealt{graham:2017, jochmans:2018}, and \citealt{auerbach:2019}).

\begin{assumption}\label{Ass:SR01:distr} For any finite $n$, the following holds.
	\begin{enumerate}
		\vspace{-1em}
		\item The sequence $ \{ U_{ij} \mid \bm{Z}_{n}, \bm{A}_{n}\}_{i\neq j}$ is conditionally independent and identically distributed for any dyad $\{i,j\}$.
		Moreover, $U_{ij} = U_{ji}$ for any dyad $\{i,j\}$.
		\item For any dyad $\{i,j\}$, $U_{ij} \mid \bm{Z}_{n}, \bm{A}_{n} \overset{d}{=} U_{ij}\mid Z_{i},Z_{j}, A_{i}, A_{j}$.
	\end{enumerate}
\end{assumption}

Assumption \ref{Ass:SR01:distr}.1  states that conditional on 
$(\bm{Z}_{n}, \bm{A}_{n})$
the link-specific disturbance terms $\{U_{ij}\}_{i\neq j}$ are independent across dyads $\{i,j\}$ and drawn from the same distribution.
Furthermore, Assumption \ref{Ass:SR01:distr}.2 requires that
conditional on $(Z_{i},Z_{j}, A_{i}, A_{j})$, the link-specific disturbance terms $U_{ij}$ are independent 
of any observed or unobserved feature in  $(\bm{Z}_{-ij},\bm{A}_{-ij})$.
Assumption \ref{Ass:SR01:distr} ensures that each of the linking decisions in the network is conditionally independent. 
In other words, it rules out interdependence across linking decisions due to externalities across the network. 

Notice that Assumption \ref{Ass:SR01:distr} allows for heteroskedasticity of a general form in the distribution of $U_{ij}$.
Moreover, it allows for flexible dependence between the unobserved agent-specific factors and the observed attributes. 
In other words, Assumption \ref{Ass:SR01:distr} does not restrict the joint distribution $(\bm{Z}_{n}, \bm{A}_{n})$.
Assumption \ref{Ass:SR01:distr} is commonly used in semiparametric nonlinear panel data models, for example in \citet{arellano2001panel}.
In network formation models, full stochastic independence $U_{ij}\perp \bm{Z}_{n}, \bm{A}_{n}$ is usually imposed (see, e.g., \citealt{leung:2015, menzel:2015, graham:2017, toth:2017}, and \citealt{gao:2020}).
Arbitrary heteroskedasticity is also considered in \citet{zeleneev:2020}.

\begin{assumption}\label{Ass:SR02:indep}
	Given $n$ and any distinct $i,j \in \mathcal{N}_n$, 
	let $e_{ij} = A_{i} + A_{j}-U_{ij}$ and suppose that $e_{ij}$ is conditionally independent of $v_{ij}$ given $(X_{i}, X_{j})$. 
	Let $F_{e\mid x}\left(e_{ij}\mid X_{i}, X_{j} \right)$ denote the conditional distribution of $e_{ij}$ given $(X_{i}, X_{j})$, 
	with support given by $\mathbb{S}_{e}(X_{i}, X_{j})$ and finite first moment. 
\end{assumption}
Assumption  \ref{Ass:SR02:indep} represents an exclusion restriction, and it entails that the regressor $v_{ij}$ is conditionally independent of $e_{ij}$ given the observed attributes $(X_{i}, X_{j})$.\footnote{The conditional independence property needs to hold after conditioning on the observed attributes $(X_{i}, X_{j})$, and not just the dyad-specific covariates $W_{ij}$. 
The intuition behind this insight follows from Assumption \ref{Ass:SR00:iidsampling}, 
which allows for unrestricted dependence between $X_i$, and $A_i$. 
In particular, the proof of Theorem \ref{Theo:IdSR} requires that any stochastic variation left in $A_i+A_j$ 
after conditioning on $(X_{i}, X_{j})$, is independent of $W_{kl}$ for any $k,l \in \mathcal{N}_n$, including, for example  $W_{il}$.
This property no longer holds if the conditioning variable used is $W_{ij}$ since it is only a feature of $(X_i, X_j)$.
}
In other words, $v_{ij}$ is a special regressor in the sense of \citet{lewbel:1998}, \citet{lewbel:2000}, and \citet{lewbel:2012}.

\begin{assumption}\label{Ass:SR03:fullsupport}
	Given $n$ and any distinct $i,j \in \mathcal{N}_n$, 
	the conditional distribution of $v_{ij}$ given $(X_{i}, X_{j})$ is absolutely continuous with respect to the Lebesgue measure with conditional density $f_{v\mid x}(v_{ij}\mid X_{i}, X_{j} )$
	and support given by $\mathbb{S}_{v}(X_{i}, X_{j}) =[\underline{s}_v, \overline{s}_v]$ for some constants $\underline{s}_v$ and $\overline{s}_v$, with $-\infty \leq \underline{s}_v< 0<\overline{s}_v \leq \infty$.
	For any $(X_i, X_j)$, the support of $-W_{ij}' \theta_0-e_{ij}$ is a subset of the interval $[\underline{s}_v,\overline{s}_v]$.
\end{assumption}
Assumption  \ref{Ass:SR03:fullsupport} is a support condition, and it ensures that $v_{ij}\mid X_{i}, X_{j}$ has a positive density function 
$f_{v\mid x}(v_{ij}\mid X_{i}, X_{j} )$ on $\mathbb{S}_{v}(X_{i}, X_{j})$. 
Furthermore, it requires that for any $(X_i, X_j)$ the support of $(-W_{ij}' \theta_0-e_{ij})$ is contained in $\mathbb{S}_{v}(X_{i}, X_{j})$. 
Notice that Assumption \ref{Ass:SR03:fullsupport} does not restrict $v_{ij}\mid X_i, X_j$ to having full support on the real line.
Hence the point identification result introduced in this section is general enough to include both cases: 
(i) the full support case, and (ii) the existence of a continuous covariate with bounded support that contains $supp\left(-W_{ij}' \theta_0-e_{ij} \mid X_i, X_j\right)$.
Moreover, observe that Assumption \ref{Ass:SR03:fullsupport} leaves unrestricted the distribution of the observed attributes $(X_i,X_j)$. 
Hence, this identification strategy also allows for discrete covariates in $W_{ij}$.

\begin{assumption}\label{Ass:SR04:identif}
Given $n$ and any tetrad $\sigma \in \mathcal{N}_{m_n}$, $\mathbb{E} \left[U_{ij}\mid X_{i}, X_{j} \right]=0$, and 
\[
\Gamma_0=\mathbb{E} \left[ \tilde{W}_{\sigma}\tilde{W}_{\sigma}' \right]
\]
is a finite and nonsingular matrix. 
\end{assumption}
The first part of assumption \ref{Ass:SR04:identif} represents a stochastic restriction on the link-specific disturbance term.
In particular, it requires that $U_{ij}\mid X_{i}, X_{j}$ has conditionally mean zero. 
The second part of assumption \ref{Ass:SR04:identif} is the standard full rank condition on the pairwise variation of the observed attributes $\tilde{W}_{\sigma}$, 
and it ensures that $\theta_0$ is point identified.

The network formation model specified by equation \eqref{eq:NFESpecReg} and Assumptions 
\ref{Ass:SR00:iidsampling}-\ref{Ass:SR04:identif} represents, to the best of my knowledge, the first generalization
of the special regressor to analyze network data. 
Following \citet{lewbel:1998, lewbel:2000}, \citet{honore/lewbel:2002}, and \citet{chen/khan/tang:2019}, 
let $D_{ij}^\ast$ be defined as
\begin{eqnarray}\label{eq:Dstar}
	D_{ij}^\ast 
	&=& 
	\left[
		\frac{
			D_{ij}- \bm{1}\left[ v_{ij} > 0 \right]
		}{
			f_{v\mid x}(v_{ij}\mid X_{i}, X_{j} )
		}
	\right]	
\end{eqnarray}
for any distinct $i, j \in \mathcal{N}_n$.

The following theorem and appended corollary formalize the first point identification result for $\theta_0$.

\begin{theorem}\label{Theo:IdSR}
If Assumptions \ref{Ass:SR02:indep}-\ref{Ass:SR04:identif}  hold in equation \eqref{eq:NFESpecReg}, then for any distinct $i$ and $j$ in $\mathcal{N}_n$
\begin{eqnarray*}
	\mathbb{E}[D_{ij}^{\ast}\mid X_{i}, X_{j}] 
	&=& 
	W_{ij}' \theta_0 
	+ 
	\mathbb{E}[A_{i} + A_{j}\mid X_{i}, X_{j} ].	
\end{eqnarray*}
\end{theorem}
\begin{proof}
See Appendix \ref{Appx:ProofsIdentification}.
\end{proof}

\begin{corollary}\label{cor:id1}
If Assumptions \ref{Ass:SR00:iidsampling}-\ref{Ass:SR04:identif} hold in equation \eqref{eq:NFESpecReg}, then for any tetrad $\sigma \in \mathcal{N}_{m_n}$
\begin{eqnarray}\label{eq:IDcondition}
	\mathbb{E} \left[  \tilde{W}_{\sigma} \tilde{D}^\ast_{\sigma} \right]			
	&=& 
	\mathbb{E} \left[ \tilde{W}_{\sigma} \tilde{W}_{\sigma}'\right] \theta_0,
\end{eqnarray}
and hence, 
\begin{eqnarray}\label{eq:beta}
	\theta_{0} 
	&=& 
	\Gamma_0^{-1} \times \Psi_{0}
\end{eqnarray}
with $\Psi_{0}= \mathbb{E} \left[  \tilde{W}_{\sigma} \tilde{D}^\ast_{\sigma} \right].$
\end{corollary}
\begin{proof}
See Appendix \ref{Appx:ProofsIdentification}.
\end{proof}

Theorem \ref{Theo:IdSR} and Corollary \ref{cor:id1} demonstrate that 
$\theta_0$ 
is point identified using the information contained in the joint distribution of $\{\tilde{D}_{\sigma}^\ast, \tilde{W}_{\sigma}\}$ at tetrad level, 
and with analytic expression given by equation \eqref{eq:beta}.
This result shows that $\theta_0$ is identified as an average of the linking decisions $\tilde{D}_{\sigma}$ which are weighted by the inverse of the conditional density of the special regressor given the observed attributes, 
$f_{v\mid x}(v_{ij} \mid X_i, X_j)$.
The result in Corollary \ref{cor:id1} will be used as a foundation of the semiparametric estimator introduced in Section \ref{S:Inference: Special Regressor}.

Given the results in Theorem \ref{Theo:IdSR} and Corollary \ref{cor:id1}  
the average contribution of the unobserved agent-specific factors to the formation of a link is also identified.
\begin{corollary}\label{cor:id2}
If Assumptions \ref{Ass:SR00:iidsampling}-\ref{Ass:SR04:identif} hold in equation \eqref{eq:NFESpecReg}, then for any $i$ and $j$ in $\mathcal{N}_n$
\begin{eqnarray}\label{eq:IDAcondition}
	\mathbb{E} \left[ A_{i} + A_{j} \right]			
	&=&
	\mathbb{E} \left[ D^\ast_{ij} \right]			
	-
	\mathbb{E} \left[ W_{ij} \right]' \theta_0,
\end{eqnarray}
\end{corollary}

\subsection{Second Point Identification Result}\label{S:Identification:Bounds} 

In this section, I provide a second point identification result for the vector of coefficients $\theta_0$.
This result does not require the regressor $v_{ij}$ to be conditionally independent of the unobserved terms, $A_{i}+A_{j}-U_{ij}$.
Nonetheless, it imposes a large support condition on $v_{ij}$ and bounds the contribution that the unobserved heterogeneity $A_{i}+A_{j}$ has on the formation of links.

The following notation will be used to state and prove this result.
For any fixed tetrad $\sigma(\{i,j,k,l\})\in \mathcal{N}_{m_n}$, 
denote the profile of observed attributes at tetrad level as 
$\bar{\bm{v}}_{\sigma}= (v_{ik}, v_{il}, v_{jk}, v_{jl})$
and 
$\bar{\bm{Z}}_{\sigma}= (\bar{\bm{v}}_{\sigma}, X_{\sigma})$.
Moreover, for any $\sigma(\{i,j,k,l\})\in \mathcal{N}_{m_n}$ and agent $r$ with $r \in \{i,j\}$ 
denote the within-individual $r$ variation of the observed attributes as
$\Delta_{\sigma} v_{r} = v_{rk}- v_{rl}$ and $\Delta_{\sigma} W_{r} = W_{rk}- W_{rl}$,  
and the within-individual $r$ variation of the unobserved attributes as 
$\Delta_{\sigma} A = A_{k}-A_{l}$. 

The following assumptions are sufficient to show the second point identification result.

\begin{assumption}\label{Ass:B02:distr}
For any finite $n$ and dyad $\{i,j\}$, Assumption \ref{Ass:SR01:distr} holds.
Furthermore,  the link-specific unobserved term 
$U_{ij}\mid Z_{i},Z_{j}, A_{i}, A_{j}$ 
has a positive density over the real line.  
\end{assumption}

Assumption \ref{Ass:B02:distr} ensures that the disturbance term $U_{ij}$ has a large support for any value of $(Z_i, Z_j, A_i, A_j)$.
This assumption is used for simplicity to ensure that the conditional probability of forming a link is well defined for any value of $(Z_i, Z_j, A_i, A_j)$.
Notice that any model where the disturbance term $U_{ij}$ is logistically or normally distributed will satisfy this condition. 

\begin{assumption}\label{Ass:B03:compact}
	The parameter space $\Theta$ is compact.
\end{assumption}

Assumption \ref{Ass:B03:compact} is a standard assumption in the semiparametrics literature, (see, e.g., \citealt{manski1975maximum,manski1985semiparametric,newey/mcfadden:1994}, and \citealt{powell1994estimation}). 
This assumption is used to control the contribution that the variation in $W_{ij}$ has on the formation of links.

\begin{assumption}
\label{Ass:B04:SufficientVariation}
	For any finite $n$, the following holds for any  $\sigma(\{i,j,k,l\})\in \mathcal{N}_{m_n}$.
	\begin{enumerate}
		\vspace{-1em}
		\item For all $X_{\sigma}$, $\bar{\bm{v}}_{\sigma}$ is continuously distributed with a positive density over $\mathbb{R}^4$. 
		\item For all $X_{\sigma}$ and $r \in \{i,j\}$,	
		$\Delta_{\sigma} v_{r}$ is continuously distributed with a positive density over the real line, 
		and the
		$\supp\left( - \Delta_{\sigma} W_{r}'\theta_0 - \Delta_{\sigma} A \mid X_{\sigma} \right) = [\underline{s}_{\varepsilon},\overline{s}_{\varepsilon}]$
		is known with
		$-\infty < \underline{s}_{\varepsilon} < 0 < \overline{s}_{\varepsilon} < \infty$. 
	\end{enumerate}
\end{assumption}

Assumption \ref{Ass:B04:SufficientVariation} ensures that the regressor $v_{ij}$ has a large support. 
Moreover, it requires that the variation in $v_{ij}$ dominates the contribution that the remaining factors have in creating a network link.
Notice that this condition does not impose that $v_{ij}$ is conditionally independent of $A_{i}+A_{j}$ given $X_\sigma$.
Intuitively, Assumption  \ref{Ass:B04:SufficientVariation} guarantees that the information at the tails of the distribution of $\Delta_{\sigma} v_{r}$
can disentangle the contributions of the preferences for homophily and unobserved heterogeneity on the creation of network links.

\begin{assumption}\label{Ass:B05:fullrank}
For any finite $n$ and tetrad  $\sigma(\{i,j,k,l\})\in \mathcal{N}_{m_n}$, $P\left[ \tilde{W}_{\sigma}'\gamma \neq 0 \right] >0  $ for all non-zero vectors $\gamma \in \mathbb{R}^K$.
\end{assumption}

Assumption \ref{Ass:B05:fullrank} is a rull rank condition.

For any fixed $\sigma(\{i,j,k,l\})\in \mathcal{N}_{m_n}$ and given $X_{\sigma}$, let 
$\mathcal{V}(X_\sigma)$ denote the set of values for which the variations in $\Delta_\sigma v_{i}$ and $\Delta_\sigma v_{j}$
dominates the contribution of the remaining factors. That is to say:
\begin{eqnarray}
	\mathcal{V}(X_\sigma)
	&=&
	\left\{ 
		\bar{\bm{v}}_{\sigma} : 
		\Delta_{\sigma} v_{i} \leq 	\underline{s}_{\varepsilon} 
		\; \& \;
		\Delta_{\sigma} v_{j} \geq   \overline{s}_{\varepsilon},
		\quad 
		\mbox{or} 
		\quad 
		\Delta_{\sigma} v_{i} \geq \overline{s}_{\varepsilon}
		\; \& \;
		\Delta_{\sigma} v_{j} \leq   \underline{s}_{\varepsilon}
	\right\}.
\end{eqnarray}
Notice that this set can be characterized using Assumption \ref{Ass:B04:SufficientVariation}.
Also, define $\xi({\theta})$ as 
\begin{eqnarray*}
\xi({\theta})
&=&
\left\{ 
\bar{\bm{z}}_{\sigma} 
:
\bar{\bm{v}}_{\sigma} \in \mathcal{V}(X_\sigma)
\quad 
\mbox{and}
\quad 
\begin{array}{c}
		\sign
		\left\{ 
			\mathbb{E}_{\theta_0}
			\left[ 
				\tilde{D}_{\sigma}
				\mid 
				X_{\sigma},
				\bar{\bm{v}}_{\sigma} \in \mathcal{V}(X_\sigma), 
				\tilde{D}_{\sigma} \in \left\{ -2, 2\right\}
			\right]
		\right\}
		\\	
		\neq 
		\sign
		\left\{ 
		\mathbb{E}_{\theta}
		\left[ 
			\tilde{D}_{\sigma}
			\mid 
			X_{\sigma},
			\bar{\bm{v}}_{\sigma} \in \mathcal{V}(X_\sigma), 
			\tilde{D}_{\sigma} \in \left\{ -2, 2\right\}
		\right]
		\right\}
\end{array}
\right\},
\end{eqnarray*}
which characterizes the set of states for which the sign of the conditional expectation of the pairwise variations of the links
$\tilde{D}_{\sigma}$  implied by $\theta$ differs from the sign of the conditional expectation generated under $\theta_0$.
In other words, the set $\xi({\theta})$ summarizes the values of observed attributes for which $\theta$ can be identified from $\theta_0$ 
using the information contained in the conditional expectation of $\tilde{D}_{\sigma}$.
Hence, $\theta_0$ is said to be identified relative to $\theta \neq \theta_0$ if 
\begin{eqnarray*}
	P
	\left[ 
		\bar{\bm{Z}}_{\sigma}
	\in \xi(\theta)
	\right]
	> 0.
\end{eqnarray*}

The next theorem and appended corollary formalizes the second point identification result.

\begin{theorem}\label{Theo:ID2:main}
Suppose Assumptions \ref{Ass:SR00:iidsampling}, \ref{Ass:B02:distr}, \ref{Ass:B03:compact}, and \ref{Ass:B04:SufficientVariation} hold in equation \eqref{eq:NFESpecReg}. 
Let 
\begin{eqnarray*}
	Q_\theta
	&=&
	\left\{ 
		\bar{\bm{z}}_{\sigma}
		:
		\bar{\bm{v}}_{\sigma}
		\in 
		\mathcal{V}(X_{\sigma})
		\quad 
		\mbox{and}
		\quad
		\tilde{W}_{\sigma}'\theta_0
		\leq 
		- \tilde{v}_{\sigma}
		< 
		\tilde{W}_{\sigma}'\theta
		\quad 
		\mbox{or}
		\quad 
		\tilde{W}_{\sigma}'\theta
		\leq 
		- \tilde{v}_{\sigma}
		< 
		\tilde{W}_{\sigma}'\theta_0
	\right\}.	
\end{eqnarray*} 
If	
$
P
\left[ 
	\bar{\bm{Z}}_{\sigma}
	\in 
	Q_\theta	
\right]
>0
$
, $\theta_0$ is point identified relative to $\theta$.
\end{theorem}
\begin{proof}
See Appendix \ref{Appx:ProofsIdentification}.
\end{proof}

\begin{corollary}\label{Theo:ID2:corollary}
	Suppose Assumptions
	\ref{Ass:SR00:iidsampling}, \ref{Ass:B02:distr}- \ref{Ass:B05:fullrank}
	hold in equation \eqref{eq:NFESpecReg}. Then $\theta_0$ is point identified.
\end{corollary}
\begin{proof}
	See Appendix \ref{Appx:ProofsIdentification}.
\end{proof}

The results in Theorem \ref{Theo:ID2:main} and Corollary \ref{Theo:ID2:corollary}
can be used to define an estimator for $\theta_0$ as the maximizer of a $U$-process of order 4 with a trimming sequence $\gamma_n$ such that $\gamma_{n} \rightarrow \infty$ as $n \rightarrow \infty$.
In particular, the estimator of  $\theta_0$ can be defined as 
\begin{eqnarray*}
	\hat{\theta}
	&=&
	\argmax_{\theta \in \Theta}
	\hat{H}_n(\theta, \gamma_n)  
\end{eqnarray*}
where
\begin{eqnarray*}
	\hat{H}_n(\theta, \gamma_n)  
	&=&
	\left[ 
	4!	
	\binom{n}{4}
	\right]^{-1}
	\sum_{i_1=1}^n
	\sum_{j_1 \neq i_1}
	\sum_{i_2 \neq i_1, j_1}
	\sum_{j_2 \neq i_1, j_1, i_2}
	H
	\left( 
		\bar{\bm{Z}}_{\sigma(\{i_1, j_1;i_2,j_2\})}, 
		\tilde{D}_{\sigma(\{i_1, j_1;i_2,j_2\})}; 
		\theta, \gamma_n 
	\right)
	\\
	H
	\left( \bar{\bm{Z}}_{\sigma}, \tilde{D}_{\sigma}; \theta, \gamma_n \right)
	&=&
	\left[ 
	\sign
	\left\{ 
		\tilde{v}_{\sigma}
		+
		\tilde{W}_{\sigma}'\theta 
	\right\}
	\times 
	\tilde{D}_{\sigma}
	\right]
	\times 
	\bm{1}
	\left[ 
		\mid \tilde{D}_{\sigma}\mid =2
	\right]
	\times 
	\bm{1}
	\left[ 
		\mid \Delta_{\sigma} v_{i} \mid, \mid \Delta_{\sigma} v_{j} \mid \geq \gamma_{n}
	\right].
\end{eqnarray*}

Although point identification of $\theta_0$ is achieved assuming that the bounds $[\underline{s}_{\varepsilon},\overline{s}_{\varepsilon}]$ are known, 
notice that they are not needed to define the estimator $\hat{\theta}$.
In other words, it is sufficient to assume that $\Delta_\sigma v_i$ has a large support which contains $\supp \left( -\Delta_\sigma W_i'\theta_0 - \Delta_\sigma A \mid X_\sigma \right) $ to characterize the estimator for $\theta_0$. 

Naturally, the asymptotic properties of $\hat{\theta}$ will depend on the frequency of subgraph configurations that satisfy the restriction
$\bm{1}
\left[ 
	\mid \tilde{D}_{\sigma}\mid =2
\right]$ 
in the sample, and the rate at which $\gamma_n \rightarrow \infty$ as $n\rightarrow \infty$.
The rest of this paper prioritizes the study of the semiparametric estimator introduced in section \ref{S:Inference: Special Regressor} since it is computationally more tractable than $\hat{\theta}$.

\section{Inference}\label{S:Inference: Special Regressor}
In this section, I introduce a semiparametric estimator for $\theta_0$ based on the point identification result derived in section \ref{S:SpecialRegressor}.
The estimator for $\theta_0$ denoted by $\widehat{\theta}_n$ is a two-step estimator with a nonparametric estimate of the conditional distribution of $v_{ij}$ given $\{X_i, X_j\}$, i.e., $f_{v\mid x}(v_{ij} \mid X_i, X_j)$.
Section \ref{S:Consistency02} provides sufficient conditions to study the large sample properties of $\widehat{\theta}_n$.  
Theorem \ref{Theo:Inf:Consistency} proves that $\widehat{\theta}_n$ is a consistent estimator of $\theta_0$.
Theorem \ref{Theo:Inf:AN} shows that the limiting distribution of $\widehat{\theta}_n$ is normal. 

\subsection{Consistency}\label{S:Consistency02}
The estimator for $\theta_0$ is defined as the sample analog of equation \eqref{eq:beta} and is 
obtained by averaging over the linking decisions $\tilde{D}_\sigma$ for all distinct tetrads $\sigma \in  \mathcal{N}_{m_{n}}$.
Given that the inverse of $f_{v\mid x}(v_{ij}\mid X_{i}, X_{j} )$ is used as a weight in the definition of $\Psi_0$, and hence $\theta_0$, 
I introduce a trimming sequence intended to avoid boundary effects arising from the first-step estimation of $f_{v\mid x}(v_{ij}\mid X_{i}, X_{j} )$.

Recall that $\tilde{D}_{\sigma}$ is defined as the pairwise variation across the linking decisions for a given tetrad
$\sigma\left( \left\{ i_1,i_2, j_1,j_2\right\}\right) \in \mathcal{N}_{m_n}$.
I extend that notation to define as follows the pairwise variation of the trimmed network links given a trimming parameter $\tau$ 
\begin{eqnarray*}
	\widetilde{D}_{\sigma, \tau}^\ast
	&=&
	\left(  D_{i_1 j_1, \tau}^\ast - D_{i_1 j_2, \tau}^\ast \right) 
	- 
	\left(  D_{i_2 j_1, \tau}^\ast  - D_{i_2 j_2, \tau}^\ast \right)
	\\
	\widehat{D}_{\sigma, \tau}^\ast
	&=&
	\left(  \widehat{D}_{i_1 j_1,\tau}^\ast  - \widehat{D}_{i_1 j_2,\tau}^\ast \right) 
	- 
	\left( \widehat{D}_{i_2 j_1,\tau}^\ast  - \widehat{D}_{i_2 j_2,\tau}^\ast \right),
\end{eqnarray*}
where for any distinct $i_1$ and $j_1$ in $\mathcal{N}_{n}$
\begin{eqnarray*}
	D_{i_1 j_1, \tau}^\ast
	&=& 
	\left(  
	\frac{
		D_{i_1 j_1}- \bm{1}\left[ v_{i_1 j_1} > 0 \right]
	}{
		f_{v\mid x}(v_{i_1 j_1}\mid X_{i_1}, X_{j_1})
	}	
	\right)
	I_\tau(v_{i_1j_1}, X_{i_1}, X_{j_1})
	\\
	\widehat{D}_{i_1 j_1, \tau}^\ast 
	&=&
	\left(  
	\frac{
		D_{i_1j_1}- \bm{1}\left[ v_{i_1j_1} > 0 \right]
	}{
		\widehat{f}_{v\mid x}(v_{i_1 j_1}\mid X_{i_1}, X_{j_1})
	}	
	\right)
	I_\tau(v_{i_1j_1}, X_{i_1}, X_{j_1}).
	\\	
\end{eqnarray*}

In the equations above, 
$f_{v\mid x}(v_{i_1 j_1}\mid X_{i_1}, X_{j_1})$ 
denotes the true conditional density function of $v_{i_1 j_1}$ given $(X_{i_1}, X_{j_1})$, 
and 
$\widehat{f}_{v\mid x}(v_{i_1 j_1}\mid X_{i_1}, X_{j_1})$ 
denotes a kernel estimator of the conditional density of $v_{i_1 j_1}$ given $(X_{i_1}, X_{j_1})$. 
Thus,  
$\widetilde{D}_{\sigma, \tau}^\ast$ denotes the pairwise variation of the trimmed network links assuming that 
the conditional distribution of the special regressor given the observed attributes is known. 
Conversely, $\widehat{D}_{\sigma, \tau}^\ast$ denotes the pairwise variation of the trimmed network links 
when $f_{v\mid x}(v_{i_1 j_1}\mid X_{i_1}, X_{j_1})$ is replaced by a first-stage kernel estimator $\widehat{f}_{v\mid x}(v_{i_1 j_1}\mid X_{i_1}, X_{j_1})$

The trimming sequence 
$I_\tau(v_{i_{1j_1}}, X_{i_1}, X_{j_1})$ 
is a function of the observed attributes at a dyad level, and it converges to 1 as the trimming parameter $\tau\rightarrow 0$ when $n\rightarrow \infty$.
Assumptions \ref{Ass:Inf02:Trimming} and \ref{Ass:Inf04:Kernel} below describe the conditions imposed on the trimming parameter $\tau$, 
(see \citealt{honore/lewbel:2002} and \citealt{khan/tamer:2010}).

To ease the exposition, I introduce the following notation for any distinct $i_1, j_1 \in \mathcal{N}_n$
\begin{eqnarray*}
I_{\tau, i_1j_1} 		&=& I_\tau(v_{i_1j_1}, X_{i_1},X_{j_1})	\\
f_{vx,i_1 j_1} 			&=& f_{v,x}(v_{i_1 j_1}, X_{i_1}, X_{j_1}) \\
f_{x,i_1 j_1} 			&=& f_{x}(X_{i_1}, X_{j_1})  \\
\varphi_{i_1 j_1} 		&=& D_{i_1 j_1} -  \mathbf{1}\left[ v_{i_1 j_1}>0 \right]   \\
\varphi_{i_1 j_1, \tau} &=& \varphi_{i_1 j_1} I_{\tau, i_1j_1}.
\end{eqnarray*}

With this notation at hand, the semiparametric estimator for $\theta_0$ is defined as 
\begin{eqnarray}\label{eq:thetaestimator}
	\widehat{\theta}_n
	&=&
	\widehat{\Gamma}_n^{-1} \times \widehat{\Psi}_{n, \tau} 
\end{eqnarray}
where 
\begin{eqnarray*}
\widehat{\Gamma}_n
&=& 
\frac{1}{m_{n}} 
	\sum_{ \sigma \in \mathcal{N}_{m_{n}}}
	\left[
		\tilde{W}_{\sigma} \tilde{W}_{\sigma}'
	\right] 
\\
\widehat{\Psi}_{n, \tau} 
&=& 
	\frac{1}{m_{n}} 
	\sum_{ \sigma \in \mathcal{N}_{m_{n}}}
	\left[  
		\tilde{W}_{\sigma} \widehat{D}_{\sigma, \tau}^\ast
	\right]
\end{eqnarray*}
and $m_{n} = 4! \binom{n}{4}$.

The first-stage kernel estimator $\widehat{f}_{v\mid x}(v_{i_1j_1} \mid X_{i_1}, X_{j_1})$ is defined as the 
ratio of the kernel estimators $\widehat{f}_{vx,i_1 j_1}$ and $\widehat{f}_{x,i_1 j_1}$ with
\begin{eqnarray*}
\widehat{f}_{vx,i_1 j_1}
&=&
\frac{1}{(n-2)(n-3) h^{L+1}} 
\sum_{k_1\neq i_1,j_1}
\sum_{k_2\neq i_1,j_1, k_1} 
K_{vx,h} 
\left[   
	v_{k_1k_2}-v_{i_1j_1} , 
	X_{k_1}-X_{i_1} , 
	X_{k_2}-X_{j_1}
\right]
\\
\widehat{f}_{x,i_1 j_1}
&=&
\frac{1}{(n-2)(n-3)h^{L}} 
\sum_{k_1\neq i_1,j_1}
\sum_{k_2\neq i_1,j_1, k_1}
K_{x,h} 
\left[ 
	X_{k_1}-X_{i_1} , 
	X_{k_2}-X_{j_1}  
\right], 
\end{eqnarray*}
where  $h$ denotes a bandwith parameter and $L = 2K$. 
The kernels $K_{vx,h}$ and $K_{x,h}$ are defined as 
\begin{eqnarray*}
K_{vx,h} 
\left[  
v_{k_1k_2}-v_{i_1j_1}  , X_{k_1}-X_{i_1} , X_{k_2}-X_{j_1}
\right]
& = &
K_{vx} 
\left[  
	\frac{v_{k_1k_2}-v_{i_1j_1}}{h}  , \frac{X_{k_1}-X_{i_1}}{h} , \frac{X_{k_2}-X_{j_1}}{h}
\right]
\\
K_{x,h} 
\left[ 
X_{k_1}-X_{i_1} , X_{k_2}-X_{j_1} 
\right]
& = &
K_{x} 
\left[  
	\frac{X_{k_1}-X_{i_1}}{h} , \frac{X_{k_2}-X_{j_1}}{h}  
\right]. 
\end{eqnarray*}
Assumption \ref{Ass:Inf04:Kernel} below describes the conditions imposed on the kernel functions $K_{vx,h}$ and $K_{x,h}$, and bandwith parameter $h$.

The estimator defined in equation \eqref{eq:thetaestimator} represents, 
to the best of my knowledge, 
the first effort to estimate the vector of parameters $\theta_0$ 
defined in the network formation model given by equation \eqref{eq:NFESpecReg} using a two-step semiparametric estimator that utilizes the existence of a special regressor. 

A semiparametric approach is attractive because it does not restrict the distribution of the disturbance term to any specific parametric family. 
Furthermore, it allows for a flexible statistical dependence between the agent-specific unobserved factors and the observed attributes, i.e., $\{\mathbf{X}_n, \mathbf{A}_n\}$. 
As an additional appealing property, the estimator defined in equation \eqref{eq:thetaestimator} has an analytical form. 
This characteristic increases its computational tractability compared with the estimator defined as the maximizer of a U-process and introduced in section \ref{S:Identification:Bounds}.
Regarding the non-parametric first-stage estimator,
\citet[Supp. Appendix]{leung:2015} and \citet{graham/niu/powell:2019} 
have studied the properties of kernel estimators for network data. 
I use their findings to analyze the asymptotic properties of $\widehat{\theta}_n$.

The following technical conditions are needed to prove Theorems \ref{Theo:Inf:Consistency} and \ref{Theo:Inf:AN}.
For simplicity, the theorems are stated and proved assuming that all of the elements of $X_i$ are continuously distributed. 
However, the results can be readily extended to include discretely distributed variables by applying the density estimator separately to each discrete cell of data.

\begin{assumption}\label{Ass:Inf01:SamplingMoments}
	For any distinct indices $i$ and $j$ in $\mathcal{N}_n$, the dyad-level covariates $(X_i, X_j)$ and $(v_{ij}, X_i, X_j)$ are absolutely continuous with respect to some Lebesgue measures
	with Radon-Nikodym densities $f_{x,ij}$ and $f_{vx,ij}$, and supports denoted by $\mathbb{S}_{x}$ and $\mathbb{S}_{vx}$.
	Assume that $f_{x,ij}$ and $f_{vx,ij}$ are bounded, $f_{vx,ij}$ is bounded away from zero, and there 
	exists a constant $\overline{M} > L+1$ (recall that $L=2^K$, with $dim(X_i)=K$) such that $f_{x,ij}$ and $f_{vx,ij}$
	are $\overline{M}$-times differentiable with respect to all of its arguments with bounded derivatives. 
	There exist finite constants $C_{w,1}$ and $C_{w,2}$ such that $\sup_{\sigma \in \mathcal{N}_{m_n}}\mid\mid \tilde{W}_{\sigma} \mid\mid \leq C_{w,1}$ w.p.1 
	and $\mathbb{E}\left[ \mid\mid \tilde{W}_{\sigma} \mid\mid^4 \right]<C_{w,2}$.
\end{assumption}

Assumption \ref{Ass:Inf01:SamplingMoments} ensures that the densities $f_{x,ij}$ and $f_{vx,ij}$ are continuous and $M$-times differentiable. 
Also, it requires the existence of fourth-order moments for $\tilde{W}_{\sigma}$, for any $\sigma \in \mathcal{N}_{m_n}$.
This assumption has been used in the literature of semiparametric methods, 
for example in \citet{ahn/powell:1993, aradillas:2012}, and \citet{honore/lewbel:2002}.   

\begin{assumption}\label{Ass:Inf02:Trimming}
	Let $\tau$ be a density trimming parameter defined above.
	Assume that the support $\mathbb{S}_{vx}$ is known, and 
	the trimming function $I_{\tau,ij}$ is equal to zero if $(v_{ij},X_i, X_{j})$ is within a distance $\tau$ of the 
	boundary of $\mathbb{S}_{vx}$, and otherwise, $I_{\tau,ij}$ equals one. 
	Also, assume that $\tau \rightarrow 0$  and $\tau n^2 \rightarrow 0$ as $n\rightarrow \infty$.
\end{assumption}

Due to the weighting scheme used in the definition of $\widehat{D}_{i_1 j_1}^\ast $,
boundary effects could arise from the density estimation step 
when computing 
$\widehat{\Psi}_{n,\tau}$.
Assumptions \ref{Ass:Inf01:SamplingMoments} and \ref{Ass:Inf02:Trimming} deal with this technicality
by assuming that $f_{vx,i_1j_1}$ is bounded away from zero and by introducing a trimming sequence 
$I_\tau(v_{i_1j_1}, X_{i_1},X_{j_1} )$ 
that sets to zero the terms in $\widehat{\Psi}_{n,\tau}$ with data within a $\tau$ distance of the boundary of $\mathbb{S}_{vx}$, (see, e.g., \citealt{lewbel:1997,lewbel:2000,honore/lewbel:2002}, and \citealt{khan/tamer:2010})

Assumptions \ref{Ass:Inf01:SamplingMoments} and \ref{Ass:Inf02:Trimming} require that the support $\mathbb{S}_{vx}$ is known. 
The support $\mathbb{S}_{vx}$ is identified from the distribution of observables, and hence, it can be estimated in an empirical application.
As an alternative approach to Assumption \ref{Ass:Inf02:Trimming}, a fixed trimming function that is not $n$-dependent could be used instead,
(see, e.g., \citealt{aradillas/honore/powell:2007} and \citealt{aradillas:2012}).

\begin{assumption}\label{Ass:Inf03:SmoothnessDensity}
Let $\overline{M}$ be as defined above. 
Given any tetrad $\sigma(\{i_1,i_2, j_1,j_2\})\in \mathcal{N}_{m_n}$, let
\begin{eqnarray*}
	\Xi
	\left( X_{l_1}, X_{l_2} \right) 
	&=&
	E
	\left[  
	\tilde{W}_{\sigma}	
	D_{ l_1l_2, \tau}^\ast
	\mid 
	X_{l_1}, X_{l_2}
	\right]
	\\
	\Xi
	\left( v_{l_1l_2}, X_{l_1}, X_{l_2} \right) 
	&=&
	E
	\left[  
	\tilde{W}_{\sigma}	
	D_{ l_1l_2, \tau}^\ast 
	\mid 
	v_{l_1l_2}, X_{l_1}, X_{l_2}
	\right]
\end{eqnarray*}
for any dyad $(l_1, l_2) \in \{ (i_1,j_1), (i_1,j_2), (i_2,j_1), (i_2,j_2) \}$.
The expectations
$	\Xi\left( x,x \right) $	
and 
$	\Xi\left( v, x, x \right) $	
exist and are continuous in the components of $(v, x, x')$ for all 
$(v,x,x') \in \mathbb{S}_{vx}$. 
Also,
$	\Xi\left( x,x \right) $	
and 
$	\Xi\left( v, x, x \right) $	
are $\overline{M}$-times differentiable in the components of $(v, x, x')$ for all 
$(v,x,x') \in \overline{\mathbb{S}}_{vx}$, where  
$\overline{\mathbb{S}}_{vx}$ differs from $\mathbb{S}_{vx}$ by a set of measure zero.

There exist some functions $m_x(x, x)$ and $m_{vx}(v,x,x')$ such that the following 
local Lipschitz conditions hold for some $(x_0, x_0')$ and  $(v_0, x_0, x_0')$ in an open neighborhood of zero and for all $\tau>0$: 
\begin{eqnarray*}
	\mid\mid  
	f_{vx}(v +v_0, x + x_0, x' + x_0')
	-
	f_{vx}(v, x, x')
	\mid\mid 
	&\leq& 
	m_{vx}(v, x, x') 
	\mid\mid 
	(v_0, x_0, x_0')
	\mid\mid  
	\\
	\mid\mid 
	f_{x}(x + x_0, x' + x_0')
	-
	f_{x}(x, x')
	\mid\mid 
	&\leq& 
	m_x(x, x') 
	\mid\mid  
	(x_0, x_0')
	\mid\mid
	\\
	\mid\mid  
	\Xi(v +v_0, x + x_0, x' + x_0')
	-
	\Xi(v, x, x')
	\mid\mid 
	&\leq& 
	m_{vx}(v, x, x') 
	\mid\mid 
	(v_0, x_0, x_0')
	\mid\mid  
	\\
	\mid\mid 
	\Xi(x + x_0, x' + x_0')
	-
	\Xi(x, x')
	\mid\mid 
	&\leq& 
	m_x(x, x') 
	\mid\mid  
	(x_0, x_0')
	\mid\mid.
\end{eqnarray*}
\end{assumption}

Assumption \ref{Ass:Inf03:SmoothnessDensity} imposes local smoothness conditions that are needed to derive the H\'{a}jek projection of a $V$-statistic. 
Similar conditions have been used in \citet{ahn/powell:1993, aradillas:2012}, and \citet{honore/lewbel:2002}.

\begin{assumption}\label{Ass:Inf0302:Boundedness}
Given any $\sigma(\{i_1,i_2, j_1,j_2\})\in \mathcal{N}_{m_n}$ and $(l_1, l_2) \in \{ (i_1,j_1), (i_1,j_2), (i_2,j_1), (i_2,j_2) \}$,
let
$
	\chi_{l_1l_2}
	=
	\chi(X_{l_1}, X_{l_2}) 
	=
	\mathbb{E}
	\left[ 
		\tilde{W}_{\sigma}	
		\mid 
		X_{l_1}, X_{l_2}
	\right].
$

The following moments exist
	\begin{eqnarray*}
		&&
		\sup_{(x,x')\in \mathbb{S}_{x}}
		\chi(x, x') 
		\\
		&&
		\sup_{ (v,x,x')\in \mathbb{S}_{v, x}, \tau \geq 0}  
		\mathbb{E} \left[ 
			\left(
		   \frac{
			   \varphi_{l_1l_2,\tau}
			   }{
			   f_{vx}(v, x, x')
			   }	 
		   \right)^2 	
		 \mid 
		 v, x, x' \right] 
		 \\
		 &&
		 \sup_{ (v, x,x')\in \mathbb{S}_{v, x}, \tau \geq 0}  
		 \mathbb{E} \left[ 
			 \left(  
			\frac{
				D_{l_1l_2,\tau}^\ast
				}{
				f_{vx}(v, x, x')
				}	 
			\right)^2 	
		  \mid 
		  v, x, x' \right], 
	\end{eqnarray*}
	and the objects 
	\begin{eqnarray*}
		&&
		\chi(x, x') 
		\\
		& &
		\mathbb{E} \left[ 
			 \left(
			\frac{
				\varphi_{l_1l_2,\tau}
				}{
				f_{vx}(v, x, x')
				}	 
			\right)^2 	
		  \mid 
		  v, x, x' \right] 
		  \\
		  &&
		\mathbb{E} \left[ 
			\left(
		   \frac{
			   D_{l_1l_2,\tau}^\ast
			   }{
			   f_{vx}(v, x, x')
			   }	 
		   \right)^2 	
		 \mid 
		 v, x, x' \right] 
	\end{eqnarray*}	
	are continuous in the components of $(v, x, x') \in \mathbb{S}_{vx}$.
	Moreover, there exists a finite constant $C_{\chi}$, such that 
	\begin{eqnarray*}
		E
		\left[ 
			\mid \mid 
			\chi(x, x')^6 
			\mid \mid 
		\right]
		\leq C_{\chi}
	\end{eqnarray*}	
	for any $(x, x') \in \mathbb{S}_{x}$.
\end{assumption}

Assumption \ref{Ass:Inf0302:Boundedness} ensures the existence and boundedness of the conditional expectations defined above. 
These conditions are needed to invoke a uniform law of large numbers for $V$-statistics. 
The last part of Assumption \ref{Ass:Inf0302:Boundedness} guarantees the existence of sixth-order moments, 
and it will be used to invoke a conditional central limit theorem. 

\begin{assumption}\label{Ass:Inf04:Kernel}
	Let $\overline{M}$ and $\tau$ be as defined above. 
	The kernel $K_{x}(x,x'):\mathbb{R}^{L} \mapsto \mathbb{R}$ and bandwith $h$ used to define the kernel estimator $\hat{f}_{x}$ satisfy:
	\begin{enumerate}
		\item $K_{x}(x,x')=0$ for all $(x,x')$ on the boundary of, and outside of, a convex bounded subset of $\mathbb{R}^{L}$. This subset has an nonempty interior and has the origin as an interior point.
		\item $K_{x}(\cdot, \cdot)$ is symmetric around zero, bounded, differentiable, and bias-reducing of order $2\overline{M}$.
		\item  There exists $\overline{\delta}>0$ such that  $n^{1-\overline{\delta}} h^{L+1} \rightarrow \infty$, 
		$n h^{\overline{M}} \rightarrow 0$, and $h/ \tau \rightarrow 0$. 
		\end{enumerate}
	The kernel function $K_{v,x}(v,x,x')$  has all the same properties, replacing $(x,x')$ with $(v,x, x')$.
\end{assumption}
 
Assumption \ref{Ass:Inf04:Kernel} requires the use of a higher-order kernel. 
This selection is motivated to control the bias induced by using the inverse of $f_{v\mid x}(v_{i_1 j_1}\mid X_{i_1}, X_{j_1})$ as a weighting function.
This assumption has been used by \citet{honore/lewbel:2002} and \citet{leung:2015}. 
\citet{graham/niu/powell:2019} provide a comprehensive treatment of kernel estimation for undirected network data.

Using the assumptions above, it follows that $\widehat{\theta}_n$ defined in equation \eqref{eq:beta} is a consistent estimator of $\theta_0$.
Theorem \ref{Theo:Inf:Consistency} formally states this result.
\begin{theorem}\label{Theo:Inf:Consistency}
	Let Assumptions \ref{Ass:SR00:iidsampling}-\ref{Ass:SR04:identif} and \ref{Ass:Inf01:SamplingMoments}-\ref{Ass:Inf04:Kernel}  hold. 
	Then $(\widehat{\theta}_n -\theta_0) \overset{p}{\rightarrow} \mathbf{0}$ as $n\rightarrow \infty$.
\end{theorem}
\begin{proof}
See Appendix \ref{Appx:ProofsIdentification}.
\end{proof}

\subsection{Asymptotic Distribution}\label{S:asydist}
The following theorem derives the asymptotic distribution of $\hat{\theta}_n$. 
A key step in proving this result is to show that 
\begin{eqnarray*}
	\sqrt{n(n-1)}
	\Upsilon_{n}^{-1/2}
	\left\{ 
		\widehat{\Psi}_{n, \tau} 
		-
		E \left[ 
		\tilde{W}_{\sigma}
		\widetilde{D}_{\sigma, \tau}^\ast
		\mid 
		v_{\sigma}, X_{\sigma}, A_{\sigma}
		\right]
	\right\} 
	\Rightarrow
	\mathcal{N}\left( 0,I \right),
\end{eqnarray*}
where $I$ denotes the $K$-dimensional identity matrix, and $\Upsilon_{n}=n(n-1)Var\left(\widehat{\Psi}_{n, \tau}\right)$, which is defined as 
\begin{eqnarray*}
	\Upsilon_{n} 
	&=& 
	\frac{1}{n(n-1)}
	\sum_{i_1 =1}^n  \sum_{j_1 \neq i_1} 
	\mathbb{E}
	\left[ 
		\left\{ 
			\frac{
				p_n(\omega_{i_1j_1})
				\left[ 1-  p_{n}(\omega_{i_1j_1})  \right]
			}{
				f_{v\mid x, i_1j_1}
			}	
		\right\}
		I_{\tau, i_1j_1}
	\right]
		\overline{\chi}_{i_1j_1}
		\overline{\chi}_{i_1j_1}'
\end{eqnarray*}
with 
$
	\overline{\chi}_{i_1j_1}
	=
	\left\{ 
		\frac{1}{(n-2)(n-3)}
		\sum_{i_2 \neq i_1 , j_1} \sum_{j_2 \neq i_1 , j_1, i_2} 
			\mathbb{E}
			\left[ 
				\tilde{W}_{\sigma\{i_1,i_2;j_1,j_2\}}
				\mid 	
				X_{i_1}, X_{j_1}
			\right]
	\right\}.
$

The proof of this result follows from showing that 
\begin{eqnarray*}
	\left\{ 
		\widehat{\Psi}_{n, \tau} 
		-
		\mathbb{E} 
		\left[ 
		\tilde{W}_{\sigma}
		\widetilde{D}_{\sigma, \tau}^\ast
		\mid 
		v_{\sigma}, X_{\sigma}, A_{\sigma}
		\right]
	\right\} 
\end{eqnarray*}
is asymptotically equivalent to its H\'{a}jek Projection onto an arbitrary function of 
\[
	\zeta_{i_1j_1} = (v_{i_1j_1}, X_{i_1}, X_{j_1}, A_{i_1}, A_{j_1}, U_{i_1j_1}).
\]
The resulting H\'{a}jek Projection
is an average of conditionally independent random variables at a dyad level, with conditional mean equal to $0$ 
and a conditional variance that approximates $\Upsilon_{n}$ in the limit.
The result follows from a conditional version of Lyapunov’s central limit theorem (see, e.g., \citealt{rao:2009}).

The remaining information needed to derive the limiting distribution of the semiparametric estimator $\hat{\theta}_{n}$, is 
the convergence rate of  $\Upsilon_n$, which is given by 
\begin{eqnarray*}
\varrho_{n}
&=&
O
\left( 
	\Upsilon_{n} 
 \right)
=
O
\left( 
	\mathbb{E}
	\left[ 
		\left\{ 
			\frac{
				p_n(\omega_{i_1j_1})
				\left[ 1-  p_{n}(\omega_{i_1j_1})  \right]
			}{
				f_{v\mid x, i_1j_1}
			}	
		\right\}
		I_{\tau, i_1j_1}
	\right]
\right),
\end{eqnarray*}
and the following matrix 
\begin{eqnarray*}
\Sigma_n 
&=&
\Gamma_0^{-1}
\times 
\Upsilon_{n}
\times 
\Gamma_0^{-1}.
\end{eqnarray*}

The next theorem formalizes the limiting distribution of $\widehat{\theta}_n$.

\begin{theorem}\label{Theo:Inf:AN}
Suppose Assumptions \ref{Ass:SR00:iidsampling}-\ref{Ass:SR04:identif}, \ref{Ass:Inf01:SamplingMoments}-\ref{Ass:Inf04:Kernel}, and $n(n-1)\varrho_{n}^{-1} \rightarrow \infty$ hold.
It then follows that 
\begin{eqnarray}\label{eq:asylinearrep}
	\sqrt{n(n-1)}
	\Sigma_{n}^{-1/2} 
	\left(  \widehat{\theta}_{n} - \theta_0\right)
	&=&
	\Sigma_{n}^{-1/2}
	\times 
	\Gamma_{0}^{-1}
	\times 
	\left\{ 
	\frac{
		1	
	}{
		\sqrt{n(n-1)}
	}
	\sum_{i_1=1}^{n}
	\sum_{j_1=i_1}
		\xi_{i_1j_1, \tau}	
	\right\} 
	+ 
	o_p(1)
\end{eqnarray}
with 
\begin{eqnarray*}
	\xi_{i_1j_1, \tau}	
	&=&
	\left\{ 
		D^\ast_{ i_1 j_1} 
		-
		\mathbb{E}
		\left[ 
			D^\ast_{i_1j_1}
			\mid 
			\omega_{i_1j_1}	
		\right]
	\right\}
	I_{\tau, i_1j_1} 
	\overline{\chi}_{i_1j_1},
\end{eqnarray*}
and thus, 
\begin{eqnarray*}
	\sqrt{n(n-1)}
	\Sigma_{n}^{-1/2}
	\left(  \widehat{\theta}_{n} - \theta_0\right)
	&
	\Rightarrow
	&
	N \left( 0, I \right).
\end{eqnarray*}
\end{theorem}
\begin{proof}
	See Appendix \ref{Appx:ProofsIdentification}.
\end{proof}

Equation \eqref{eq:asylinearrep} describes the asymptotic linear representation of 
$\widehat{\theta}_{n}$.
The limiting distribution of $\widehat{\theta}_{n}$ is derived following a studentized approach as in 
\citet{andrews/schafgans:1998}, \citet{khan/tamer:2010}, and \citet{jochmans:2018}
to control for the possible varying rates of convergence due to sparsity of the network.
Notice that if $\varrho_{n}^{-1}$ converges to a finite constant that is bounded away from zero,  
$\widehat{\theta}_{n} - \theta_0$ converges at a parametric rate $\sqrt{n(n-1)}$, 
with effective sample given by the square root of the number of dyads.
Alternatively, if $\varrho_{n}^{-1}$ decays as $n$ increases, $\widehat{\theta}_{n} - \theta_0$ has a slower rate of convergence given 
by $O_p\left( \sqrt{n(n-1)\varrho_{n}^{-1}} \right)$.

\section{Simulations}\label{S:Simulations}

This section presents simulation evidence for the finite sample performance of the semiparametric estimator introduced in Section \ref{S:Inference: Special Regressor}.
I explore the properties of the estimation technique under a wide array of DGP designs that are meant to capture differences in the sample size and in the level of sparsity of the network (see, e.g.,  \citealt{jochmans:2018,dzemski:2019,yan/jiang/fienberg/leng:2019}). 

The undirected network is simulated according the network model in equation \eqref{eq:NFESpecReg}.
I consider a single observed attribute in $X_i$, which is drawn as $X_{i} \sim \mbox{Beta}(2,2) - \frac{1}{2}$.
The pair-specific covariate $W_{ij}=g_0(X_i, X_j)$ is constructed to account for complementarities on the observed attributes and is defined as $W_{ij}= X_{i} X_{j}$.
The agent-specific unobserved factor $A_i$ is generated such that it is correlated with $X_i$ and depends on the sample size $n$. 
This last feature offers a useful approach to control the degree of sparsity in the network.
In particular, I set
\begin{eqnarray*}
	A_i &=& 
	\lambda X_i
	-
	(1-\lambda)
	C_n
	\times
	\mbox{Beta}(0.5,0.5),
\end{eqnarray*}
where the $\mbox{Beta}$ random variable is independent of $X_i$ and concentrates mass at the boundary of the unit interval.
This implies that, conditional on $X_i$, the individuals cluster at small or high types of unobserved attributes.
The parameter $\lambda \in (0,1)$ controls the degree of correlation between the agent-specific heterogeneity and the observed covariate $X_i$, which is set to $\lambda=\frac{3}{4}$. 
The constant $C_n$ depends on the size of the network and takes the values  $C_n \in \left\{\log(\log(n)), \log(n)^{1/2}, \log(n)\right\}$.
Under this design, the choice of $C_n$ regulates the degree of sparsity of the network. 
For larger values of $C_n$, fewer links are formed in the network.
The special regressor $v_{ij}$ is simulated as  $v_{ij} \sim N\left(0,2 \right)$ for $i<j$, and thus satisfies 
the support and independence conditions in Assumptions \ref{Ass:SR02:indep} and \ref{Ass:SR03:fullsupport}.
The link-specific disturbance term is generated as $U_{ij} \sim Beta(2,2) - \frac{1}{2}$ for $i<j$.
The true DGP is completed by setting the parameter value $\theta_0=1.5$ and considering two different network sizes $n \in \left\{ 50, 100\right\}$.

The implementation of the semiparametric estimator for $\theta_0$ requires the estimation of the conditional density of $v_{ij}$ in a nonparametric first stage. 
I consider two approaches to isolate the approximation error induced by the density estimation. 
The first one assumes that the conditional distribution of $v_{ij}$ is known and considers a fixed trimming design given by 
$I_{\tau, ij} = 1\left[ \mid v_{ij}\mid < \tau \right] $ with $\tau = 2std(v_{ij})$. 
In the second approach, I compute the semiparametric estimator as defined in equation \eqref{eq:thetaestimator}.
Although assumption \ref{Ass:Inf03:SmoothnessDensity} requires the use of higher-order kernels to eliminate the asymptotic bias, I compute $\widehat{\theta}_n$ using a standard second-order kernel.  
The motivation for this choice is that semiparametric estimators computed using high-order kernels tends to have inferior finite sample properties
compared to those obtained using standard kernels.
Furthermore, this choice is a common practice in many semiparametric applications (see \citealt{rothe:2009} and \citealt{jochmans:2013}).
I use the standard-normal density as the kernel function. 
The trimming design is the same as in the first approach to ensure a proper comparison between the two alternative methods. 
The bandwidth parameter $h$ is set to be equal to $0.025$. 
I consider different values for the bandwidth parameter, obtaining qualitatively similar results. 
These results are summarized in Appendix \ref{Appx:simulations}.

Table \ref{table:simulation} summarizes the results of computing the semiparametric estimator, assuming that the density function $f_{v}(v_{ij})$ is known, over 500 Monte Carlo replications for all the designs.
In particular, I report the mean, median, standard deviation, and mean square error of $\widehat{\theta}_n$ over the total number of simulations.
The final column of Table \ref{table:simulation} reports the average degree of the network across the total number of simulations.
This information will be used to describe the degree of sparsity in the network across the different designs.

The top panel in Table \ref{table:simulation} shows the results of estimating $\theta_0$ in a small network with $n=50$. 
Both the mean and the median show that the estimator approximates well the true value of $\theta_0=1.5$ independently of the degree of sparsity in the network.
Furthermore, these results suggest that the estimator $\widehat{\theta}_n$ presents the smallest dispersion in the dense network design, with $C_n = \log(\log(n))$ and an average degree of 42\% of the links formed. 
As fewer links are present in the network, the performance of the estimator deteriorates.

In the bottom panel of Table \ref{table:simulation}, I show the results of estimating $\theta_0$ in a large network with $n=100$.
The evidence in this scenario reinforces the previous findings and suggests that the performance of the estimator $\widehat{\theta}_n$ improves across all the designs. 
For example, in the dense network scenario $C_n =\log(\log(n))$, the standard deviation decreases by an order of less than one half and the mean square error by an order greater than one third. 
A similar conclusion is obtained from the sparse network case $C_n =\log(n)$, where only 28\% of the links are formed.

Table \ref{table:simulation_step2} summarizes the results of computing the semiparametric two-step estimator for $\theta_0$ with a first-step kernel estimator $\widehat{f}_{v}(v_{ij})$
over 500 Monte Carlo replications for all the designs.
The top panel in Table \ref{table:simulation_step2} shows the results of estimating $\theta_0$ in a small network with $n=50$. 
These estimates suggest that $\widehat{\theta}_n$ approximates well the true value of $\theta_0$. 
However, this approach obtains less accurate results than those by the first method due to the approximation error induced by the nonparametric first-stage estimation.  
In particular, the estimator presents the best performance and smallest dispersion in the dense network design, where the network has an average degree of 42\% of the links formed. 

In the bottom panel of Table \ref{table:simulation_step2}, I show the results of estimating $\theta_0$ in a large network with $n=100$.
The estimates show that the  performance of the estimator $\widehat{\theta}_n$ improves across all the designs as the network's size grows large, 
including the sparse case where the network has an average degree of 29\% of the links formed.
Overall these numerical experiments suggest that the semiparametric estimator $\widehat{\theta}_n$ yields reliable inference 
for the preference parameter $\theta_0$ in an undirected network formation model.

\vspace{1em}
\begin{table}[h!]       
	\centering
	\caption{Simulation results for the semiparametric estimator $\widehat{\theta}_n$ with known density function $f_v(v_{ij})$}
	\label{table:simulation}
	\begin{threeparttable}
	   \begin{tabular}{l c c c c c}
\toprule
                & mean &median & std & MSE & Degree \\ \midrule
\multicolumn{6}{c}{$n=50$} \\                
$\log(\log(n))$     &1.4764&1.4627&0.9158&0.8393&0.4250\\
$\log(n)^{1/2}$     &1.5052&1.4980&1.0712&1.1476&0.3976\\
$\log(n)$           &1.5217&1.6001&1.3832&1.9136&0.3131\\
\multicolumn{6}{c}{$n=100$} \\               
$\log(\log(n))$     &1.5212&1.5022&0.4809&0.2317&0.4204\\
$\log(n)^{1/2}$     &1.5571&1.5318&0.5381&0.2928&0.3853\\
$\log(n)$           &1.5057&1.4979&0.6916&0.4783&0.2893\\
\bottomrule
\end{tabular}
\begin{tablenotes}
    \item[1] \footnotesize{Total number of Monte Carlo simulations $=500$.}
\end{tablenotes}

	\end{threeparttable}   
\end{table}
\vspace{1em}

\begin{table}[h!]       
	\centering
	\caption{Simulation results for the semiparametric estimator $\widehat{\theta}_n$ with kernel estimator $\widehat{f}_{v}(v_{ij})$ }
	\label{table:simulation_step2}
	\begin{threeparttable}
	   \begin{tabular}{l c c c c c}
\toprule
                & mean &median & std & MSE & Degree \\ \midrule
\multicolumn{6}{c}{$n=50$} \\                
$\log(\log(n))$     &1.6047&1.6164&1.1253&1.2772&0.4237\\
$\log(n)^{1/2}$     &1.6630&1.6179&1.2352&1.5522&0.3963\\
$\log(n)$           &1.6444&1.6643&1.5801&2.5176&0.3125\\
\multicolumn{6}{c}{$n=100$} \\               
$\log(\log(n))$     &1.5373&1.5011&0.4911&0.2425&0.4214\\
$\log(n)^{1/2}$     &1.5955&1.5778&0.5547&0.3168&0.3859\\
$\log(n)$           &1.5415&1.5197&0.7317&0.5371&0.2907\\
\bottomrule
\end{tabular}
\begin{tablenotes}
    \item[1] \footnotesize{Total number of Monte Carlo simulations $=500$.}
    \item[2] \footnotesize{Bandwith parameter $h=0.025$.}
\end{tablenotes}

	\end{threeparttable}   
\end{table}
\clearpage
\section{Conclusion}\label{S:Conclusion}
This paper has studied a network formation model with unobserved agent-specific heterogeneity. 
This paper offers two main contributions to the literature on network formation.
The first contribution is to propose a new identification strategy that identifies the vector of coefficients $\theta_0$, 
which accounts for the preferences for homophilic relationships on the observed attributes.
The point identification result relies on the existence of a special regressor. 
This study represents, to the best of my knowledge, 
the first generalization of a special regressor to analyze a network formation model (\citealt{lewbel:1998} and \citealt{lewbel:2000}).

The second contribution is to introduce a two-step semiparametric estimator for $\theta_0$. 
The estimator has a closed-form and is computationally tractable even in large networks. 
I show in Monte Carlo simulations that the estimator performs well in finite samples, as well as in sparse and dense networks.

Two different strands of the literature on network formation have highlighted the importance of accounting for (i) network externalities, and (ii) general forms of unobserved heterogeneity, (see, e.g., \citealt{graham:2019b}). 
In future research, I plan to explore the identification power that the special regressor has when considering an augmented model of network formation with network externalities and general forms of unobserved heterogeneity. 

\newpage
\singlespacing\bibliography{Bib_SemNet}
\newpage
\appendix
\section{Appendix}\label{Appx:ProofsIdentification}
\subsection{Proof of  Theorem \ref{Theo:IdSR}}
\begin{proof}
Let $e_{ij}=A_{i} + A_{j}-U_{ij}$ and $s(w,e)=-w' \theta_0-e$. Consider  
\begin{eqnarray*}
E[D_{ij}^{\ast}\mid X_{i}, X_{j}] 
&=& 
E
\left[ 
	E
	\left[ 
		D_{ij}^{\ast}
		\mid v_{ij}, X_{i}, X_{j}
	\right]
	\mid 
	X_{i}, X_{j}
\right]
\\
&=& 
\int_{\underline{s}_v}^{\overline{s}_{v}} 
\frac{
	E
	\left[ 
		D_{ij}
		- 
		\bm{1}\left[ v_{ij} > 0 \right]
		\mid v_{ij}, X_{i}, X_{j}
	 \right]
	}{
		f_{v\mid x}(v_{ij}\mid X_{i}, X_{j}) 
	} 
	f_{v \mid x}(v_{ij}\mid X_{i}, X_{j}) \, dv_{ij} 
\\
&=& 
\int_{\underline{s}_v}^{\overline{s}_{v}} 
E
\left[ 
	\bm{1}
	\left[ 
		v_{ij} 
		\geq 
		s(W_{ij}, e_{ij}) 
	\right] 
	- 
	\bm{1}\left[ v_{ij} > 0 \right]
	\mid 
	v_{ij}, X_{i}, X_{j}
\right]
\;
dv_{ij}
\\
&=& 
\int_{\underline{s}_v}^{\overline{s}_{v}} 
\int_{\mathbb{S}_{e}(X_{i}, X_{j})} 
\left\{ 
	\bm{1}\left[ v_{ij} \geq s(W_{ij}, e_{ij}) \right] - \bm{1}\left[ v_{ij} > 0 \right]
\right\}
dF_{e\mid x}(e_{ij} \mid v_{ij}, X_{i}, X_{j}) 
\;
dv_{ij} 
\\
&=& 
\int_{\mathbb{S}_{e}(X_{i}, X_{j} )}  
\int_{\underline{s}_v}^{\overline{s}_{v}} 
\left\{ 
	\bm{1}\left[ v_{ij} \geq s(W_{ij}, e_{ij}) \right] 
	- 
	\bm{1}\left[ v_{ij} > 0 \right]
\right\}
dv_{ij} 
\; 
dF_{e\mid x}(e_{ij}\mid X_{i}, X_{j})  
\\
&=& 
\int_{\mathbb{S}_{e}(X_{i}, X_{j})}  
-s(W_{ij}, e_{ij}) 
dF_{e\mid x}(e_{ij}\mid X_{i}, X_{j})
\\
&=& 
\int_{\mathbb{S}_{e}(X_{i}, X_{j})}  
\left(W_{ij}' \theta_0+e_{ij}\right) 
dF_{e\mid x}(e_{ij}\mid X_{i}, X_{j})
\\
&=& 
W_{ij}' \theta_0 
+  
E \left[e_{ij}\mid X_{i}, X_{j} \right].	
\end{eqnarray*}

The third to last equality follows from  the following result
\begin{eqnarray*}
	\int_{\underline{s}_v}^{\overline{s}_{v}} 
	\left\{ 
		\bm{1}\left[ v_{ij} \geq s(W_{ij}, e_{ij}) \right] - \bm{1}\left[ v_{ij} > 0 \right]
	\right\}
	dv_{ij}
	&=&
	\int_{s(W_{ij}, e_{ij})}^{\overline{s}_{v}} 1 dv_{ij} 
	-
	\overline{s}_{v}
	\\	
	&=& -s(W_{ij}, e_{ij}).
\end{eqnarray*}
\end{proof}

\subsection{Proof of Corollary \ref{cor:id1}} 
\begin{proof}
Theorem \ref{Theo:IdSR} concludes that
\begin{eqnarray*}
	E
	\left[ 
		D_{ik}^{\ast}
		\mid 
		X_{i}, X_{k}
	\right]
	&=& 
	W_{ik}' \theta_0 
	+  
	E
	\left[ 
		A_{i} + A_{k}
		\mid 
		X_{i}, X_{k}
	\right].
\end{eqnarray*}

Observe that $D_{ik}^\ast$ is a function of $(Z_{i}, Z_{k}, A_{i}, A_{k}, U_{ik})$. 
It follows from the the random sampling of nodes, Assumption \ref{Ass:SR00:iidsampling},
and the
conditionally independent formation of links, Assumption \ref{Ass:SR01:distr},  that the following condition holds for any 
tetrad $\sigma\{i,j,k,l\} \in \mathcal{N}_{m_n}$ 
\begin{eqnarray*}
	E [D_{ik}^{\ast}\mid X_{i}, X_{k}]  
	&=& 
	E [D_{ik}^{\ast}\mid X_{\sigma(\{i,j,k,l\})} ]
	\\
	E  \left[A_{i} + A_{k}\mid X_{i}, X_{k}\right]
	&=&
	E  \left[A_{i} + A_{k}\mid X_{\sigma(\{i,j,k,l\})} \right],
\end{eqnarray*}
since the joint distribution of $(v_{i}, v_{k}, A_{i}, A_{k}, U_{ik})$ is conditionally independent of $(X_{j}, X_{l})$, given $(X_{i}, X_{k})$, i.e., 
\begin{eqnarray*}
	Pr(v_{i}, v_{k}, A_{i}, A_{k}, U_{ik}\mid X_{i}, X_{k}) 
	&=& 
	Pr(U_{ik}\mid X_{i}, X_{k}, v_{i}, v_{k}, A_{i}, A_{k})
	Pr(v_{i}, v_{k}, A_{i}, A_{k} \mid X_{i}, X_{k}) 
	\\
	&=&
	Pr(U_{ik}\mid X_{\sigma(\{i,j,k,l\})}, v_{i}, v_{k}, A_{i}, A_{k})
	Pr(v_{i}, v_{k}, A_{i}, A_{k} \mid X_{\sigma(\{i,j,k,l\})} ) 
	\\
	&=&
	Pr(v_{i}, v_{k}, A_{i}, A_{k}, U_{ik}\mid X_{\sigma(\{i,j,k,l\})} ), 
\end{eqnarray*}
where the second equality follows from Assumptions \ref{Ass:SR00:iidsampling} and \ref{Ass:SR01:distr}. 
Thus, the results above yield   
\begin{eqnarray*}
	E [D_{ik}^{\ast} - D_{il}^{\ast}\mid X_{\sigma(\{i,j,k,l\})} ] 
	&
	= (W_{ik}-W_{il})' \theta_0 
	+  
	E  \left[A_{k}-A_{l}\mid X_{\sigma(\{i,j,k,l\})} \right]
	\\	
	E [D_{jk}^{\ast} - D_{jl}^{\ast}\mid X_{\sigma(\{i,j,k,l\})} ] 
	&
	= (W_{jk}-W_{jl})' \theta_0 
	+  
	E  \left[A_{k}-A_{l}\mid X_{\sigma(\{i,j,k,l\})} \right],
\end{eqnarray*}
for any tetrad $\sigma\{i,j,k,l\}$, which in turn implies
\begin{eqnarray}
	E [ \tilde{D}^\ast_{\sigma} \mid X_{\sigma}] 
	&
	= \tilde{W}_{\sigma}' \theta_0.	
\end{eqnarray}

The result follows from Assumption \ref{Ass:SR04:identif}. 
The proof is complete.
\end{proof}
\subsection*{Proof of Theorem \ref{Theo:ID2:main}}
\begin{proof}	
First, notice that for any  
$
\left( X_\sigma, \bar{\bm{v}}_{\sigma} \right)
\in 
Q_\theta
$
\begin{eqnarray*}
	\sign
	\left\{ 
		\tilde{v}_{\sigma} 
	\right\}
	&=&
	\sign
	\left\{ 
		\tilde{v}_{\sigma} 
		+
		\left( 
			\Delta_{\sigma} W_{i}'\theta_0 
			+ 
			\Delta_{\sigma} A
		\right)
		-
		\left( 
			\Delta_{\sigma} W_{j}'\theta_0 
			+ 
			\Delta_{\sigma} A
		\right)
	\right\}
\end{eqnarray*}
since $\mid \tilde{v}_{\sigma}  \mid \geq \overline{s}_{\varepsilon} - \underline{s}_{\varepsilon}$ with probability 1.

Consider a $\theta\neq \theta_0$ with 
$P
\left[ 
	\bar{\bm{Z}}_{\sigma}
	\in 
	Q_\theta	
\right]
>0$.
Without loss of generality, consider some 
$
\left( X_\sigma, \bar{\bm{v}}_{\sigma} \right)
\in 
Q_\theta
$, 
with
$
\tilde{W}_{\sigma}'\theta
\leq 
- \tilde{v}_{\sigma}
< 
\tilde{W}_{\sigma}'\theta_0
$.
From the previous observation, it follows that 
$	
\tilde{v}_{\sigma}	+ \tilde{W}_{\sigma}'\theta_0 + \Delta_{\sigma} A - \Delta_{\sigma} A > 0
$
and 
$\Delta_{\sigma} v_{i} > \overline{s}_{\varepsilon}$, 
$\Delta_{\sigma} v_{j} < \underline{s}_{\varepsilon}$
with probability 1. 

Given
$
\left( X_\sigma, \bar{\bm{v}}_{\sigma} \right)
\in 
Q_\theta		
$, 
it follows that 
$	
\tilde{v}_{\sigma}	+ \tilde{W}_{\sigma}'\theta_0 + \Delta_{\sigma} A - \Delta_{\sigma} A > 0
$
and 
$\Delta_{\sigma} v_{i} > \overline{s}_{\varepsilon}$, 
$\Delta_{\sigma} v_{j} < \underline{s}_{\varepsilon}$
hold if and only if 
\begin{eqnarray}
	\Delta_{\sigma} v_{i} 
	&>& 
	-
	\left( 
	\Delta_{\sigma} W_{i}'\theta_0 
	+
	\Delta_{\sigma} A
	\right)
	\nonumber
	\\
	\Delta_{\sigma} v_{j}
	&\leq& 
	- 
	\left( 
	\Delta_{\sigma} W_{j}'\theta_0 
	+
	\Delta_{\sigma} A
	\right)
	\label{eq:ID02:ineq01}
\end{eqnarray}
with probability 1. 
The inequalities in \eqref{eq:ID02:ineq01} are sufficient conditions for 
\begin{eqnarray*}
	&& P_{\theta_0}
	\left[ 
		\tilde{D}_{\sigma} =2
		\mid 
		X_{\sigma},
		A_{\sigma},
		\bar{\bm{v}}_{\sigma} \in 	\mathcal{V}(X_\sigma), 
		\tilde{D}_{\sigma} \in \left\{ -2, 2\right\}
	\right]
	\\
	&>&
	P_{\theta_0}
	\left[ 
		\tilde{D}_{\sigma} = -2 
		\mid 
		X_{\sigma},
		A_{\sigma},
		\bar{\bm{v}}_{\sigma} \in 	\mathcal{V}(X_\sigma), 
		\tilde{D}_{\sigma} \in \left\{ -2, 2\right\}
	\right],
\end{eqnarray*}
or equivalently, for 
\begin{eqnarray*}
	\mathbb{E}_{\theta_0}
\left[ 
	\tilde{D}_{\sigma}
	\mid 
	X_{\sigma},
	A_{\sigma},
	\bar{\bm{v}}_{\sigma} \in 	\mathcal{V}(X_\sigma), 
	\tilde{D}_{\sigma} \in \left\{ -2, 2\right\}
\right]
> 
0.
\end{eqnarray*}

Notice that for a 
$
\left( X_\sigma, \bar{\bm{v}}_{\sigma} \right)
\in 
Q_\theta
$, 
$
	\mathbb{E}_{\theta_0}
\left[ 
	\tilde{D}_{\sigma}
	\mid 
	X_{\sigma},
	A_{\sigma},
	\bar{\bm{v}}_{\sigma} \in 	\mathcal{V}(X_\sigma), 
	\tilde{D}_{\sigma} \in \left\{ -2, 2\right\}
\right]
> 
0
$
is also sufficient to conclude that 
$\tilde{v}_{\sigma}	+ \tilde{W}_{\sigma}'\theta_0 + \Delta_{\sigma} A - \Delta_{\sigma} A >0$
with probability 1. 
Otherwise, if
$\tilde{v}_{\sigma}	+ \tilde{W}_{\sigma}'\theta_0 + \Delta_{\sigma} A - \Delta_{\sigma} A \leq 0$
with 
$\bar{\bm{v}}_{\sigma} \in \mathcal{V}(X_\sigma)$, 
it would be the case that $\tilde{v}_{\sigma} < 0$, and thus 
\begin{eqnarray*}
	\Delta_{\sigma} v_{i} 
	&\leq& 
	-
	\left( 
	\Delta_{\sigma} W_{i}'\theta_0 
	+
	\Delta_{\sigma} A
	\right)
	\nonumber
	\\
	\Delta_{\sigma} v_{j}
	&>& 
	- 
	\left( 
	\Delta_{\sigma} W_{j}'\theta_0 
	+
	\Delta_{\sigma} A
	\right)
\end{eqnarray*}
with probability 1, which contradicts 
\begin{eqnarray*}
	\mathbb{E}_{\theta_0}
	\left[ 
		\tilde{D}_{\sigma}
		\mid 
		X_{\sigma},
		A_{\sigma},
		\bar{\bm{v}}_{\sigma} \in 	\mathcal{V}(X_\sigma), 
		\tilde{D}_{\sigma} \in \left\{ -2, 2\right\}
	\right]
	> 
	0.
\end{eqnarray*} 

Hence,
\begin{eqnarray*}
	\sign
	\left\{ 
		\mathbb{E}_{\theta_0}
		\left[ 
			\tilde{D}_{\sigma}
			\mid 
			X_{\sigma},
			A_{\sigma},
			\bar{\bm{v}}_{\sigma} \in 	\mathcal{V}(X_\sigma), 
			\tilde{D}_{\sigma} \in \left\{ -2, 2\right\}
		\right]
	\right\}
	&=&
	\sign
	\left\{ 
		\tilde{v}_{\sigma}	+ \tilde{W}_{\sigma}'\theta_0
	\right\}
\end{eqnarray*} 
for any $(X_{\sigma}, A_{\sigma}, \bar{\bm{v}}_{\sigma} \in 	\mathcal{V}(X_\sigma) )$.

The previous result implies that for any
$
\left( X_\sigma, \bar{\bm{v}}_{\sigma} \right)
\in 
Q_\theta
$
with $P \left[ \bar{\bm{Z}}_{\sigma} \in  Q_\theta  \right] >0$,
it will hold that 
$
\tilde{W}_{\sigma}'\theta
\leq 
- \tilde{v}_{\sigma}
< 
\tilde{W}_{\sigma}'\theta_0
$
if and only if  
\begin{eqnarray*}
	\sign
	\left\{ 
		\mathbb{E}_{\theta_0}
		\left[ 
			\tilde{D}_{\sigma}
			\mid 
			X_{\sigma},
			\bar{\bm{v}}_{\sigma} \in 	\mathcal{V}, 
			\tilde{D}_{\sigma} \in \left\{ -2, 2\right\}
		\right]
	\right\}
	>
	\sign
	\left\{ 
		\mathbb{E}_{\theta}
		\left[ 
			\tilde{D}_{\sigma}
			\mid 
			X_{\sigma},
			\bar{\bm{v}}_{\sigma} \in 	\mathcal{V}, 
			\tilde{D}_{\sigma} \in \left\{ -2, 2\right\}
		\right]
	\right\}.
\end{eqnarray*}
This result implies that 
$\bar{\bm{z}}_{\sigma} \in \xi_\theta(X_{\sigma})$, 
and
$P \left[ \bar{\bm{Z}}_{\sigma} \in \xi_\theta \right] >0$.
Therefore, $\theta_0$ is identified relative to $\theta$.
\end{proof}

\subsection*{Proof of Corollary \ref{Theo:ID2:corollary}}
\begin{proof}
	Consider any $\theta \neq \theta_0$.
	It follows from Assumption \ref{Ass:B05:fullrank} that $P\left[ \tilde{W}_\sigma'(\theta - \theta_0) \neq 0  \right] >0$ for any tetrad $\sigma \in \mathcal{N}_{m_n}$.
	Suppose without loss of generality that $P\left[ \tilde{W}_\sigma' \theta <  \tilde{W}_\sigma' \theta_0   \right] >0$.
	Under Assumptions \ref{Ass:SR00:iidsampling} and \ref{Ass:B04:SufficientVariation}, for any $X_\sigma$, with 
	$\tilde{W}_\sigma' \theta <  \tilde{W}_\sigma' \theta_0 $, there exists an interval of $\tilde{v}_\sigma = \Delta_\sigma v_{i} -\Delta_\sigma v_{j}$ 
	with $\tilde{W}_\sigma' \theta \leq - \tilde{v}_\sigma  < \tilde{W}_\sigma' \theta_0 $.
	This implies that
	$
	P
	\left[ 
	\bar{\bm{Z}}_{\sigma}
	\in 
	Q_\theta	
	\right]
	>0
	$
	, and thus $\theta_0$ is point identified relative to all $\theta \neq \theta_0$. 
\end{proof}
\subsection{Proof of Theorem \ref{Theo:Inf:Consistency}} 
\begin{proof}
	
Consider $\widehat{\theta}_n=\widehat{\Gamma}_n^{-1} \times \widehat{\Psi}_{n,\tau}$, with 	
\begin{eqnarray*}
	\widehat{\Gamma}_n
	&=& 
	\frac{1}{m_{n}} 
	\sum_{ \sigma \in \mathcal{N}_{m_{n}}}
		\left[\tilde{W}_{\sigma}
		\tilde{W}_{\sigma}'\right] 
	\\
	\widehat{\Psi}_{n, \tau} 
	&=& 
	\frac{1}{m_{n}} 
	\sum_{ \sigma \in \mathcal{N}_{m_{n}}}
	\left[  
		\tilde{W}_{\sigma} 
		\widehat{D}_{\sigma, \tau}^\ast
	\right]
\end{eqnarray*}	

First, I will show that 
$\widehat{\Gamma}_n \overset{p}{\rightarrow} \Gamma_0$ 
and $
\widehat{\Psi}_{n, \tau} \overset{p}{\rightarrow} \Psi_0$;  
the result will follow Assumption \ref{Ass:SR04:identif}, the continuous mapping theorem and Slutsky.

\textbf{Part 1.} 
Notice that $\widehat{\Gamma}_n-\Gamma_0$ is a mean zero fourth-order V-statistic, without common indices 
\begin{eqnarray*}
	\widehat{\Gamma}_n - \Gamma_0
	&=&
	\frac{1}{m_{n}} 
	\sum_{ \sigma \in \mathcal{N}_{m_{n}}}
	\left\{ 
		\left[\tilde{W}_{\sigma}\tilde{W}_{\sigma}'\right] 
		-
		E \left[ \tilde{W}_{\sigma}\tilde{W}_{\sigma}' \right]
	\right\}.	
\end{eqnarray*}
Lemma \ref{Lemma:TechAppx:UEquiv}  implies that $\hat{\Gamma}_n-\Gamma_0$ can be approximated by a mean zero U-statistic of order 4 at a rate $\sqrt{n}$. 
Assumption \ref{Ass:SR04:identif} ensures that $\Gamma_0$ is finite.
It follows from  Assumption \ref{Ass:SR00:iidsampling} that a Strong Law of Large Numbers  for a U-statisitc holds, and hence, 
$\hat{\Gamma}_n - \Gamma_0 = o_{p}(1)$, (see \citealt[Theorem A, p. 190]{serfling:2009}). 

\textbf{Part 2.}
For a fixed tetrad $\sigma= \sigma(\{i_1,i_2,j_1,j_2\}) \in \mathcal{N}_{m_n}$, let
\begin{eqnarray*}
	\widehat{\eta}_{[l_1l_2], \tau}
	&=&
	\frac{1}{m_{n}} 
	\sum_{ \sigma \in \mathcal{N}_{m_{n}}}
	\tilde{W}_{\sigma}
	\varphi_{ l_1l_2, \tau}
	\left( \frac{ \hat{f}_{x,l_1l_2} }{ \hat{f}_{vx,l_1l_2} } \right),
\end{eqnarray*}
for $(l_1,l_2) \in \left\{ (i_1,j_1), (i_1,j_2), (i_2,j_1), (i_2,j_2)\right\}$.
Next, observe that $\widehat{\Psi}_{n, \tau}$ can be written as 
\begin{eqnarray*}
	\widehat{\Psi}_{n, \tau}
	&=& 
	\left( 
		\widehat{\eta}_{ \left[ i_1j_1 \right], \tau}
		-
		\widehat{\eta}_{ \left[ i_1j_2 \right], \tau}
	 \right)  
	-  
	\left( 
		\widehat{\eta}_{ \left[ i_2j_1 \right], \tau}
		-
		\widehat{\eta}_{ \left[ i_2j_2 \right], \tau}
	\right).	
\end{eqnarray*}

Consistent estimation of $\Psi_0$ will follow from repeated applications of Lemma \ref{Lemma:TechAppx:Consistency}. 
It follows from Lemma \ref{Lemma:TechAppx:Expansion} that 
$\widehat{\eta}_{ \left[ l_1l_2 \right], \tau}$ can be written as 
\begin{eqnarray*}
	\widehat{\eta}_{ \left[ l_1l_2 \right], \tau}
	=
	\frac{1}{m_{n}} 
	\sum_{ \sigma \in \mathcal{N}_{m_{n}}} 
	\tilde{W}_{\sigma}
	\varphi_{ l_1l_2, \tau}
	\left\{ 
		\frac{f_{x,l_1l_2} }{ f_{vx,l_1l_2}}
		+
		\frac{ \widehat{f}_{x,l_1l_2} - f_{x,l_1l_2}  }{ f_{vx,l_1l_2} }
		-
		\frac{ f_{x,l_1l_2} }{ f_{vx,l_1l_2} } 
		\times 
		\frac{ \widehat{f}_{vx,l_1l_2} - f_{vx,l_1l_2} }{ f_{vx,l_1l_2} }
	\right\} 
	+ o_{p}(1).	
\end{eqnarray*}

Then, Lemma \ref{Lemma:TechAppx:Consistency} yields 
\begin{eqnarray*}
	\frac{1}{m_{n}} 
	\sum_{ \sigma \in \mathcal{N}_{m_{n}}} 
	\tilde{W}_{\sigma}
	\varphi_{ l_1l_2, \tau}
	\left\{ 
		\frac{ \hat{f}_{x,l_1l_2} }{ f_{vx,l_1l_2} }
	\right\}
	&=&
	E 
	\left[  
		\tilde{W}_{\sigma}
		\varphi_{ l_1l_2, \tau}
		\frac{f_{x,l_1l_2} }{f_{vx,l_1l_2}}  
	\right]
	+ o_{p}(1)
	\\
	\frac{1}{m_{n}} 
	\sum_{ \sigma \in \mathcal{N}_{m_{n}}} 
	\tilde{W}_{\sigma}
	\varphi_{ l_1l_2, \tau}
	\left\{ 
		\frac{ f_{x,l_1l_2} }{ f_{vx,l_1l_2} }
		\left(  
		\frac{ \hat{f}_{vx,l_1l_2}  }{ f_{vx,l_1l_2} }
		\right)
	\right\}
	&=&
	E 
	\left[  
		\tilde{W}_{\sigma}
		\varphi_{ l_1l_2, \tau}
		\frac{f_{x,l_1l_2} }{f_{vx,l_1l_2}}   
	\right] 
	+ o_{p}(1).	
\end{eqnarray*}

It follows from the previous results and the definition of $D^\ast_{ l_1 l_2, \tau}$ that 
\begin{eqnarray*}
	\widehat{\eta}_{ \left[ i_1j_1 \right], \tau}
	&=&
	\frac{1}{m_{n}} 
	\sum_{ \sigma \in \mathcal{N}_{m_{n}}} 
	\left\{  
		\tilde{W}_{\sigma}
		D^\ast_{ l_1l_2, \tau} 
	\right\} 
	+ E \left[  \tilde{W}_{\sigma}D^\ast_{ l_1l_2, \tau}  \right] 
	- E \left[  \tilde{W}_{\sigma}D^\ast_{ l_1l_2, \tau}  \right] 
	+ o_{p}(1),
\end{eqnarray*}
which is a V-statistic of order 4. 
It follows from Lemma \ref{Lemma:TechAppx:UEquiv} that it can be approximated by a U-statisitcs of order 4.
Assumptions \ref{Ass:Inf01:SamplingMoments} and \ref{Ass:Inf02:Trimming},  and equation \eqref{eq:beta} ensure that $E\left[\tilde{W}_{\sigma}D^\ast_{ l_1l_2, \tau}  \right]$ is finite.
It follows then from Assumptions \ref{Ass:SR00:iidsampling} that a Strong Law of Large Numbers for U-statistics holds, and hence,  
\begin{eqnarray*}
	\frac{1}{m_{n}} 
	\sum_{ \sigma \in \mathcal{N}_{m_{n}}} 
	\left\{  
		\tilde{W}_{\sigma}
		D^\ast_{ l_1l_2, \tau} 
	-
	E 
	\left[  
		\tilde{W}_{\sigma}D^\ast_{ l_1l_2, \tau}  
	\right] 
	\right\} 
	&=&
	o_p(1).
\end{eqnarray*}

Consider next
\begin{eqnarray*}
	\frac{1}{m_n}
	\sum_{\sigma \in \mathcal{N}_{m_n} }
	\tilde{W}_{\sigma}
	\left\{ 
		D^\ast_{ l_1l_2} 
		-	
		D^\ast_{ l_1l_2, \tau} 
	\right\}
	&=&
	\frac{1}{n(n-1)}
	\sum_{l_1}^{n}
	\sum_{l_2\neq l_1}
	D^\ast_{ l_1l_2}
	\left\{ 
		\left( 
			1 
			-	
			I_{\tau, l_1l_2}
		\right)
	\right\}
	\tilde{W}_{l_1l_2}(\sigma)
\end{eqnarray*}
where the equality follows from the definition of $D^\ast_{ l_1l_2, \tau} $ and 
\begin{eqnarray*}
	\tilde{W}_{l_1l_2}(\sigma)
	&=&
	\frac{1}{(n-2)(n-3)}
	\sum_{s_1 \neq l_1, l_2}
	\sum_{s_2\neq l_1, l_2, s_1}
		\tilde{W}_{\sigma\{l_1,s_1;l_2,s_2\}}.
\end{eqnarray*}

It follows from using a Cauchy-Schwarz inequality, that the expectation
\begin{eqnarray*}
	E
	\left[ 
		\left( 
			\frac{1}{ n(n-1)}
			\sum_{l_1}^{n}
			\sum_{l_2\neq l_1}
			D^\ast_{ l_1l_2}
			\left\{ 
				\left( 
					1 
					-	
					I_{\tau,l_1l_2}
				\right)
			\right\}
			\tilde{W}_{l_1l_2}(\sigma)
		\right)^2
	\right]	
\end{eqnarray*}
is bounded by 
\begin{eqnarray*}
	\frac{1}{ n(n-1)}
	\sum_{l_1}^{n}
	\sum_{l_2\neq l_1}
	E
	\left[ 
		\left(  
		D^\ast_{ l_1l_2}
		\left\{ 
			\left( 
				1 
				-	
				I_{\tau, l_1l_2}
			\right)
		\right\}
		\tilde{W}_{l_1l_2}(\sigma)
		\right)^2	
	\right]
	&=&
	O
	\left(  
	E
	\left[ 
		\tilde{W}_{l_1l_2}(\sigma)^2
		\left( D^\ast_{ l_1l_2} \right)^2
	\left( 
		1 
		-	
		I_{\tau, l_1l_2}
	\right)^2	
	\right]
	\right)
	\\
	&\leq& 
	\sup_{\sigma} \left( \tilde{W}_{\sigma}^2 \right)
	\sup_{l_1l_2} \left( D^{\ast}_{ l_1l_2} \right)^2
	O
	\left( 
	E\left[ 
		\left( 
			1 
			- 
			I_{\tau,l_1l_2}
		\right)^2 
	\right]
	\right).
\end{eqnarray*}
where the inequality follows from Assumption \ref{Ass:Inf01:SamplingMoments}. Assumption \ref{Ass:Inf02:Trimming}  yields
\begin{eqnarray*}
	E\left[ 
		\left( 
			1 
			- 
			I_{\tau, l_1l_2}
		\right)^2 
		\right]
	&=&
	P\left[ 
		I_{\tau, l_1l_2}=0
	\right]
	=
	o(\tau). 
\end{eqnarray*}

Using the results above to conclude that 
\begin{eqnarray*}
	E
	\left[ 
		\left( 
			\frac{1}{ n(n-1) }
			\sum_{l_1}^{n}
			\sum_{l_2\neq l_1}
			D^\ast_{ l_1l_2}
			\left\{ 
				\left( 
					1 
					-	
					I_{\tau,l_1l_2}
				\right)
			\right\}
			\tilde{W}_{l_1l_2}(\sigma)
		\right)^2
	\right]	
	&\leq&
	o(\tau),
\end{eqnarray*} 
and hence, 
\begin{eqnarray*}
	\frac{1}{m_{n}} 
	\sum_{ \sigma \in \mathcal{N}_{m_{n}}} 
	\left\{  
		\tilde{W}_{\sigma}
		D^\ast_{ l_1l_2, \tau} 
	-
	E 
	\left[  
		\tilde{W}_{\sigma} D^\ast_{ l_1l_2}  
	\right] 
	\right\} 
	&=& o(1).
\end{eqnarray*}

Using similar steps for 
$\widehat{\eta}_{ \left[ i_1j_2 \right], \tau}, 
\widehat{\eta}_{ \left[ i_2j_1 \right], \tau}$, and 
$\widehat{\eta}_{ \left[ i_2j_2 \right], \tau}$, yields: 
\begin{eqnarray*}
	\hat{\Psi}_{n,\tau}
	-
	E 
	\left[  
		\tilde{W}_{\sigma}
		\tilde{D}_{\sigma}^\ast
	\right] 
	&=&
	o_{p}(1).
\end{eqnarray*}
The result follows from Assumption \ref{Ass:SR04:identif}, the Continuous Mapping Theorem and Slutsky's Theorem. 
\end{proof}
\subsection{Proof of Theorem \ref{Theo:Inf:AN}} 

\begin{proof}

\textbf{Part 1: H\'{a}jek Projection}
	
Under  Assumptions \ref{Ass:SR00:iidsampling}-\ref{Ass:SR04:identif}, \ref{Ass:Inf01:SamplingMoments}-\ref{Ass:Inf04:Kernel}, 
it follows from the proof of Theorem \ref{Theo:Inf:Consistency} that 
$\hat{\Gamma}_{n} \overset{p}{\rightarrow} \Gamma_0$, 
and from Lemma \ref{Lemma:TechAppx:Expansion} that 
$	\widehat{\eta}_{ \left[ l_1l_2 \right], \tau}$ can be written as 
\begin{align*}
	\widehat{\eta}_{ \left[ l_1l_2 \right], \tau}
	=
	\frac{1}{m_{n}} 
	\sum_{ \sigma \in \mathcal{N}_{m_{n}}} 
	\tilde{W}_{\sigma}
	\varphi_{ l_1l_2, \tau}
	\left\{ 
		\frac{f_{x,l_1l_2} }{ f_{vx,l_1l_2}}
		+
		\frac{ \widehat{f}_{x,l_1l_2} - f_{x,l_1l_2}  }{ f_{vx,l_1l_2} }
		-
		\frac{ f_{x,l_1l_2} }{ f_{vx,l_1l_2} } 
		\times 
		\frac{ \widehat{f}_{vx,l_1l_2} - f_{vx,l_1l_2} }{ f_{vx,l_1l_2} }
	\right\} 
	+ o_{p}(1)
\end{align*}
for $(l_1, l_2) \in  \left\{ (i_1,j_1), (i_1,j_2), (i_2,j_1), (i_2,j_2) \right\}$.

Hence,  
$
\widehat{\Psi}_{n, \tau}
=
\left( 
	\widehat{\eta}_{ \left[ i_1j_1 \right], \tau}
	-
	\widehat{\eta}_{ \left[ i_1j_2 \right], \tau}
 \right)  
-  
\left( 
	\widehat{\eta}_{ \left[ i_2j_1 \right], \tau}
	-
	\widehat{\eta}_{ \left[ i_2j_2 \right], \tau}
\right),
$
which can be expressed as
$\widehat{\Psi}_{n, \tau} = S_{1,n \tau} + S_{2,n\tau} - S_{3,n\tau} + o_p(1)$ 
using the expression above, with 
\begin{eqnarray*}
	S_{1,n \tau}
	&=&
	\frac{1}{m_{n}} 
	\sum_{ \sigma \in \mathcal{N}_{m_{n}}} 
	\tilde{W}_{\sigma}
	\left\{ 
		\left( 
			D^\ast_{ {i_1 j_1}, \tau}
			-
			D^\ast_{ {i_1 j_2}, \tau}	
		\right)
		-
		\left( 
			D^\ast_{ {i_2 j_1}, \tau}
			-
			D^\ast_{ {i_2 j_2}, \tau}	
		\right)
	\right\}
	\\
	S_{2,n \tau}
	&=&
	\frac{1}{m_{n}} 
	\sum_{ \sigma \in \mathcal{N}_{m_{n}}} 
	\tilde{W}_{\sigma}
	\left\{ 
		\left( 
			\frac{ 
				\varphi_{ {i_1 j_1}, \tau}	
				\widehat{f}_{x,{i_1 j_1}} 
			}{ 
				f_{vx,{i_1 j_1}} 
			}
			-
			\frac{ 
				\varphi_{ {i_1 j_2}, \tau}	
				\widehat{f}_{x,{i_1 j_2}} 
				}{ 
				f_{vx,{i_1 j_2}} 
				}
		\right)
		-
		\left( 
			\frac{ 
				\varphi_{ {i_2 j_1}, \tau}
				\widehat{f}_{x,{i_2 j_1}} 
				}{ 
				f_{vx,{i_2 j_1}} 
			}
			-
			\frac{ 
				\varphi_{ {i_2 j_2}, \tau}	
				\widehat{f}_{x,{i_2 j_2}} 
			}{ 
				f_{vx,{i_2 j_2}} 
			}
		\right)
	\right\}
	\\
	S_{3,n \tau}
	&=&
	\frac{1}{m_{n}} 
	\sum_{ \sigma \in \mathcal{N}_{m_{n}}} 
	\tilde{W}_{\sigma}
	\left\{ 
		\left( 
			\frac{ 
				D_{ {i_1 j_1}, \tau}^\ast 	
				\widehat{f}_{vx,{i_1 j_1}} 
			}{ 
				f_{vx,{i_1 j_1}} 
			}
			-
			\frac{ 
				D_{ {i_1 j_2}, \tau}^\ast 	
				\widehat{f}_{vx,{i_1 j_2}} 
			}{ 
				f_{vx,{i_1 j_2}} 
			}
		\right)
		-
		\left( 
			\frac{ 
				D_{ {i_2 j_1}, \tau}^\ast 
				\widehat{f}_{vx,{i_2 j_1}} 
			}{ 
				f_{vx,{i_2 j_1}} 
			}
			-
			\frac{ 
				D_{ {i_2 j_2}, \tau}^\ast 	
				\widehat{f}_{vx,{i_2 j_2}} 
			}{ 
				f_{vx,{i_2 j_2}} 
			}
		\right)
	\right\}.
\end{eqnarray*}

Consider
\begin{eqnarray*}
\left( 
	\widehat{\Psi}_{n, \tau} 
	-
	E \left[ 
	\tilde{W}_{\sigma}
	\widetilde{D}_{\sigma,\tau}^\ast
	\mid 
	\Omega_n
	\right]
\right)
&=& 
\left\{ 
S_{1,n \tau} 
-
E 
\left[ 
	\tilde{W}_{\sigma}
	\widetilde{D}_{\sigma,\tau}^\ast
\mid 
\Omega_n	
\right]
\right\}
+ 
S_{2,n\tau} 
- 
S_{3,n\tau} 
+ o_p(1),
\end{eqnarray*}
it follows from Lemmas 
\ref{Lemma:TechAppx:Hajek_Sum1}, 
\ref{Lemma:TechAppx:Hajek_Sum2}, 
and 
\ref{Lemma:TechAppx:Hajek_Sum3}
that the 
H\'{a}jek projection of
\begin{eqnarray*}
	\left(  
		\widehat{\Psi}_{n, \tau} 
		-
		E \left[ 
		\tilde{W}_{\sigma}
		\widetilde{D}_{\sigma,\tau}^\ast
		\mid 
		\Omega_n 
		\right]
	\right)	
\end{eqnarray*}
into an arbitrary function of $\zeta_{i_1j_1} = \left( X_{i_1},X_{j_1}, A_{i_1},A_{j_1}, v_{i_1j_1},U_{i_1j_1} \right)$
is given by 
\begin{eqnarray*}
	\left( 
	\widehat{\Psi}_{n,\tau}
	-
	E
	\left[ 
		\tilde{W}_{\sigma}
		\tilde{D}_{\sigma\tau}^\ast
		\mid 
		\Omega_n
	\right]	
	\right)
	&=&
	V_{n}^\ast
	+
	o_{p}
	\left(  
	\sqrt{
		\frac{
		\varrho_{n}
		}{n(n-1)}	
		}
	\right)
\end{eqnarray*}
where
\begin{eqnarray*}
V_{n}^\ast
&=&
\frac{1}{n(n-1)} 
\sum_{i_1 = 1}^{n} \sum_{j_1 \neq i_1}
\xi_{i_1j_1,\tau}
\\		
\xi_{i_1j_1, \tau}
&=&
\left\{ 
	D^\ast_{ i_1 j_1} 
	-
	E
	\left[ 
		D^\ast_{i_1j_1}
		\mid 
		\omega_{i_1j_1}	
	\right]
\right\}
I_{\tau, i_1j_1} 
\overline{\chi}_{i_1j_1}
\\
\overline{\chi}_{i_1j_1}
&=&
\left\{ 
	\frac{1}{(n-2)(n-3)}
	\sum_{i_2 \neq i_1 , j_1} \sum_{j_2 \neq i_1 , j_1, i_2} 
		E\left[ 
			\tilde{W}_{\sigma\{i_1,i_2;j_1,j_2\}}
			\mid 	
			X_{i_1}, X_{j_1}
		\right]
\right\}	
\end{eqnarray*}
and 
\begin{eqnarray*}
	\Upsilon_{n,\tau} 
	&=& n(n-1)	
	Var
	\left(
	V_{n}^\ast
	 \right)
	=
	\frac{1}{n(n-1)}
	\left\{ 
	\sum_{i_1 =1}^n  \sum_{j_1 \neq i_1} 
	\Lambda^\ast_{i_1, j_1}
	\right\}
	\\
	\Lambda^\ast_{i_1, j_1}
	&=&
	E
	\left[ 
		\left\{ 
			E
			\left[ 
				D^\ast_{i_1j_1}
				D^\ast_{i_1j_1}
				\mid 
				\omega_{i_1j_1}
			\right]	
			-
			E
			\left[ 
				D^\ast_{i_1j_1}
				\mid 
				\omega_{i_1j_1}
			\right]^2		
		\right\}
		I_{\tau, i_1j_1}^2
		\overline{\chi}_{i_1j_1}
		\overline{\chi}_{i_1j_1}'
	\right]
	\\
	\varrho_{n, \tau}
	&=&
	O
	\left( 
		\Upsilon_{n,\tau} 
	 \right)
	=
	O
	\left( 
		E
		\left[ 
			\left\{ 
				\frac{
					p_n(\omega_{i_1j_1})
					\left[ 1-  p_{n}(\omega_{i_1j_1})  \right]
				}{
					f_{v\mid x, i_1j_1}
				}	
			\right\}
			I_{\tau, i_1j_1}
		\right]
	\right). 
\end{eqnarray*}

\textbf{Part 2: Bias Reduction}

Consider next, 
\begin{eqnarray*}
	n(n-1)
\varrho_{n}^{-1}
E
\left[ 
\left( 	
\frac{1}{m_n}
\sum_{ \sigma \in \mathcal{N}_{m_{n}}} 	
	\tilde{W}_{\sigma}
	\left\{ 
		\tilde{D}_{\sigma\tau}^\ast
	-
		\tilde{D}_{\sigma}^\ast
	\right\}
\right)^2
\mid 
\Omega_n
\right].	
\end{eqnarray*}
It follows from a Cauchy-Schwarz inequality that the term above is bounded by 
\begin{eqnarray*}
&&
n(n-1)
\varrho_{n}^{-1}
\frac{1}{m_n}
\sum_{ \sigma \in \mathcal{N}_{m_{n}}} 	
E
\left[ 
\left( 	
		\tilde{D}_{\sigma\tau}^\ast
	-
		\tilde{D}_{\sigma}^\ast
\right)^2
\mid 
\Omega_n
\right]
\tilde{W}_{\sigma}
\tilde{W}_{\sigma}'
\end{eqnarray*}
which is equal to 
\begin{eqnarray*}
	O\left(  
n(n-1)
\varrho_{n}^{-1}
\left\{ 
	E
	\left[ 
		D_{i_1j_1}^\ast
		D_{i_1j_1}^\ast
	\mid 
	\omega_{i_1j_1}	
	\right]
	-
	E
	\left[ 
		D_{i_1j_1}^\ast
	\mid 
	\omega_{i_1j_1}	
	\right]
	E
	\left[ 
		D_{i_1j_1}^\ast
	\mid 
	\omega_{i_1j_1}	
	\right]'
\right\}
\left( I_{\tau, i_1j_1} -1 \right)^2
\tilde{W}_{\sigma}
\tilde{W}_{\sigma}'
\right).
\end{eqnarray*}

Assumptions \ref{Ass:Inf01:SamplingMoments}
and \ref{Ass:Inf02:Trimming}  yield 
\begin{eqnarray*}
&&
\sup_{\sigma}\left( \tilde{W}_{\sigma}\right)
\sup_{\sigma}\left( \tilde{W}_{\sigma}\right)'
O\left(  
n(n-1)
\varrho_{n}^{-1}
\left\{ 
	\frac{
		p_n\left( \omega_{i_1j-1} \right)
		\left[ 1- p_n\left( \omega_{i_1j-1} \right) \right]
	}{
			f_{v\mid x, i_1j_1 }
	}
\right\}
\left( I_{\tau, i_1j_1} -1 \right)^2
\right)
\\
&=&
O
\left( 
	n(n-1)
	\tau 
\right)
=0
\end{eqnarray*}
since 
$\left( I_{\tau, i_1j_1} -1 \right)$
as $\tau \rightarrow 0$ and $n \rightarrow \infty$.

Therefore, 
\begin{eqnarray*}
n(n-1)
\varrho_{1n}^{-1} 
E
\left[ 	
\left( 	
	\frac{1}{m_n}
	\sum_{ \sigma \in \mathcal{N}_{m_{n}}} 	
	\tilde{W}_{\sigma}
	\left\{ 
		\tilde{D}_{\sigma\tau}^\ast
	-
		\tilde{D}_{\sigma}^\ast
	\right\}
\right)^2
\mid 
\Omega_{n}
\right]	
&=&
o(1),
\end{eqnarray*}
and so 
\begin{eqnarray*} 	
\\
n(n-1)
\varrho_{1n}^{-1} 	
\left( 
	E
	\left[ 
		\tilde{W}_{\sigma}
		\tilde{D}_{\sigma\tau}^\ast	
		\mid 
		\Omega_\sigma
	\right]
	-
	E
	\left[ 
		\tilde{W}_{\sigma}
		\tilde{D}_{\sigma}^\ast	
		\mid 
		\Omega_\sigma
	\right]
\right)
&=&
o(1).
\end{eqnarray*}

\textbf{Part 3: Limit Distribution of Projection}

Given Assumptions \ref{Ass:SR01:distr}, 
the H\'{a}jek projection $V_{n}^\ast$
is an average of $\left\{ \xi_{i_1j_1, \tau} \right\}$, 
which are conditionally independent given 
$\Omega_{n}= \left(\mathbf{v}_n,  \mathbf{X}_n, \mathbf{A}_n \right) $, 
with conditional mean 
\begin{eqnarray*}
	E
	\left[ 
		\xi_{i_1j_1, \tau}
		\mid \Omega_n
	\right]
	&=&
	0
\end{eqnarray*}
and conditional variance
\begin{eqnarray*}
\Upsilon\left( {\Omega_n} \right) 
&=&
n(n-1)
Var 
\left( 
\frac{1}{n(n-1)}
\sum_{i_1 =1}^{n}\sum_{j_1 \neq i_1}	
\xi_{i_1j_1}
\mid 
\Omega_{n}	
\right)
\\
&=&
\frac{1}{n(n-1)}
\sum_{i_1 =1}^{n}\sum_{j_1 \neq i_1}
\left\{ 
E
\left[ 	
D^\ast_{ i_1 j_1} 
D^\ast_{ i_1 j_1}
\mid 
\omega_{i_1j_1}
\right] 
-
E
\left[ 
	D^\ast_{i_1j_1}
	\mid 
	\omega_{i_1j_1}	
\right]^2
\right\}
I_{\tau, i_1j_1}
\overline{\chi}_{i_1j_1}
\overline{\chi}_{i_1j_1}'.
\end{eqnarray*}

Given Assumption
\ref{Ass:Inf0302:Boundedness},  
a conditional version of Lyapunov's Central Limit Theorem holds, and hence  
\begin{eqnarray*}
	\Upsilon\left( {\Omega_n} \right)^{-1/2} 
	\left\{ 
		\frac{1}{\sqrt{n(n-1)}}
		\sum_{i_1 =1}^{n}\sum_{j_1 \neq i_1}	
		\xi_{i_1j_1, \tau}
	\right\}
	\Rightarrow
	\mathcal{N}
	\left( 
	0, 
	I	
	\right).
\end{eqnarray*}
Now, it follows from using  \ref{Ass:Inf0302:Boundedness} that  
$
\| 
\Upsilon\left( {\Omega_n} \right)
-
\Upsilon_{n}
\|
\overset{p}{\rightarrow} 0
$
as $n\rightarrow \infty$.
It follows then that the limiting distribution is independent of the conditional values, and
therefore, the limiting distribution continues to hold unconditionally, with $\Upsilon_{n}$ replacing 
$\Upsilon\left( {\Omega_n} \right)$. That is, 
\begin{eqnarray*}
	\Upsilon_{n}^{-1/2}	
\left\{ 
	\frac{1}{\sqrt{n(n-1)}}
\sum_{i_1 = 1}^{n} \sum_{j_1 \neq i_1}
\xi_{i_1j_1, \tau}
\right\}
\Rightarrow
\mathcal{N}
\left( 
0, 
I	
\right).
\end{eqnarray*}

\textbf{Part 4: Limiting distribution of  $\widehat{\theta}_n$}

Consider the matrix $\Sigma_{n}$, defined as 
$
\Sigma_{n}
=
\Gamma_{0}^{-1}
\times 
\Upsilon_{n} 
\times 
\Gamma_{0}^{-1}.
$
The limiting distribution of the $\widehat{\theta}_n$
follows from the definitions of 
$\widehat{\Psi}_{n}^{-1} $ and $\Sigma_{n}$, 
and from applying Slutsky's theorem. 
In other words, 
\begin{eqnarray*}
&&
\sqrt{n(n-1)}
\Sigma_{n}^{-1/2}
\left( 
\widehat{\theta}_n
-
\theta_0
\right)
\\
&=&
\sqrt{n(n-1)}
\Sigma_{n}^{-1/2}
\times 
\left\{ 
	\widehat{\Gamma}_{n}^{-1}
\left[ 	
	\frac{1}{m_{n}} 
	\sum_{ \sigma \in \mathcal{N}_{m_{n}}} 
	\left\{ 
		\tilde{W}_{\sigma}
		\tilde{D}^\ast_{\sigma,\tau} 
	-
	E
	\left[ 
		\tilde{W}_{\sigma}
		\tilde{D}_{\sigma}^\ast	
		\mid 
		\Omega_{\sigma}
	\right]
	\right\}
\right]
\right\}
\\
&=&
\Gamma_{0}^{1/2}
\times 
\Upsilon_{n}^{-1/2}	
\times 
\Gamma_{0}^{-1/2}
\times 
\left\{ 
	\frac{1}{\sqrt{n(n-1)}}
	\sum_{i_1 = 1}^{n} \sum_{j_1 \neq i_1} 
	\xi_{i_1j_1, \tau}	
\right\}
+
o_{p}
\left( 1 \right)
\\
&\Rightarrow&
\mathcal{N}
\left( 
0, I
\right).
\end{eqnarray*}
	
The proof is complete.
\end{proof}
\newpage
\section{Technical Appendix}
\subsection{Equivalent representation for V statistics}

The following lemma provides a U-statistic representation for a V-statistic when the kernel varies with $n$.
Given $n$ and for $m\leq n$, let $\sum_{(n,m)}$ denote the sum over the $\binom{n}{m}$ combinations of $m$ distinct elements  $(i_{1}, \cdots, i_{m})$ from $(1, \cdots, n)$,  and let $\sum_{\Pi_{m!}}$ denote the sum over the $m!$ permutations $(i_{1}, \cdots, i_{m})$ of $(1, \cdots, m)$.

Let $V_n$ be a $V$-statistic or order $m$, without common indices
\begin{eqnarray*}
	V_{n} &=& 
	\frac{1}{n^m}
	\sum_{i_{1}, \cdots, i_{m}=1}^{n}  
	\frac{1}{h^L}
	\gamma(X_{i_1}, \cdots, X_{i_m})
	\bm{1}\left[ i_{1} \neq \cdots \neq i_{m}  \right]
\end{eqnarray*}
where  $h\rightarrow 0$ as $n\rightarrow \infty$,  and $\gamma: \mathbb{R}^L \mapsto \mathbb{R}$.

Let
\begin{eqnarray*}
	U_{n}
	&=&
	\binom{n}{m}^{-1} 
	\sum_{(n,m)} 
	\phi_{h}(X_{i_1}, \cdots, X_{i_m})
	\\
	\phi_{h}(X_{1}, \cdots, X_{m})
	&=& 
	\frac{1}{m!}
	\sum_{\Pi_{m!}} 
	\frac{1}{h^L}
	\gamma(X_{\pi_1}, \cdots, X_{\pi_m})
\end{eqnarray*}

\begin{lemma}
\label{Lemma:TechAppx:UEquiv}
Suppose that $E \mid\mid  \gamma(X_{i_1}, \cdots, X_{i_m})  \mid\mid^{2} < \infty$ for all $1 \leq i_{1}, \cdots, i_m \leq m$ and $m \leq n$, and $nh^2 \rightarrow \infty$.
Then, 
\begin{eqnarray*}
V_{n} - U_{n} &=& o_p(1).
\end{eqnarray*}	
\end{lemma}
\begin{proof}
Let 
\begin{eqnarray*}
	\gamma_{h}(X_{i_1}, \cdots, X_{i_m})
	=
	\frac{1}{h^L}
	\gamma(X_{i_1}, \cdots, X_{i_m}),
\end{eqnarray*}
and notice that 
\begin{eqnarray}
	n^{m} V_{n}
	&=&
	\sum_{(n,m)}
	\sum_{\Pi_{m!}}
	\gamma_{h}(X_{\pi_1}, \cdots, X_{\pi_m})	
	\label{eq:appx:ustatexpansion}
	\\
	&=&
	\left[
		n(n-1) \cdots (n-m+1)
	\right]
	\binom{n}{m}^{-1}
	\sum_{(n,m)}
	\phi_{h}(X_{i_1}, \cdots, X_{i_m})
	\nonumber
	\\
	&=&
	\left[
		n(n-1) \cdots (n-m+1)
	\right]
	U_{n},
	\nonumber
\end{eqnarray}
and hence, $\left( U_{n} -V_{n} \right) = O(n^{-1}) U_{n}$.

Consider now
\begin{eqnarray*}
	E \left[ \left( U_n - V_n \right)^2  \right]	
	&=&
	O\left( \frac{1}{n^2} \right)
	E \left[ U_{n}^2 \right],
\end{eqnarray*}
and notice that a Cauchy-Schwarz inequality yields
\begin{eqnarray*}
	E \left[ U_{n}^2 \right]
	&=&
	\binom{n}{m}^{-2} 
	E \left[
	\left(  \sum_{(n,m)} 
	\phi_{h}(X_{i_1}, \cdots, X_{i_m})
	\right)^2
	\right]
	\\
	&\leq&
	\binom{n}{m}^{-2}
	\binom{n}{m}^{2} 
	E \left[
		\phi_{h}(X_{i_1}, \cdots, X_{i_m})^2
	\right]
\end{eqnarray*}
where 
\begin{eqnarray*}
	E \left[
		\phi_{h}(X_{i_1}, \cdots, X_{i_m})^2
	\right]
	&=&
	\frac{1}{h^{2L}}
	O
	\left(  
	E \left[
		\gamma(X_{i_1}, \cdots, X_{i_m})^2
	\right]
	\right)
	\\
	&=&
	O\left( \frac{1}{h^{2L}} \right)
\end{eqnarray*}
since $E \mid\mid  \gamma(X_{i_1}, \cdots, X_{i_m})  \mid\mid^{2} < \infty$ by assumption, and hence, 
\begin{eqnarray*}
	E \left[ \left( U_n - V_n \right)^2  \right]	
	&\leq&
	O\left( \frac{1}{(nh^L)^2} \right)
	=
	o\left( 1 \right)
\end{eqnarray*}
as $nh^L \rightarrow \infty$.

\end{proof}

Notice that, unlike Lemma 5.7.3 in \citet[page 206]{serfling:2009} and Theorem 1 in \citet[page 183]{lee:2019}, 
in equation \ref{eq:appx:ustatexpansion} 
the average of terms with at least one common index is equal to zero due to the specification of the V-statistic without common indices.
\subsection{Consistency for V-statistics}
\begin{lemma}\label{Lemma:TechAppx:Consistency}
Suppose that the Assumptions in Theorem \ref{Theo:Inf:Consistency} hold. 
Then 
\begin{eqnarray*}
	\frac{1}{m_{n}} 
	\sum_{ \sigma \in \mathcal{N}_{m_{n}}} 
	\tilde{W}_{\sigma}
	\varphi_{ l_1l_2, \tau}
	\left\{ 
		\frac{ \hat{f}_{x,l_1l_2} }{ f_{vx,l_1l_2} }
	\right\}
	-
	E 
	\left[  
		\tilde{W}_{\sigma}
		\varphi_{ l_1l_2, \tau}
		\frac{f_{x,l_1l_2} }{f_{vx,l_1l_2}}  
	\right]
	&=& o_{p}(1)
	\\
	\frac{1}{m_{n}} 
	\sum_{ \sigma \in \mathcal{N}_{m_{n}}} 
	\tilde{W}_{\sigma}
	\varphi_{ l_1l_2, \tau}
	\left\{ 
		\frac{ f_{x,l_1l_2} }{ f_{vx,l_1l_2} }
		\left(  
		\frac{ \hat{f}_{vx,l_1l_2}  }{ f_{vx,l_1l_2} }
		\right)
	\right\}
	-
	E 
	\left[  
		\tilde{W}_{\sigma}
		\varphi_{ l_1l_2, \tau}
		\frac{f_{x,l_1l_2} }{f_{vx,l_1l_2}}   
	\right] 
	&=& o_{p}(1)
\end{eqnarray*}	
with $(l_1,l_2) \in \left\{ (i_1, j_1), (i_1, j_2), (i_2, j_1), (i_2, j_2) \right\}$ for a given 
tetrad 
$\sigma\{i_1,i_2, j_1,j_2\} \in \mathcal{N}_{m_n}$.
\end{lemma}
\begin{proof}
This proof focuses on the first result since the second one follows from similar arguments. Let 
\begin{eqnarray*}
	\hat{V}_{n}
	&=&
	\frac{1}{m_{n}} 
	\sum_{ \sigma \in \mathcal{N}_{m_{n}}} 
	\tilde{W}_{\sigma}
	\frac{ \varphi_{l_1l_2, \tau} }{ f_{vx,l_1l_2} }
	\hat{f}_{x,l_1l_2},
\end{eqnarray*}
and recall that the kernel estimator $\hat{f}_{x,l_1 l_2}$ is defined as 
\begin{eqnarray*}
	\hat{f}_{x,l_1 l_2} 
	&=&
	\frac{1}{(n-2)(n-3)} 
	\sum_{k_1\neq i_1,j_1}
	\sum_{k_2\neq i_1,j_1, k_1}
	\frac{1}{h^L} 
	K_{x,h} \left( X_{k_1}-X_{l_1} , X_{k_2}-X_{l_2}  \right). 
\end{eqnarray*}

Plugging $\hat{f}_{x,l_1 l_2} $ into $\hat{V}_{n}$ yields the following V-statistic of order six 
\begin{eqnarray*}
	6! \binom{n}{6}^{-1} 
	\sum_{i_1\neq i_2\neq j_1\neq j_2 \neq k_1 \neq k_2}
	\frac{1}{h^L}
	\tilde{W}_{i_1i_2;j_1j_2}
	\frac{ \varphi_{ l_1l_2, \tau} }{ f_{vx, l_1l_2} }
	K_{x,h} \left( X_{k_1}-X_{l_1} , X_{k_2}-X_{l_2}  \right).
\end{eqnarray*}

Assumptions \ref{Ass:Inf01:SamplingMoments} and \ref{Ass:Inf04:Kernel} imply that 
\begin{eqnarray*}
	E 
	\left[  
	\mid\mid 
	\frac{1}{h^{L}}
	\tilde{W}_{i_1i_2;j_1j_2}
	\frac{ \varphi_{ l_1l_2, \tau} }{ f_{vx, l_1l_2} }
	K_{x,h} \left( X_{k_1}-X_{l_1} , X_{k_2}-X_{l_2}  \right)
	\mid\mid^2 
	\right] 
	&<& 
	\infty,	
\end{eqnarray*}
it then follows from Lemma \ref{Lemma:TechAppx:UEquiv} that $\hat{V}_n$ is asymptotically equivalent to a six-order U-statistic as $n h^{L} \rightarrow \infty$.
In particular,  $\left(U_{n} - \hat{V}_{n} \right)= o_{p}(1)$ where 
\begin{eqnarray*}
	U_{n} 
	&=&  
	\binom{n}{6}^{-1} 
	\sum_{i_1 < \cdots < i_6}
	\phi_{\bar{\sigma}\{i_1, \cdots, i_6\}  , \tau}
	\\
	\phi_{\bar{\sigma}\{i_1, \cdots, i_6\}  , \tau}, 
	&=&
	(6!)^{-1} 
	\sum_{ \pi \in  \Pi_{6!} } 
	\frac{1}{h^{L}}
	\tilde{W}_{\pi_1 \pi_2 ; \pi_3  \pi_4}
	\frac{ \varphi_{ \pi_{l_1} \pi_{l_2}, \tau} }{ f_{vx, \pi_{l_1} \pi_{l_2}} }
	K_{x,h} \left( X_{\pi_{5}} - X_{\pi_{l_1}} ,X_{\pi_{6}} - X_{\pi_{l_2}}  \right) 
\end{eqnarray*}
where   
$\sum_{i_1 < \cdots < i_6}$ denotes 
sum over the $\binom{n}{6}$ combinations of $6$ distinct elements  $(i_{1}, \cdots, i_{6})$ from $(1, \cdots, n)$, and 
$\bar{\sigma}\{i_1, \cdots, i_6\}$ is used to denote the 6-tuple $\{i_{1}, \cdots, i_{6}\}$.

$U_n$ is a sixth order U-statistic where the kernel $\phi_{\bar{\sigma}, \tau}$ varies with $n$ as in \citet{powell/stock/stoker:1989}. 
Using Lemma A.3 in \citet{ahn/powell:1993}, it is sufficient to show
$E \left[ \mid\mid  \phi_{\bar{\sigma}, \tau} \mid\mid^2   \right] = o(n)$ to conclude that 
$
U_n 
-
E\left[ 
	\frac{1}{h^{L}}
	\tilde{W}_{i_1i_2;j_1j_2}
	\frac{ \varphi_{ l_1l_2, \tau} }{ f_{vx, l_1l_2} }
	K_{x,h} \left( X_{k_1}-X_{l_1} , X_{k_2}-X_{l_2}  \right)
\right] 
=
o_p(1)$.

A Cauchy-Schwarz inequality can be used to show that the expectation 
\begin{eqnarray*}
	E
	\left[ 
		\mid\mid   
		\frac{1}{6!} 
		\sum_{ \pi \in  \Pi_{6!} } 
		\frac{1}{h^{L}}
		\tilde{W}_{\pi_1 \pi_2 ; \pi_3  \pi_4}
		\frac{ \varphi_{ \pi_{l_1} \pi_{l_2}, \tau} }{ f_{vx, \pi_{l_1} \pi_{l_2}} }   
		K_{x,h} \left( X_{\pi_{5}} - X_{\pi_{l_1}} ,X_{\pi_{6}} - X_{\pi_{l_2}}  \right) 
		\mid\mid^2
	\right]
\end{eqnarray*}
is bounded above by 
\begin{eqnarray*}
		\frac{1}{6! h^{2L}}	
		\sum_{ \pi \in  \Pi_{6!} } 
		E
		\left[  	
			\left(  
				\frac{ \varphi_{ \pi_{l_1} \pi_{l_2}, \tau} }{ f_{vx, \pi_{l_1} \pi_{l_2}} }
			\right)^2 	
		K_{x,h} \left( X_{\pi_{5}} - X_{\pi_{l_1}} ,X_{\pi_{6}} - X_{\pi_{l_2}}  \right)^2
		\tilde{W}_{\pi_1 \pi_2 ; \pi_3  \pi_4}
		\tilde{W}_{\pi_1 \pi_2 ; \pi_3  \pi_4}'
		\right].
\end{eqnarray*}

Let $X_{\bar{\sigma}\{\pi_1, \cdots, \pi_6\}} = \left\{ X_{\pi_1}, \cdots, X_{\pi_6}\right\}$, and 
observe that 
\begin{align*}
&
E 
\left[ 
	\frac{1}{h^{2L}} 
	\left(  
		\frac{ \varphi_{ \pi_{l_1} \pi_{l_2}, \tau} }{ f_{vx, \pi_{l_1} \pi_{l_2}} }
	\right)^2 	
	K_{x,h} \left( X_{\pi_{5}} - X_{\pi_{l_1}} ,X_{\pi_{6}} - X_{\pi_{l_2}}  \right)^2
	\tilde{W}_{\pi_1 \pi_2 ; \pi_3  \pi_4}
	\tilde{W}_{\pi_1 \pi_2 ; \pi_3  \pi_4}'
\right]
\\
=&
E 
\left[ 
E
\left[  
	\left(  
		\frac{ \varphi_{ \pi_{l_1} \pi_{l_2}, \tau} }{ f_{vx, \pi_{l_1} \pi_{l_2}} }
	\right)^2 	
\mid 
X_{\bar{\sigma}\{\pi_1, \cdots, \pi_6\}}
\right]
\frac{1}{h^{2L}}
	K_{x,h} \left( X_{\pi_{5}} - X_{\pi_{l_1}} ,X_{\pi_{6}} - X_{\pi_{l_2}}  \right)^2
	\tilde{W}_{\pi_1 \pi_2 ; \pi_3  \pi_4}
	\tilde{W}_{\pi_1 \pi_2 ; \pi_3  \pi_4}'
\right]
\\
\leq 
&
\frac{1}{h^{2L}}
\sup_{i_1i_2 ; i_3 i_4}\left( \tilde{W}_{i_1i_2 ; i_3 i_4} \right)
\sup_{i_1i_2 ; i_3 i_4}\left( \tilde{W}_{i_1i_2 ; i_3 i_4} \right)'
\\
& \times 
E 
\left[ 
	E
	\left[  
		\left(  
		\frac{ \varphi_{ \pi_{l_1} \pi_{l_2}, \tau} }{ f_{vx, \pi_{l_1} \pi_{l_2}} }
		\right)^2 	
	\mid 
	X_{\pi_{l_1} }, X_{ \pi_{l_2} }
	\right]
		K_{x,h} \left( X_{\pi_{5}} - X_{\pi_{l_1}} ,X_{\pi_{6}} - X_{\pi_{l_2}}  \right)^2
\right]
\\
\leq &
\frac{1}{h^{2L}}
\sup_{i_1i_2 ; i_3 i_4}\left( \tilde{W}_{i_1i_2 ; i_3 i_4} \right)
\sup_{i_1i_2 ; i_3 i_4}\left( \tilde{W}_{i_1i_2 ; i_3 i_4} \right)'
\sup_{ (x,x')\in \mathbb{S}_{x}, \tau \geq 0}
\left(
	E
	\left[  
		\left(  
			\frac{ \varphi_{ i_1 i_2, \tau} }{ f_{vx, l_1 l_2} }
		\right)^2 	
	\mid 
	X_{\pi_{l_1} }, X_{ \pi_{l_2} }
	\right]	
\right)
\\
& \times	
E 
\left[ 
		K_{x,h} \left( X_{\pi_{5}} - X_{\pi_{l_1}} ,X_{\pi_{6}} - X_{\pi_{l_2}}  \right)^2
\right]
\\
=&
O\left( \frac{1}{h^{L}}  \right)
\times 
\int 
K_{x} \left[ \nu_1 , \nu_2 \right]^2
 f(X_{l_1}, X_{l_2}) 
f(X_{l_1}+ \nu_1 h , X_{l_2}+ \nu_2 h ) d X_{l_{1}} dX_{l_{2}} d\nu_1 d\nu_2
\\
=& h^{-L} O(1) = O(n (n h^{L})^{-1} ) = o(n),
\end{align*}
where the first inequality follows from Assumptions \ref{Ass:SR00:iidsampling} and \ref{Ass:Inf01:SamplingMoments}. 
The second inequality follows from Assumption \ref{Ass:Inf0302:Boundedness}.
The second to last equality follows from Assumption \ref{Ass:Inf01:SamplingMoments}, and the change of variables 
$X_{i_{5}} = X_{l_{1}} + \nu_1 h$ and 
$X_{i_{6}} = X_{l_{2}} + \nu_2 h$ with Jacobian $h^{L}$. 
The last equality follows from Assumption \ref{Ass:Inf04:Kernel}.

Consequently, $E \left[ \mid\mid   \phi_{\bar{\sigma}, \tau} \mid\mid^2   \right] = o(n)$ if $n h^{L} \rightarrow \infty$.
Thus, Lemma A.3 in \citet{ahn/powell:1993} implies that 
\begin{eqnarray*}
	U_{n} 
	-  
	E 
	\left[ 
		\frac{1}{h^{L}}
		\tilde{W}_{\pi_1 \pi_2 ; \pi_3  \pi_4}
		\left\{ 
			\frac{ \varphi_{ \pi_{l_1} \pi_{l_2}, \tau} }{ f_{vx, \pi_{l_1} \pi_{l_2}} }	
		\right\}
		K_{x,h} \left( X_{\pi_{5}} - X_{\pi_{l_1}} ,X_{\pi_{6}} - X_{\pi_{l_2}}  \right) 
	\right] 
	&=&
	o_{p}(1)
\end{eqnarray*}
as $n \rightarrow \infty$.

Notice that 
\begin{eqnarray*}
	& &
	E 
	\left[ 
		\frac{1}{h^{L}}
		\tilde{W}_{i_1 i_2 ; i_3  i_4}
		\left\{ 
			\frac{ \varphi_{ l_1 l_2, \tau} }{ f_{vx, l_1 l_2} }	
		\right\}
		K_{x,h} \left( X_{i_{5}} - X_{l_{1}} ,X_{i_{6}} - X_{l_{2}}  \right) 
	\right] 
	\\
	&=&	
	\frac{1}{h^{L}}
	E \left[  
		E \left[ 
			\tilde{W}_{i_1 i_2 ; i_3  i_4}
			\left\{ 
				\frac{ \varphi_{ l_1 l_2, \tau} }{ f_{vx, l_1 l_2} }	
			\right\}
		\mid 
		X_{\bar{\sigma}\{\pi_1, \cdots, \pi_6\}}
		\right] 
	K_{x,h} \left[ X_{i_{5}} - X_{l_{1}}, X_{i_{6}} - X_{l_{2}} \right] 
	\right]  
	\\
	&=&
	\frac{1}{h^{L}}
	E \left[  
		E \left[ 
			\tilde{W}_{i_1 i_2 ; i_3  i_4}
			\left\{ 
				\frac{ \varphi_{ l_1 l_2, \tau} }{ f_{vx, l_1 l_2} }	
			\right\}
		\mid 
		X_{i_{1}}, X_{i_{2}}, X_{i_{3}}, X_{i_{4}} 
		\right] 
	K_{x,h} \left[ X_{i_{5}} - X_{l_{1}}, X_{i_{6}} - X_{l_{2}} \right]
	\right]
\\
&=&
\int 
E \left[ 
	\tilde{W}_{i_1 i_2 ; i_3  i_4}
	\left\{ 
		\frac{ \varphi_{ l_1 l_2, \tau} }{ f_{vx, l_1 l_2} }	
	\right\}
\mid 
X_{i_{1}}, X_{i_{2}}, X_{i_{3}}, X_{i_{4}} 
\right]
\\
&& 
\qquad \qquad \qquad
\times 
K_{x,h} \left[ X_{i_{5}} - X_{l_{1}}, X_{i_{6}} - X_{l_{2}} \right]
f(X_{\bar{\sigma}\{i_1, \cdots, i_6\}}) 
dX_{\bar{\sigma}\{i_1, \cdots, i_6\}}	
\end{eqnarray*}
where the second equality follows from Assumption \ref{Ass:SR00:iidsampling}
Next, consider the change of variables 
$X_{i_{5}} = X_{l_{1}} + h \nu_1 $ 
and 
$X_{i_{6}} = X_{l_{2}} + h \nu_2 $
with Jacobian $h^L$. It then follows that
\begin{eqnarray*}
	&&
	\int   
	E \left[ 
		\tilde{W}_{i_1 i_2 ; i_3  i_4}
		\left\{ 
			\frac{ \varphi_{ l_1 l_2, \tau} }{ f_{vx, l_1 l_2} }	
		\right\}
	\mid 
	X_{i_{1}}, X_{i_{2}}, X_{i_{3}}, X_{i_{4}} 
	\right]
	\\
	&&
	\qquad \qquad \qquad 
	\times 
	K \left[  \nu_1, \nu_2  \right]  
	f(X_{i_{1}}, \cdots, X_{i_{4}}) 
	\left\{ 
		f( X_{l_{1}} + h \nu_1, X_{l_{2}} + h \nu_2 ) 
	\right\} 
	d X_{i_{1}} \cdots, dX_{i_{4}} 
	d \nu_1 d \nu_2.
\end{eqnarray*}	

Assumption \ref{Ass:Inf01:SamplingMoments} guarantees that $f_{x}(\cdot, \cdot)$ is $\overline{M}$-times differentiable with respect to all of its arguments, 
and Assumption \ref{Ass:Inf04:Kernel} ensures that $K_{x}(\cdot, \cdot)$ is a bias-reducing kernel of order $2\overline{M}$.
It follows from an $\overline{M}$-order Taylor expansion
$f( X_{l_{1}} + h \nu_1, X_{l_{2}} + h \nu_2 ) $ 
around 
$f(X_{i_{1}}, X_{i_{3}})$, 
and the properties of the kernel that 
\begin{eqnarray*}
	& &
	\int   
	E \left[ 
		\tilde{W}_{i_1 i_2 ; i_3  i_4}
		\left\{ 
			\frac{ \varphi_{ l_1 l_2, \tau} }{ f_{vx, l_1 l_2} }	
		\right\}
		f_{x, l_{1}l_{2}}  	
	\mid 
	X_{i_{1}}, X_{i_{2}}, X_{i_{3}}, X_{i_{4}} 
	\right]
	f(X_{i_{1}}, \cdots, X_{i_{4}}) 
	d X_{i_{1}} \cdots, X_{i_{4}}  
	+ 
	h^{\overline{M}} 
	O(1) 
	\\
	&=&
	  E 
	  \left[  
		\tilde{W}_{i_1 i_2 ; i_3  i_4}  
		\varphi_{ l_1 l_2, \tau}
		\left\{ 
			\frac{ f_{x, l_{1}l_{2}}  	 }{ f_{vx, l_1 l_2} }	
		\right\}
	  \right]    
	  + 
	  h^{\overline{M}} 
	  O(1) 
	 \\
	 &=&
	 E 
	 	\left[  
		 \tilde{W}_{i_1 i_2 ; i_3  i_4}    
		 D_{ l_1 l_2, \tau}^\ast
		\right]    + o(1). 
\end{eqnarray*} 

The proof is complete.
\end{proof}
\subsection{Lemmas for Asymptotic Normality Theorem}

\subsubsection*{Notation}
The following notation will prove to be useful to show 
Lemmas \ref{Lemma:TechAppx:Expansion}-\ref{Lemma:TechAppx:Hajek_Sum3}.
For any finite $n$, let $\Omega_n = \left\{ X_{n}, A_{n}, v_{n}\right\}$.
Given a fixed tetrad $\sigma\{i_1,i_2, j_1,j_2\} \in \mathcal{N}_{m_n}$,
let 
\begin{eqnarray*}
	X_{\sigma} = \left\{ X_{i_1}, X_{i_2}, X_{j_1}, X_{j_2}\right\}, \quad 
	A_{\sigma} = \left\{ A_{i_1}, A_{i_2}, A_{j_1}, A_{j_2}\right\}, \quad 
	v_{\sigma} = \left\{ v_{i_1}, v_{i_2}, v_{j_1}, v_{j_2}\right\}, \quad 
	\Omega_\sigma = \left\{ X_{\sigma}, A_{\sigma}, v_{\sigma}\right\},
\end{eqnarray*}
and for any dyad $(l_1,l_2) \in \left\{ (i_1, j_1), (i_1, j_2), (i_2, j_1), (i_2, j_2) \right\}$, define 
\begin{eqnarray*}
	\omega_{l_1l_2} &=& \left\{ X_{l_1}, X_{l_2}, A_{l_1}, A_{l_2}, v_{l_1l_2}\right\}	
	\\ 
	T_{l_1l_2}^{\dagger}
	&=&
	T_{l_1l_2}
	-
	E\left[ \tilde{W}_{\sigma} \tilde{D}^\ast_{\sigma,\tau} \mid \Omega_\sigma \right]
\end{eqnarray*}
for any random variable $T_{l_1l_2}$.
\begin{lemma}\label{Lemma:TechAppx:Expansion}
Suppose that the Assumptions in Theorem \ref{Theo:Inf:AN} hold, and consider 
\begin{eqnarray*}
	\widehat{\eta}_{ \left[ l_1l_2 \right], \tau}
	&=&
	\frac{1}{m_{n}} 
	\sum_{ \sigma \in \mathcal{N}_{m_{n}}}
	\tilde{W}_{\sigma}
	\varphi_{ l_1l_2, \tau}
	\left( 
		\frac{ \widehat{f}_{x,l_1l_2} }{ \widehat{f}_{vx,\sigma{l_1 l_2}} } 
	\right).
\end{eqnarray*}
with $(l_1,l_2) \in \left\{ (i_1,j_1), (i_1,j_2), (i_2,j_1), (i_2,j_2)\right\}$.
It follows that $\widehat{\eta}_{ \left[ l_1l_2 \right], \tau}$ can be written as
\begin{align*}
	\widehat{\eta}_{ \left[ l_1l_2 \right], \tau}
	=
	\frac{1}{m_{n}} 
	\sum_{ \sigma \in \mathcal{N}_{m_{n}}} 
	\tilde{W}_{\sigma}
	\varphi_{ l_1l_2, \tau}
	\left\{ 
		\frac{f_{x,l_1l_2} }{ f_{vx,l_1l_2}}
		+
		\frac{ \widehat{f}_{x,l_1l_2} - f_{x,l_1l_2}  }{ f_{vx,l_1l_2} }
		-
		\frac{ f_{x,l_1l_2} }{ f_{vx,l_1l_2} } 
		\times 
		\frac{ \widehat{f}_{vx,l_1l_2} - f_{vx,l_1l_2} }{ f_{vx,l_1l_2} }
	\right\} 
	+ o_{p}(1).
\end{align*}
\end{lemma}

\begin{proof}
Given $h \rightarrow 0$,and $n^{1-\delta}h^{L+1} \rightarrow \infty$ for any $\delta>0$, it follows from a variance calculation argument that
\begin{align*}
	\sup_{ (v,x,x')\in \overline{\Omega}_{v,x} } 
	\mid \hat{f}_{vx}(v,x,x') - f_{vx}(v,x,x') \mid 
	&
	= o_{p}\left(  1 \right)
	\\
	\sup_{(x,x') \in \overline{\Omega}_{x}} 
	\mid \hat{f}_{x}(x,x') - f_{x}(x,x') \mid 
	&
	= 
	o_{p}\left(  1 \right),
\end{align*}
for any $\delta >0$. See, e.g., \citet{silverman:1978}, 
\citet{collomb/hardle:1986},\citet{aradillas2010semiparametric},  and for applications to network models \citet{leung:2015} and \citet{graham/niu/powell:2019}.

Consider a second order Taylor expansion of 
$\widehat{f}_{x,l_1l_2}/\widehat{f}_{vx,l_1l_2}$ around $f_{x,l_1l_2}/f_{vx,l_1l_2}$.
The quadratic terms in the expansion involve second order derivatives of $f_{x,l_1l_2}/f_{vx,l_1l_2}$ evaluated at $\tilde{f}_{x,l_1l_2}$ and $\tilde{f}_{vx,l_1l_2}$, where $\tilde{f}_{x,l_1l_2}$ lies between $\widehat{f}_{x,l_1l_2}$ and $f_{x,l_1l_2}$, 
and similarly $\tilde{f}_{vx,l_1l_2}$ lies between $\widehat{f}_{vx,l_1l_2}$ and $f_{vx,l_1l_2}$.
By substituting a second order Taylor expansion of $\widehat{f}_{x,l_1l_2}/\widehat{f}_{vx,l_1l_2}$ around $f_{x,l_1l_2}/f_{vx,l_1l_2}$ 
into $\widehat{\eta}_{ \left[ l_1l_2 \right], \tau}$, I obtain 
\begin{eqnarray*}
	\widehat{\eta}_{ \left[ l_1l_2 \right], \tau}
	=
	\frac{1}{m_{n}} 
	\sum_{ \sigma \in \mathcal{N}_{m_{n}}} 
	\tilde{W}_{\sigma}
	\varphi_{ l_1l_2, \tau}
	\left\{ 
	\frac{f_{x,l_1l_2}}{f_{vx,l_1l_2}}
	+
	\frac{ \widehat{f}_{x,l_1l_2} - f_{x,l_1l_2}  }{ f_{vx,l_1l_2} }
	-
	\frac{ f_{x,l_1l_2} }{ f_{vx,l_1l_2} } \times 
	\frac{ \widehat{f}_{vx,l_1l_2} - f_{vx,l_1l_2} }{ f_{vx,l_1l_2} }
	\right\} 
	+ R_{n},	
\end{eqnarray*}
where $R_{n}$ denotes the reminder term. The result follows from showing that  $R_n = o_{p}(1)$. 
	
The first component of $R_n$ is 
\begin{eqnarray*}
	& &
	\frac{1}{m_{n}} 
	\sum_{ \sigma \in \mathcal{N}_{m_{n}}} 
	\tilde{W}_{\sigma}
	\varphi_{ l_1l_2, \tau}
	\left\{ 
	\tilde{f}_{x,l_1l_2} 
	\frac{ 
		\left( \widehat{f}_{vx,l_1l_2} - f_{vx,l_1l_2}  \right)^2
		}{ 
		\tilde{f}_{vx,l_1l_2}^{3} 
		}	
	\right\}
	\\
	&\leq& 
	\left[ \sup_{(x,x')\in \overline{\Omega}_{x} } \vert f_{x} \vert \right] 
	\left[ \sup_{(v,x,x')\in \overline{\Omega}_{vx}}\vert f_{vx}^{-3} \vert \right]
	\left[ \sup_{(v,x,x')\in \overline{\Omega}_{vx}} 
	\vert \widehat{f}_{vx} - f_{vx} \vert \right]^2
	\left(
	\frac{1}{m_{n}} 
	\sum_{ \sigma \in \mathcal{N}_{m_{n}}} 
	\mid\mid
	\tilde{W}_{\sigma}
	\varphi_{ l_1l_2, \tau} 
	\mid\mid
	\right)
	\\
	& =&  O_{p}(1) \left[ \sup_{(v,x,x')} \vert \widehat{f}_{vx} - f_{vx} \vert \right]^2
	\\
	&=&
	o_{p}(1).
\end{eqnarray*}
The first inequality follows from Assumption \ref{Ass:Inf01:SamplingMoments}. 
The equality follows from the fact that the V-statistic inside the parenthesis converges to its expectation 
given that Assumptions \ref{Ass:SR00:iidsampling} and  \ref{Ass:Inf01:SamplingMoments}. 
The result follows from the uniform convergence of the kernel estimator. 
	
The remaining component of $R_n$ is 
\begin{eqnarray*}
	& & 
	\frac{1}{m_{n}} 
	\sum_{ \sigma \in \mathcal{N}_{m_{n}}} 
	\tilde{W}_{\sigma}
	\varphi_{ l_1l_2, \tau} 
	\left\{ 
	\frac{ (\widehat{f}_{vx,l_1l_2} - f_{vx,l_1l_2}) (\widehat{f}_{x,l_1l_2} - f_{x,l_1l_2}) }{ f_{vx,l_1l_2}^{2} }	
	\right\}
	\\
	&\leq& 
	\left[ \sup_{(v,x, x) \in \overline{\Omega}_{vx} } \mid f_{vx}^{-2} \mid  \right] 
	\left[ \sup_{(v,x, x) \in \overline{\Omega}_{vx}}  \mid \widehat{f}_{vx} - f_{vx} \mid \right] 
	\left[ \sup_{(x, x) \in \overline{\Omega}_{x}}      \mid \widehat{f}_{x} - f_{x} \mid \right]
	\\
	&& 
	\times 
	\left(
		\frac{1}{m_{n}} 
		\sum_{ \sigma \in \mathcal{N}_{m_{n}}} 
		\mid\mid  
		\tilde{W}_{\sigma}
		\varphi_{ l_1l_2, \tau} 
		\mid\mid 
	\right)
	\\
	& =& 
	O_{p}(1) 
	\left[ \sup_{(v,x, x ) \in \overline{\Omega}_{vx}} 	\mid \widehat{f}_{vx} - f_{vx} 	\mid \right] 
	\left[ \sup_{(x, x) \in \overline{\Omega}_{vx}} 	\mid \widehat{f}_{x} - f_{x} 	\mid \right].
	\\
	&=&
	o_{p}(1).
\end{eqnarray*}
The result follows from the uniform convergence of the kernel estimators. This completes the proof.

\end{proof}

\begin{lemma}\label{Lemma:TechAppx:Hajek_Sum1}
Under the same Assumptions of Theorem \ref{Theo:Inf:AN}, it follows that the H\'{a}jek projection of 
\begin{eqnarray*}
S_{1,n \tau}^\dagger
&=&	
S_{1,n \tau}
-
E 
\left[ 
	\tilde{W}_{\sigma}
	\tilde{D}^\ast_{\sigma,\tau} 
\mid 
\Omega_{n}
\right]
\\
&=&
\frac{1}{m_{n}} 
\sum_{ \sigma \in \mathcal{N}_{m_{n}}}
\left\{ 
	\tilde{W}_{\sigma}
	\tilde{D}^\ast_{\sigma,\tau} 
	-
	E 
	\left[ 
		\tilde{W}_{\sigma}
		\tilde{D}^\ast_{\sigma,\tau} 
	\mid 
	\Omega_{\sigma}
	\right]
\right\} 	
\end{eqnarray*}
into an arbitrary function $\zeta_{i_1j_1} = \left( X_{i_1},X_{j_1}, A_{i_1},A_{j_1}, v_{i_1j_1},U_{i_1j_1} \right)$ is given by 
\begin{eqnarray*}
	V_{1,n\tau}^\ast
	&=& 
	\frac{1}{n(n-1)} 
	\sum_{i_1 = 1}^{n} \sum_{j_1 \neq i_1}
	\xi_{i_1j_1,\tau}
\end{eqnarray*}
and 
\begin{eqnarray*}
	n(n-1)
	\Upsilon_{n}^{-1/2}
	E 
	\left[  
		\left(  
			S_{1,n \tau}^\dagger - 	V_{1,n\tau}^\ast
		\right)^2
	\right]
	\Upsilon_{n}^{-1/2}
	= o(1),
\end{eqnarray*}
where $\Upsilon_{n} = n(n-1)Var(V_{1,n\tau}^\ast)$
and $Var(V_{1,n\tau}^\ast) = O_{p}(p_n^2 \tau^2)$.
\end{lemma}

\begin{proof}

\textbf{Step 1. H\'{a}jek Projection}

Consider the tetrad $\sigma\{i_1, i_2, j_1, j_2\}$, let 
\begin{eqnarray*}
	s
	\left( 
		\sigma\{i_1, i_2, j_1, j_2\} 
	\right)
	&=&
	\tilde{W}_{\sigma}
	\tilde{D}^\ast_{\sigma,\tau}
	-
	E
	\left[  
		\tilde{W}_{\sigma}
		\tilde{D}^\ast_{\sigma,\tau}
		\mid 
		\Omega_\sigma
	\right]
	\\
	&=&
	\tilde{W}_{\sigma}
	\left\{ 
	\tilde{D}^\ast_{\sigma,\tau}
	-
	E
	\left[
		\tilde{D}^\ast_{\sigma,\tau}
		\mid 
		\Omega_\sigma
	\right] 
	\right\},
\end{eqnarray*}
and notice that 
\begin{eqnarray*}
E
\left[  
	s
	\left( 
		\sigma\{i_1, i_2, j_1, j_2\} 
	\right)
	\mid 
	\zeta_{i_1j_1}
\right]
&=&
E
\left[ 
	\tilde{W}_{\sigma}
	\left\{ 
	\tilde{D}^\ast_{\sigma,\tau}
	-
	E
	\left[
		\tilde{D}^\ast_{\sigma,\tau}
		\mid 
		\Omega_\sigma
	\right] 
	\right\}
	\mid
	\zeta_{i_1j_1}
\right]
\\
&=&
\left\{ 
	D^\ast_{i_1j_1,\tau}
	-
	E
	\left[ 
		D^\ast_{i_1j_1,\tau}
		\mid 
		\omega_{i_1j_1}
	\right]
\right\}
E
\left[ 
	\tilde{W}_{\sigma}
	\mid 
	X_{i_1j_1}
\right].
\end{eqnarray*}
where the second equality follows from the Law of Iterated Expectations, 
and Assumptions \ref{Ass:SR00:iidsampling} and \ref{Ass:SR01:distr}. 
To be precise, observe that for $\{l_1,l_2\} \neq \{i_1,j_1\}$ with 
$
	(l_1,l_2) \in \left\{ (i_1, j_2), (i_2, j_1), (i_2, j_2) \right\},
$
\begin{eqnarray*}
&&
	E
\left[ 
	\tilde{W}_{\sigma}
	\left\{ 
		\tilde{D}^\ast_{l_1l_2,\tau}
		-
		E
		\left[  
		\tilde{D}^\ast_{l_1l_2,\tau}
		\mid 
		\Omega_{\sigma}
		\right]
	\right\}	
\mid 
\zeta_{i_1j_1}
\right]	
\\
&=&
E
\left[ 
	\tilde{W}_{\sigma}
	\left\{ 
		E
		\left[  
		\tilde{D}^\ast_{l_1l_2,\tau}
		\mid 
		\omega_{l_1l_2}
		\right]
		-
		E
		\left[  
		\tilde{D}^\ast_{l_1l_2,\tau}
		\mid 
		\omega_{l_1l_2}
		\right]
	\right\}	
\mid 
\zeta_{i_1j_1}
\right]	
\\
&=&
0.
\end{eqnarray*}

It then follows that the H\'{a}jek projection is given by 
\begin{eqnarray*}
	V_{1,n\tau}^\ast
	&=&
	\frac{1}{n(n-1)}
	\sum_{i_1 =1}^n  \sum_{j_1 \neq i_1} 
	\xi_{i_1j_1, \tau},
\end{eqnarray*}
with
\begin{eqnarray*}
\xi_{i_1j_1, \tau}
&=&
\left\{ 
	D^\ast_{ i_1 j_1} 
	-
	E
	\left[ 
		D^\ast_{i_1j_1}
		\mid 
		\omega_{i_1j_1}	
	\right]
\right\}
I_{\tau, i_1j_1} 
\overline{\chi}_{i_1j_1}
\\
\overline{\chi}_{i_1j_1}
&=&
\left\{ 
	\frac{1}{(n-2)(n-3)}
	\sum_{i_2 \neq i_1 , j_1} \sum_{j_2 \neq i_1 , j_1, i_2} 
		E\left[ 
			\tilde{W}_{\sigma\{i_1,i_2;j_1,j_2\}}
			\mid 	
			X_{i_1}, X_{j_1}
		\right]
\right\}.
\end{eqnarray*}

Notice that $E \left[ V_{1,n\tau}^\ast  \right] = E\left[\xi_{i_1j_1, \tau}  \right] =0$.

\textbf{Step 2. Variance of H\'{a}jek Projection}

For two different dyads  $\{i_1,j_1\} \neq \{i_1',j_1'\}$ with zero common indices, 
Assumption \ref{Ass:SR00:iidsampling} implies that
\begin{eqnarray*}
	E
	\left[ 
		\xi_{i_1j_1, \tau}
		\xi_{i_1'j_1', \tau}
	\right]
	&=&
	E
	\left[ 
		\xi_{i_1j_1, \tau}
	\right]
	E
	\left[ 
		\xi_{i_1'j_1', \tau}
	\right]
	=0.
\end{eqnarray*}

Observe that for two dyads  $\{i_1,j_1\} \neq \{i_1,j_1'\}$  with one common index, 
the conditionally independent formation of links implied by Assumption \ref{Ass:SR01:distr} yields 
\begin{eqnarray*}
	E
	\left[ 
		\xi_{i_1j_1, \tau}
		\xi_{i_1'j_1', \tau}
	\right]
	&=&
	E
	\left[ 
		E\left[ 
			\xi_{i_1j_1, \tau}
			\mid 
			\Omega_{n}
		\right]
		E\left[ 
			\xi_{i_1'j_1', \tau}
			\mid 
			\Omega_{n}
		\right]
	\right]
	=0.
\end{eqnarray*}

Therefore, the variance of $ V_{1,n \tau}^\ast$ is given by 
\begin{eqnarray*}
Var\left( V_{1,n \tau}^\ast \right)
&=&
\left\{ \frac{1}{n(n-1)}\right\}^2
\left\{ 	
\sum_{i_1 =1}^n  \sum_{j_1 \neq i_1} 
	E\left[ 
	\xi_{i_1j_1, \tau}
	\xi_{i_1'j_1', \tau}'
\right] 
\right\}
\\
&=&
\left\{ \frac{1}{n(n-1)}\right\}^2
\left\{ 
\sum_{i_1 =1}^n  \sum_{j_1 \neq i_1} 
\Lambda^\ast_{i_1, j_1}
\right\}
\end{eqnarray*}
where 
\begin{eqnarray*}
\Lambda^\ast_{i_1, j_1}
&=&
E
\left[ 
	\left\{ 
		E
		\left[ 
			D^\ast_{i_1j_1}
			D^\ast_{i_1j_1}
			\mid 
			\omega_{i_1j_1}
		\right]	
		-
		E
		\left[ 
			D^\ast_{i_1j_1}
			\mid 
			\omega_{i_1j_1}
		\right]^2		
	\right\}
	I_{\tau, i_1j_1}^2
	\overline{\chi}_{i_1j_1}
	\overline{\chi}_{i_1j_1}'
\right].
\end{eqnarray*}

Define 
\begin{eqnarray*}
	\Upsilon_{n,\tau} &=& n(n-1) Var\left( V_{1,n \tau}^\ast \right)
	=
	\frac{1}{n(n-1)}
	\left\{ 
	\sum_{i_1 =1}^n  \sum_{j_1 \neq i_1} 
	\Lambda^\ast_{i_1, j_1}
	\right\}.
\end{eqnarray*}

\textbf{Step 3. Variance of $S_{1,n \tau}^\dagger$}

Given two different tetrads $\sigma\{i_1, i_2, j_1, j_2\}$ and $\sigma'\{i_1', i_2' , j_1' , j_2'\}$, let 
\begin{eqnarray*}
	\Delta_{c,n} 
	= 
	Cov
	\left( 
		s \left( \sigma\{i_1, i_2, j_1, j_2\} \right), 
		s \left( \sigma'\{i_1', i_2' , j_1' , j_2'\} \right)
	\right)	
\end{eqnarray*}
denote the covariance between 
$s(\sigma)$ 
and 
$s(\sigma')$ 
when 
$\sigma\{i_1, i_2, j_1, j_2\}$ and $\sigma'\{i_1', i_2' , j_1' , j_2'\}$ have $c=0,1,2,3,4$ indices in common. 

It follows from the conditionally independent formation of links, implied by Assumption \ref{Ass:SR01:distr}, 
and the conditional mean zero, 
$E
\left[ 
	s
	\left( 
		\sigma\{i_1, i_2, j_1, j_2\} 
	\right)	
\mid \Omega_\sigma
\right]=0$,  
that $\Delta_{0,n}= \Delta_{1,n} =0$. 

Consider 
\begin{eqnarray*}
	\Delta_{2,n}
	&=&
	E\left[ 
		s(\sigma\{i_1,i_2 ,j_1,j_2\} )	
		s(\sigma'\{i_1,i_2',j_1,j_2'\} )'	
	\right]
	\\
	&=&
	E
	\left[ 
	\left\{ 
		\tilde{D}^\ast_{\sigma,\tau}
		-
		E
		\left[ 
			\tilde{D}^\ast_{\sigma,\tau} 
			\mid 
			\Omega_{\sigma}
		\right]
	\right\}
	\left\{ 
		\tilde{D}^\ast_{\sigma',\tau}
		-
		E
		\left[ 
			\tilde{D}^\ast_{\sigma',\tau} 
			\mid 
			\Omega_{\sigma'}
		\right]
	\right\}
	\tilde{W}_{\sigma}
	\tilde{W}_{\sigma'}
	\right]
	\\
	&=&
	E
	\left[ 
	\left\{ 
		E
		\left[ 
			\tilde{D}^\ast_{i_1j_1,\tau} 
			\tilde{D}^\ast_{i_1j_1,\tau} 
			\mid 
			\omega_{i_1j_1}
		\right]	
		-
		E
		\left[ 
			\tilde{D}^\ast_{i_1j_1,\tau} 
			\mid 
			\omega_{i_1j_1}
		\right]^2
	\right\}
	I_{\tau, i_1j_1}^2
	\tilde{W}_{\sigma}
	\tilde{W}_{\sigma'}
	\right].
\end{eqnarray*}

It follows from the results above that 
$Var\left( S_{1,nt}^\dagger \right)$ can be expanded as
\begin{eqnarray*}
	&&
	Var\left( S_{1,nt}^\dagger \right)
	\\
	&=&
	\left(  \frac{1}{m_{n}} \right)^2
	\sum_{ \sigma \in \mathcal{N}_{m_{n}}} 
	\sum_{ \sigma' \in \mathcal{N}_{m_{n}}} 
	\left\{ 
		E\left[ 
			s(\sigma\{i_1,i_2 ,j_1,j_2\} )	
			s(\sigma'\{i_1,i_2',j_1',j_2'\} )'	
		\right]
	\right\}
	\\
	&=& 
	\left(  \frac{1}{m_{n}} \right)^2
	\sum_{ i_1=1}^{n} 
	\sum_{ j_1 \neq i_1} 
	\left\{ 
	\sum_{k_1 \neq i_1 , j_1} \sum_{k_2 \neq i_1 , j_1, k_1}
	\sum_{l_1 \neq i_1 , j_1} \sum_{l_2 \neq i_1 , j_1, l_1}
	\Delta_{2,n}
	\right\}
	+
	O\left(\frac{\Delta_{3,n}}{n^3}\right)
	+
	O\left(\frac{\Delta_{4,n}}{n^4}\right).
\end{eqnarray*}

Notice that the term inside the brackets scaled by $\left[ (n-2)(n-3) \right]^{-2}$ is equivalent to 
$\Lambda_{i_1j_1}^\ast$, in particular, 
\begin{eqnarray*}
	\Lambda_{i_1j_1}^\ast
	&=&
	\left\{ 
		\frac{1}{(n-2)(n-3)}
	\right\}^2
	\sum_{k_1 \neq i_1 , j_1} \sum_{k_2 \neq i_1 , j_1, k_1}
	\sum_{l_1 \neq i_1 , j_1'} \sum_{l_2 \neq i_1 , j_1', l_1}
	\Delta_{2,n}
	\\
	&=&
	E
	\left[ 
		\left\{ 
		E
		\left[ 
			\tilde{D}^\ast_{i_1j_1,\tau} 
			\tilde{D}^\ast_{i_1j_1,\tau} 
			\mid 
			\omega_{i_1j_1}
		\right]	
		-
		E
		\left[ 
			\tilde{D}^\ast_{i_1j_1,\tau} 
			\mid 
			\omega_{i_1j_1}
		\right]^2
		\right\}
		I_{\tau, i_1j_1}^2
		\overline{\chi}_{i_1j_1}
		\overline{\chi}_{i_1j_1}'
	\right],
\end{eqnarray*}
which follows from the definition of $\overline{\chi}_{i_1j_1}$. 

Hence, 
\begin{eqnarray*}
	Var\left( S_{1,nt}^\dagger \right)
	&=&
	\left(  \frac{1}{n(n-1)} \right)^2
	\left\{ 
	\sum_{ i_1=1}^{n} \sum_{j_1\neq i_1}
		\Lambda^\ast_{i_1,j_1} 
	\right\}
	+ o(1),
\end{eqnarray*}
and $Var \left( V_{1,n \tau}^\ast \right) - Var \left(  S_{1,n \tau}^\dagger \right) = o(1)$.

\textbf{Step 4. Asymptotic Equivalence}

To show that 
\begin{eqnarray*}
	n(n-1)
	\Upsilon_{n,\tau}^{-1/2}
	E 
	\left[  
		\left(  
			S_{1,n \tau}^\dagger - 	V_{1,n \tau}^\ast
		\right)
		\left(  
			S_{1,n \tau}^\dagger - 	V_{1,n \tau}^\ast
		\right)'
	\right]
	\Upsilon_{n,\tau}^{-1/2}
	&=&  o(1) 
\end{eqnarray*}
it is sufficient to prove that $Var \left( V_{1,n \tau}^\ast \right)^{-1/2}	Cov \left[ V_{1,n \tau}^\ast, S_{1,n \tau}\right] Var \left( V_{1,n \tau}^\ast \right)^{-1/2} = I$, 
which in turn, follows from noticing that 
\begin{eqnarray*}
	Cov \left[ V_{1,n \tau}^\ast, S_{1,n \tau}^\dagger \right]
	&=&
	E \left[ V_{1,n \tau}^\ast, S_{1,n \tau}^\dagger \right]
	\\
	&=&
	E \left[ V_{1,n \tau}^\ast \left( S_{1,n \tau}^\dagger - V_{1,n \tau}^\ast  \right)'\right]
	+ 
	E \left[ V_{1,n \tau}^\ast \left(  V_{1,n \tau}^\ast  \right)'\right]
	\\
	&=&
	Var(V_{1,n \tau}^\ast),
\end{eqnarray*}
since by construction of the orthogonal projection  
\[
	E \left[ V_{1,n \tau}^\ast \left( S_{1,n \tau} - V_{1,n \tau}^\ast  \right)'\right] =0.	
\]

The proof is complete.
\end{proof}
\begin{lemma}\label{Lemma:TechAppx:Hajek_Sum2}
Under the same Assumptions of Theorem \ref{Theo:Inf:AN}, it follows that the H\'{a}jek projection of 
\begin{eqnarray*}
S_{2,n \tau}^\dagger
&=&
S_{2,n \tau}
-
E 
\left[ 
	\tilde{W}_{\sigma}
	\tilde{D}^\ast_{\sigma,\tau} 
\mid 
\Omega_{\sigma}	
\right]
\\
S_{2,n \tau}
&=&
\frac{1}{m_{n}} 
\sum_{ \sigma \in \mathcal{N}_{m_{n}}} 
\tilde{W}_{\sigma}
\left\{ 
	\left( 
		\frac{ 
			\varphi_{ {i_1 j_1}, \tau}	
			\widehat{f}_{x,{i_1 j_1}} 
		}{ 
			f_{vx,{i_1 j_1}} 
		}
		-
		\frac{ 
			\varphi_{ {i_1 j_2}, \tau}	
			\widehat{f}_{x,{i_1 j_2}} 
			}{ 
			f_{vx,{i_1 j_2}} 
			}
	\right)
	-
	\left( 
		\frac{ 
			\varphi_{ {i_2 j_1}, \tau}
			\widehat{f}_{x,{i_2 j_1}} 
			}{ 
			f_{vx,{i_2 j_1}} 
		}
		-
		\frac{ 
			\varphi_{ {i_2 j_2}, \tau}	
			\widehat{f}_{x,{i_2 j_2}} 
		}{ 
			f_{vx,{i_2 j_2}} 
		}
	\right)
\right\}
\end{eqnarray*}
into an arbitrary function $\zeta_{i_1j_1} = \left( X_{i_1},X_{j_1}, A_{i_1},A_{j_1}, v_{i_1j_1},U_{i_1j_1} \right)$ is given by 
\begin{eqnarray*}
V_{2,n\tau}^\ast
&=& 
\frac{1}{n(n-1)} 
\sum_{i_1 = 1}^{n} \sum_{j_1 \neq i_1}
\bar{\xi}_{i_1j_1, \tau}
\end{eqnarray*}
and 
\begin{eqnarray*}
n
\Upsilon_{n}^{-1/2}
E 
\left[  
	\left(  
		S_{2,n \tau} - 	V_{2,n\tau}^\ast
	\right)^2
\right]
\Upsilon_{n}^{-1/2}
= o(1),
\end{eqnarray*}
where $\Upsilon_{n} = n Var(V_{2,n\tau}^\ast)$.
\end{lemma}

\begin{proof}

Similarly to the definition for tetrads, I introduce the function $\overline{\sigma}=\overline{\sigma}\{i_1,i_2,j_1,j_2,k_1,k_2\}$ that maps each unique 6-tuple $\{i_1,i_2,j_1,j_2,k_1,k_2\}$ into an index set $N_{\overline{m}_n} = \{1, \cdots, \overline{m}_n\}$ where $\overline{m}_n$ denotes the total number of those 6-tuples.
Hence, each distinct 6-tuple $\{i_1,i_2,j_1,j_2,k_1,k_2\}$ corresponds to a unique 
$\overline{\sigma} =  \overline{\sigma}\{i_1,i_2,j_1,j_2,k_1,k_2\} \in N_{\overline{m}_n}$. 

Consider a fixed 6-tuple $\{i_1,i_2,j_1,j_2,k_1,k_2\}$, and define  
\begin{eqnarray*}
s_{i_1, j_1}(\overline{\sigma})
&=&
\tilde{W}_{i_1i_2,j_1j_2}
\left\{ 
		\frac{1}{h^L}
		\frac{ 
			\varphi_{ i_1 j_1, \tau}	
		}{ 
			f_{vx,i_1 j_1} 
		}
		K_{x,h}\left( X_{k_1}-X_{i_1}, X_{k_2}-X_{j_1}\right) 
		-
		E\left[ 
			D_{i_1j_1,\tau}^\ast 
			\mid 
			\Omega_{i_1i_2,j_1j_2}
		\right]
\right\}		
\\
s_{i_1,j_2}(\overline{\sigma}) 
&=&
\tilde{W}_{i_1i_2,j_1j_2}
\left\{ 
		\frac{1}{h^L}
		\frac{ 
		\varphi_{ i_1 j_2, \tau}	
		}{ 
			f_{vx,i_1 j_2} 
		}
		K_{x,h}\left( X_{k_1}-X_{i_1}, X_{k_2}-X_{j_2}\right) 	
		-
		E\left[  
			D_{i_1j_2,\tau}^\ast 
			\mid 
			\Omega_{i_1i_2,j_1j_2}
		\right]
\right\}
\\
s_{i_2,j_1}(\overline{\sigma})
&=&
\tilde{W}_{i_1i_2,j_1j_2}
\left\{ 
		\frac{1}{h^L}
		\frac{ 
			\varphi_{ i_2 j_1, \tau}
			}{ 
				f_{vx,i_2 j_1} 	 
		}
		K_{x,h}\left( X_{k_1}-X_{i_2}, X_{k_2}-X_{j_1}\right)
		-
		E
		\left[ 
			D_{i_2j_1,\tau}^\ast 
		\mid 
		\Omega_{i_1i_2,j_1j_2}
		\right]
\right\}		
\\
s_{i_2, j_2}\overline{\sigma})
&=&
\tilde{W}_{i_1i_2,j_1j_2}
\left\{ 
	\frac{1}{h^L}
	\frac{ 
		\varphi_{ i_2 j_2, \tau}	
	}{ 
		f_{vx,i_2 j_2} 
	}
	K_{x,h}\left( X_{k_1}-X_{i_2}, X_{k_2}-X_{j_2}\right)
	-
	E\left[ 
		D_{i_2j_2,\tau}^\ast 
		\mid 
		\Omega_{i_1i_2,j_1j_2}	
	\right]
\right\},
\end{eqnarray*}
and $s_{2, n}(\overline{\sigma})= 	s_{i_1, j_1}(\overline{\sigma})-s_{i_1,j_2}(\overline{\sigma})- s_{i_2,j_1}(\overline{\sigma})+	s_{i_2,j_2}(\overline{\sigma})$.
It follows then that $S_{2,n \tau}^\dagger$ can be written as 
\begin{eqnarray*}
S_{2,n \tau}^\dagger
&=&
\left[ 
	6!\binom{n}{6}
\right]^{-1}
\sum_{\overline{\sigma} \in \mathcal{N}_{\overline{m}_n}} 
s_{2, n\tau}(\overline{\sigma})
\\
&=&
\left[ 
	6!\binom{n}{6}
\right]^{-1}
\sum_{\overline{\sigma} \in \mathcal{N}_{\overline{m}_n}} 
\left\{ 
	s_{i_1j_1}(\overline{\sigma})
	-
	s_{i_1j_2}(\overline{\sigma})
	-
	s_{i_2j_1}(\overline{\sigma})
	+
	s_{i_2j_2}(\overline{\sigma})
\right\}.
\end{eqnarray*}

\textbf{Step 1. H\'{a}jek Projection}

The rest of the proof makes use of the following index notation for dyads.
Given the total number of ordered dyads $\overline{\textbf{n}} = n(n-1)$, 
let the boldface indices $\boldsymbol{\pi} = \textbf{1}, \textbf{2}, \cdots$ index the $\overline{\textbf{n}}$ ordered dyads in the sample.
In an abuse of notation, also let $\boldsymbol{\pi}$ denote the set $\left\{ i_1,j_1\right\}$, where $i_1$ and $j_1$ are the indices that comprise dyad $\boldsymbol{\pi}$.
In particular, $\boldsymbol{\pi}(1)=i_1$ and $\boldsymbol{\pi}(2)=j_1$, when $\boldsymbol{\pi}=\left\{ i_1,j_1\right\}$. 

With this notation at hand, $S_{2,n \tau}^\dagger$ can be expressed as 
\begin{eqnarray*}
S_{2,n \tau}^\dagger
&=&
\left[ 
	6!\binom{n}{6}
\right]^{-1}
\sum_{\boldsymbol{\pi}_1=\textbf{1}}^{\overline{\textbf{n}}}
\sum_{\boldsymbol{\pi}_2 \neq \boldsymbol{\pi}_1 }
\sum_{\boldsymbol{\pi}_3 \neq \boldsymbol{\pi}_1 }
\left\{ 
	s_{\boldsymbol{\pi}_1}(\overline{\sigma})
	-
	s_{\boldsymbol{\pi}_1(1)\boldsymbol{\pi}_2(2)}(\overline{\sigma})
	-
	s_{\boldsymbol{\pi}_2(1)\boldsymbol{\pi}_1(2)}(\overline{\sigma})
	+
	s_{\boldsymbol{\pi}_2}(\overline{\sigma})
\right\}
\end{eqnarray*}
where $\overline{\sigma}=\overline{\sigma}\left\{\boldsymbol{\pi}_1, \boldsymbol{\pi}_2, \boldsymbol{\pi}_3 \right\}$.

Let 
\begin{eqnarray*}
p_{\boldsymbol{\pi}_1, \boldsymbol{\pi}_3} 
\left( \overline{\sigma} \right)
&=& 
\frac{1}{h^L}
\left( 
		\frac{ 
			\varphi_{ \boldsymbol{\pi}_1, \tau}	
		}{ 
			f_{vx,\boldsymbol{\pi}_1} 
		}
		\tilde{W}_{\boldsymbol{\pi}_1,\boldsymbol{\pi}_2}
		+
		\frac{ 
			\varphi_{ \boldsymbol{\pi}_3, \tau}	
		}{ 
			f_{vx,\boldsymbol{\pi}_3} 
		}
		\tilde{W}_{\boldsymbol{\pi}_3,\boldsymbol{\pi}_2}
\right)
K_{x,h}\left( X_{\boldsymbol{\pi}_3}-X_{\boldsymbol{\pi}_1}\right) 
\\
&&
-
E\left[ 
	\tilde{W}_{\boldsymbol{\pi}_1,\boldsymbol{\pi}_2}
	D_{\boldsymbol{\pi}_1,\tau}^\ast 
	\mid 
	\Omega_{\boldsymbol{\pi}_1, \boldsymbol{\pi}_2}
\right]
-
E\left[ 
	\tilde{W}_{\boldsymbol{\pi}_3,\boldsymbol{\pi}_2}
	D_{\boldsymbol{\pi}_3,\tau}^\ast 
	\mid 
	\Omega_{\boldsymbol{\pi}_3, \boldsymbol{\pi}_2}
\right]	 
\\
p_{\boldsymbol{\pi}_2, \boldsymbol{\pi}_3} 
\left( \overline{\sigma} \right)
&=& 
\frac{1}{h^L}
\left( 
	\frac{ 
		\varphi_{ \boldsymbol{\pi}_2, \tau}	
	}{ 
		f_{vx,\boldsymbol{\pi}_2} 
	}
	\tilde{W}_{\boldsymbol{\pi}_1,\boldsymbol{\pi}_2}
	K_{x,h}\left( X_{\boldsymbol{\pi}_3}-X_{\boldsymbol{\pi}_2}\right) 
	+
	\frac{ 
		\varphi_{ \boldsymbol{\pi}_2, \tau}	
	}{ 
		f_{vx,\boldsymbol{\pi}_2} 
	}
	\tilde{W}_{\boldsymbol{\pi}_3,\boldsymbol{\pi}_2}
	K_{x,h}\left( X_{\boldsymbol{\pi}_1}-X_{\boldsymbol{\pi}_2}\right) 
\right)
\\
&&
-
E\left[ 
\tilde{W}_{\boldsymbol{\pi}_1,\boldsymbol{\pi}_2}
D_{\boldsymbol{\pi}_2,\tau}^\ast 
\mid 
\Omega_{\boldsymbol{\pi}_1, \boldsymbol{\pi}_2}
\right]
-
E\left[ 
\tilde{W}_{\boldsymbol{\pi}_3,\boldsymbol{\pi}_2}
D_{\boldsymbol{\pi}_2,\tau}^\ast 
\mid 
\Omega_{\boldsymbol{\pi}_3, \boldsymbol{\pi}_2}
\right]	
\\	
p_{\boldsymbol{\pi}_1(1)\boldsymbol{\pi}_2(2), \boldsymbol{\pi}_3} 
\left( \overline{\sigma} \right)
&=& 
\frac{1}{h^L}
\tilde{W}_{\boldsymbol{\pi}_1,\boldsymbol{\pi}_2}
\left\{ 
\left( 
	\frac{ 
		\varphi_{ \boldsymbol{\pi}_1(1)\boldsymbol{\pi}_2(2), \tau}	
	}{ 
		f_{vx,\boldsymbol{\pi}_1(1)\boldsymbol{\pi}_2(2)} 
	}
\right)		
K_{x,h}\left( X_{\boldsymbol{\pi}_3}-X_{\boldsymbol{\pi}_1(1)\boldsymbol{\pi}_2(2)}\right) 
-
E\left[ 
	D_{\boldsymbol{\pi}_1(1)\boldsymbol{\pi}_2(2),\tau}^\ast 
	\mid 
	\Omega_{\boldsymbol{\pi}_1, \boldsymbol{\pi}_2}
\right]	
\right\}
\\
&&
+
\frac{1}{h^L}
\tilde{W}_{\boldsymbol{\pi}_3,\boldsymbol{\pi}_2}
\left\{ 
\left(		
	\frac{ 
		\varphi_{ \boldsymbol{\pi}_3(1)\boldsymbol{\pi}_2(2), \tau}	
	}{ 
		f_{vx,\boldsymbol{\pi}_3(1)\boldsymbol{\pi}_2(2)} 
	}
	K_{x,h}\left( X_{\boldsymbol{\pi}_1}-X_{\boldsymbol{\pi}_3(1)\boldsymbol{\pi}_2(2)}\right) 
\right)
-
E\left[ 
D_{\boldsymbol{\pi}_3(1)\boldsymbol{\pi}_2(2),\tau}^\ast 
\mid 
\Omega_{\boldsymbol{\pi}_3, \boldsymbol{\pi}_2}
\right]	
\right\}
\\
p_{\boldsymbol{\pi}_2(1)\boldsymbol{\pi}_1(2), \boldsymbol{\pi}_3} 
\left( \overline{\sigma} \right)
&=& 
\frac{1}{h^L}	
\tilde{W}_{\boldsymbol{\pi}_1,\boldsymbol{\pi}_2}
\left\{ 
\left( 
	\frac{ 
		\varphi_{ \boldsymbol{\pi}_2(1)\boldsymbol{\pi}_1(2), \tau}	
	}{ 
		f_{ \boldsymbol{\pi}_2(1)\boldsymbol{\pi}_1(2), \tau}	
	}
	K_{x,h}\left( X_{\boldsymbol{\pi}_3}-X_{\boldsymbol{\pi}_2(1)\boldsymbol{\pi}_1(2)}\right) 
\right)
-
E\left[ 
D_{\boldsymbol{\pi}_2(1)\boldsymbol{\pi}_1(2),\tau}^\ast 
\mid 
\Omega_{\boldsymbol{\pi}_1, \boldsymbol{\pi}_2}
\right]
\right\}
\\
&&
+
\frac{1}{h^L}	
\tilde{W}_{\boldsymbol{\pi}_3,\boldsymbol{\pi}_2}
\left\{ 
\left(			
\frac{ 
	\varphi_{ \boldsymbol{\pi}_2(1)\boldsymbol{\pi}_3(2), \tau}	
}{ 
	f_{vx,\boldsymbol{\pi}_2(1)\boldsymbol{\pi}_3(2)} 
}
K_{x,h}\left( X_{\boldsymbol{\pi}_1}-X_{\boldsymbol{\pi}_2(1)\boldsymbol{\pi}_3(2)}\right) 
\right)
-
E\left[ 
D_{\boldsymbol{\pi}_2(1)\boldsymbol{\pi}_3(2),\tau}^\ast 
\mid 
\Omega_{\boldsymbol{\pi}_3, \boldsymbol{\pi}_2}
\right]	
\right\}
\\		
\end{eqnarray*}
where
$K_{x,h}\left( X_{\boldsymbol{\pi}_3}-X_{\boldsymbol{\pi}_1}\right)$ 	
denotes 
$K_{x,h}
\left( 
X_{\boldsymbol{\pi}_3(1)} - X_{\boldsymbol{\pi}_1(1)}, 
X_{\boldsymbol{\pi}_3(2)} - X_{\boldsymbol{\pi}_1(2)}
\right)$, 
$\tilde{W}_{\boldsymbol{\pi}_1, \boldsymbol{\pi}_2 }$ 
denotes 
$\tilde{W}_{\boldsymbol{\pi}_1\{i_1i_2\},\boldsymbol{\pi}_2\{j_1j_2\} }$, 
and
\begin{eqnarray*}
\chi_{\boldsymbol{\pi}_1}
&=&
E
\left[ 
\tilde{W}_{\boldsymbol{\pi}_1, \boldsymbol{\pi}_2 }
\mid 
X_{\boldsymbol{\pi}_1}
\right]
\\
\chi_{\boldsymbol{\pi}_1}
&=&
\sum_{\boldsymbol{\pi}_2 \neq \boldsymbol{\pi}_1, \boldsymbol{\pi}_3 } 
\chi_{\boldsymbol{\pi}_1}.
\end{eqnarray*}

Using the symmetry of the kernel,it follows that $	S_{2,n \tau}^\dagger$ can be written as
\begin{align*}
\left[ 
	6!\binom{n}{6}
\right]^{-1}
\sum_{\boldsymbol{\pi}_1=\textbf{1}}^{\overline{\textbf{n}}}
\sum_{\boldsymbol{\pi}_3 = \boldsymbol{\pi}_1 + 1 }
\sum_{\boldsymbol{\pi}_2 \neq \boldsymbol{\pi}_1, \boldsymbol{\pi}_3 } 
\left\{ 
	p_{\boldsymbol{\pi}_1, \boldsymbol{\pi}_3}\left( \overline{\sigma} \right) 
	-
	p_{\boldsymbol{\pi}_1(1)\boldsymbol{\pi}_2(2), \boldsymbol{\pi}_3}\left( \overline{\sigma} \right) 
	-
	p_{\boldsymbol{\pi}_2(1)\boldsymbol{\pi}_1(2), \boldsymbol{\pi}_3}\left( \overline{\sigma} \right) 
	+
	p_{\boldsymbol{\pi}_2, \boldsymbol{\pi}_3}\left( \overline{\sigma} \right) 
\right\}
\end{align*}

To compute the H\'{a}jek projection of the above sum into an arbitrary function of $\zeta_{\boldsymbol{\pi}_1}$, 
consider first $E \left[ p_{\boldsymbol{\pi}_1, \boldsymbol{\pi}_3}\left( \overline{\sigma} \right) \mid \zeta_{\boldsymbol{\pi}_1} \right]$.
To that end, the following results will be useful. 
\begin{eqnarray*}
E
\left[  
	E\left[ 
		\tilde{W}_{\boldsymbol{\pi}_1, \boldsymbol{\pi}_2}
		D_{\boldsymbol{\pi}_1,\tau}^\ast 
		\mid 
		\omega_{\boldsymbol{\pi}_1} 
	\right]
	\mid 	
	\zeta_{\boldsymbol{\pi}_1}
\right]
&=&
E\left[ D_{\boldsymbol{\pi}_1,\tau}^\ast \mid \omega_{\boldsymbol{\pi}_1} \right]
E
\left[  
	\tilde{W}_{\boldsymbol{\pi}_1, \boldsymbol{\pi}_2}
	\mid 	
	X_{\boldsymbol{\pi}_1}
\right]
=
E\left[ 
	D_{\boldsymbol{\pi}_1,\tau}^\ast 
	\chi_{\boldsymbol{\pi}_1}
	\mid \omega_{\boldsymbol{\pi}_1} 
\right]
\\
E
\left[  
	E
	\left[ 
		\tilde{W}_{\boldsymbol{\pi}_3, \boldsymbol{\pi}_2}
		D_{\boldsymbol{\pi}_3,\tau}^\ast 
		\mid 
		\omega_{\boldsymbol{\pi}_3} 
	\right]
	\mid 	
	\zeta_{\boldsymbol{\pi}_1}
\right]
&=&
E
\left[ 
	E\left[ D_{\boldsymbol{\pi}_3,\tau}^\ast \mid \omega_{\boldsymbol{\pi}_3} \right]
	E
	\left[ 
	\tilde{W}_{\boldsymbol{\pi}_3, \boldsymbol{\pi}_2}
	\mid 
	X_{\boldsymbol{\pi}_3}
	\right]
\right]
=
E
\left[ 
	D_{\boldsymbol{\pi}_3,\tau}^\ast 
	\chi_{\boldsymbol{\pi}_3}
\right].	
\end{eqnarray*}

Furthermore,  
\begin{align*}
& 
E
\left[ 
	\left( 
		\frac{ 
			\varphi_{ \boldsymbol{\pi}_1, \tau}	
		}{ 
			f_{vx,\boldsymbol{\pi}_1} 
		}
		\tilde{W}_{\boldsymbol{\pi}_1,\boldsymbol{\pi}_2}
		+
		\frac{ 
			\varphi_{ \boldsymbol{\pi}_3, \tau}	
		}{ 
			f_{vx,\boldsymbol{\pi}_3} 
		}
		\tilde{W}_{\boldsymbol{\pi}_3,\boldsymbol{\pi}_2}
\right)
\frac{1}{h^L}
K_{x,h}\left( X_{\boldsymbol{\pi}_3}-X_{\boldsymbol{\pi}_1}\right) 	
\mid 
\zeta_{\boldsymbol{\pi}_1}
\right]	
\\
=&
E
\left[ 
\left\{ 
\frac{ 
	\varphi_{ \boldsymbol{\pi}_1, \tau}	
}{ 
	f_{vx,\boldsymbol{\pi}_1} 
}	 
E
\left[ 
	\tilde{W}_{\boldsymbol{\pi}_1,\boldsymbol{\pi}_2}
	\mid 
	X_{\boldsymbol{\pi}_1}
\right]	
+
E
\left[  
		\frac{ 
			\varphi_{ \boldsymbol{\pi}_3, \tau}	
		}{ 
			f_{vx,\boldsymbol{\pi}_3} 
		}
\mid 
X_{\boldsymbol{\pi}_3}
\right]
E
\left[ 
	\tilde{W}_{\boldsymbol{\pi}_3,\boldsymbol{\pi}_2}
	\mid 
	X_{\boldsymbol{\pi}_3}
\right]	
\right\}
\frac{1}{h^L}
K_{x,h}\left( X_{\boldsymbol{\pi}_3}-X_{\boldsymbol{\pi}_1}\right) 	
\mid 
\zeta_{\boldsymbol{\pi}_1}
\right]	
\\
=&
\int 
\left\{ 
\frac{ 
	\varphi_{ \boldsymbol{\pi}_1, \tau}	
}{ 
	f_{vx,\boldsymbol{\pi}_1} 
}	
\chi_{\boldsymbol{\pi}_1}
+
E
\left[  
		\frac{ 
			\varphi_{ \boldsymbol{\pi}_3, \tau}	
		}{ 
			f_{vx,\boldsymbol{\pi}_3} 
		}
		\chi_{\boldsymbol{\pi}_3}
\mid 
X_{\boldsymbol{\pi}_3}
\right]
\right\}
\frac{1}{h^L}
K_{x,h}\left( X_{\boldsymbol{\pi}_3}-X_{\boldsymbol{\pi}_1}\right) 	
f_{x}(X_{\boldsymbol{\pi}_3})
d X_{\boldsymbol{\pi}_3}
\end{align*}
where the second equality follows from a Law of Iterated Expectations and Assumption \ref{Ass:SR00:iidsampling}.

Let 
\begin{eqnarray*}
\Xi
\left( X_{\boldsymbol{\pi}_3} \right) 
&=&
E
\left[  
D_{ \boldsymbol{\pi}_3, \tau}^\ast	
\chi_{\boldsymbol{\pi}_3}
\mid 
X_{\boldsymbol{\pi}_3}
\right],
\end{eqnarray*}
and consider 
\begin{eqnarray*}
&&
\int 
\left\{ 
\frac{ 
	\varphi_{ \boldsymbol{\pi}_1, \tau}	
}{ 
	f_{vx,\boldsymbol{\pi}_1} 
}	
\chi_{\boldsymbol{\pi}_1} 
f_{x}(X_{\boldsymbol{\pi}_3})
+
\Xi\left( X_{\boldsymbol{\pi}_3} \right) 
\right\}
\frac{1}{h^L}
K_{x,h}\left( X_{\boldsymbol{\pi}_3}-X_{\boldsymbol{\pi}_1}\right) 	
d X_{\boldsymbol{\pi}_3}
-
\left\{ 
\frac{ 
	\varphi_{ \boldsymbol{\pi}_1, \tau}	
}{ 
	f_{vx,\boldsymbol{\pi}_1} 
}	
\chi_{\boldsymbol{\pi}_1} 
f_{x}(X_{\boldsymbol{\pi}_1})
+
\Xi\left( X_{\boldsymbol{\pi}_1} \right) 
\right\}
\\
&=&
\int 
\left\{ 
\frac{ 
	\varphi_{ \boldsymbol{\pi}_1, \tau}	
}{ 
	f_{vx,\boldsymbol{\pi}_1} 
}	
\chi_{\boldsymbol{\pi}_1} 
f_{x}(X_{\boldsymbol{\pi}_1} + h \boldsymbol{\nu})
+
\Xi
\left( 
	X_{\boldsymbol{\pi}_1} + h \boldsymbol{\nu}
\right) 
\right\}
K_{x,h}\left( \boldsymbol{\nu} \right) 	
d \boldsymbol{\nu}
-
\left\{ 
\frac{ 
	\varphi_{ \boldsymbol{\pi}_1, \tau}	
}{ 
	f_{vx,\boldsymbol{\pi}_1} 
}	
\chi_{\boldsymbol{\pi}_1} 
f_{x}(X_{\boldsymbol{\pi}_1})
+
\Xi\left( X_{\boldsymbol{\pi}_1} \right) 
\right\}
\\
&=&
\int 
\left\{ 
\frac{ 
	\varphi_{ \boldsymbol{\pi}_1, \tau}	
}{ 
	f_{vx,\boldsymbol{\pi}_1} 
}	
\chi_{\boldsymbol{\pi}_1} 
\left( 
f_{x}(X_{\boldsymbol{\pi}_1} + h \boldsymbol{\nu})
-
f_{x}(X_{\boldsymbol{\pi}_1})
\right)
\right\}
+
\left\{ 
\Xi
\left( 
	X_{\boldsymbol{\pi}_1} + h \boldsymbol{\nu}
\right) 
-
\Xi\left( X_{\boldsymbol{\pi}_1} \right) 
\right\}
K_{x}\left( \boldsymbol{\nu} \right) 	
d \boldsymbol{\nu}
\\
&=& 
o (h^{\overline{M}})
\end{eqnarray*}
where the first equality follows from a change of variable $\boldsymbol{\nu}=h^{-1}\left( X_{\boldsymbol{\pi}_3}- X_{\boldsymbol{\pi}_1} \right)$ 
with Jacobian $h^L$.
The last equality follows Assumptions 
\ref{Ass:Inf01:SamplingMoments},
\ref{Ass:Inf03:SmoothnessDensity}, and \ref{Ass:Inf04:Kernel} 
which guarantee that $f_{x}(X_{\boldsymbol{\pi}_1})$ and $\Xi\left( X_{\boldsymbol{\pi}_1} \right) $
are continuous and $\overline{M}$-times differentiable with respect to all of its arguments, and 
$K_{x}$ is a bias-reducing kernel of order $2\overline{M}$.
Observe that 
\begin{eqnarray*}
\frac{ 
	\varphi_{ \boldsymbol{\pi}_1, \tau}	
}{ 
	f_{vx,\boldsymbol{\pi}_1} 
}	
\chi_{\boldsymbol{\pi}_1} 
f_{x}(X_{\boldsymbol{\pi}_1})
&=&0
\end{eqnarray*}
holds for any $X_{\boldsymbol{\pi}_1}$ within a $\tau$ distance of the boundary $\mathbb{S}_{x}$, 
and having $h/\tau \rightarrow 0$ ensures that the change of variable 
$\boldsymbol{\nu}=h^{-1}\left( X_{\boldsymbol{\pi}_3}- X_{\boldsymbol{\pi}_1} \right)$
is not affected by boundary effects.

The previous results, and Assumption \ref{Ass:Inf04:Kernel}, yield 
\begin{eqnarray*}
E
\left[ 
p_{\boldsymbol{\pi}_1 , \boldsymbol{\pi}_3  }
\left( \overline{\sigma} \right)	
\mid 
\zeta_{\boldsymbol{\pi}_1}
\right]	
=
D_{ \boldsymbol{\pi}_1, \tau}^\ast	
\chi_{\boldsymbol{\pi}_1}
+
E
\left[ 
D_{ \boldsymbol{\pi}_1, \tau}^\ast
\chi_{\boldsymbol{\pi}_1}
\mid 
X_{\boldsymbol{\pi}_1}
\right]
-
E
\left[ 
D_{ \boldsymbol{\pi}_1, \tau}^\ast
\chi_{\boldsymbol{\pi}_1}
\mid 
\omega_{\boldsymbol{\pi}_1}
\right]
-
E
\left[ 
D_{ \boldsymbol{\pi}_1, \tau}^\ast
\chi_{\boldsymbol{\pi}_1}
\right]
+ o(1).
\end{eqnarray*}

Notice that for $\boldsymbol{\pi}_s \in  \{ \left( \boldsymbol{\pi}_1(1),\boldsymbol{\pi}_2(2)\right), \left( \boldsymbol{\pi}_2(1),\boldsymbol{\pi}_1(2) \right),  \boldsymbol{\pi}_2 \}$,
\begin{eqnarray*}
&&
E
\left[ 
\tilde{W}_{\boldsymbol{\pi}_1,\boldsymbol{\pi}_2}
\left\{ 
\frac{1}{h^L}
\frac{ 
	\varphi_{ \boldsymbol{\pi}_s, \tau}
	}{ 
		f_{vx,\boldsymbol{\pi}_s} 	 
}
K_{x,h}\left( X_{\boldsymbol{\pi}_3}-X_{\boldsymbol{\pi}_s}\right)
-
E\left[ 
	D_{\boldsymbol{\pi}_s, \tau}^\ast
	\mid 
	\omega_{\boldsymbol{\pi}_s}
\right]
\right\}
\mid 
\zeta_{\boldsymbol{\pi}_1}
\right]	
\\
&=&
E
\left[ 
\tilde{W}_{\boldsymbol{\pi}_1,\boldsymbol{\pi}_2}
\left\{ 
E
\left[ 
\frac{1}{h^L}
\frac{ 
	\varphi_{ \boldsymbol{\pi}_s, \tau}
	}{ 
		f_{vx,\boldsymbol{\pi}_s} 	 
}
K_{x,h}\left( X_{\boldsymbol{\pi}_3}-X_{\boldsymbol{\pi}_s}\right)
\mid 
\Omega_{\sigma}, \zeta_{\boldsymbol{\pi}_1}
\right]
-
E\left[ 
	D_{\boldsymbol{\pi}_s, \tau}^\ast
	\mid 
	\omega_{\boldsymbol{\pi}_s}
\right]
\right\}
\mid 
\zeta_{\boldsymbol{\pi}_1}
\right]
\\
&=&
O
\left( 
h^{\overline{M}} 
\right)
\end{eqnarray*}
since the expectation 
\begin{eqnarray*}
E
\left[ 
\frac{1}{h^L}
\frac{ 
	\varphi_{ \boldsymbol{\pi}_s, \tau}
	}{ 
		f_{vx,\boldsymbol{\pi}_s} 	 
}
K_{x,h}\left( X_{\boldsymbol{\pi}_3}-X_{\boldsymbol{\pi}_s}\right)
\mid 
\Omega_{\sigma}, \zeta_{\boldsymbol{\pi}_1}
\right]
&=&
\int 
\frac{1}{h^L}
E
\left[  
\frac{ 
	\varphi_{ \boldsymbol{\pi}_s, \tau}
	}{ 
		f_{vx,\boldsymbol{\pi}_s} 	 
}
\mid 
\omega_{\boldsymbol{\pi}_s}
\right]
K_{x,h}\left( X_{\boldsymbol{\pi}_3}-X_{\boldsymbol{\pi}_s}\right)
f_{x}\left( X_{\boldsymbol{\pi}_3} \right)
d X_{\boldsymbol{\pi}_3}
\\
&=&
E
\left[  
	D_{ \boldsymbol{\pi}_s, \tau}^\ast
	\mid 
\omega_{\boldsymbol{\pi}_s}
\right]
+
O
\left( 
	h^{\overline{M}} 
\right), 
\end{eqnarray*}
where the second equality follows from 
Assumptions \ref{Ass:SR00:iidsampling}, \ref{Ass:SR01:distr}, 
and properties of the bias-reducing kernel, Assumption \ref{Ass:Inf04:Kernel}.

Similarly, for a given 
$\boldsymbol{\pi}_s \in  \{ \left( \boldsymbol{\pi}_3(1),\boldsymbol{\pi}_2(2)\right), \left( \boldsymbol{\pi}_2(1),\boldsymbol{\pi}_3(2) \right),  \boldsymbol{\pi}_2 \}$, 
it follows from Assumptions \ref{Ass:SR00:iidsampling}, \ref{Ass:SR01:distr}, 
\ref{Ass:Inf03:SmoothnessDensity}, and \ref{Ass:Inf04:Kernel}, that 
\begin{eqnarray*}
&&
E
\left[
\frac{1}{h^L}	 
\left(			
\frac{ 
	\varphi_{\boldsymbol{\pi}_s, \tau}	
}{ 
	f_{vx,\boldsymbol{\pi}_s} 
}
\tilde{W}_{\boldsymbol{\pi}_3,\boldsymbol{\pi}_2}
K_{x,h}
\left( 
	X_{\boldsymbol{\pi}_1}-X_{\boldsymbol{\pi}_s}
\right) 
\right)
\mid 
\zeta_{\boldsymbol{\pi}_1}
\right]
-
\Xi
\left[ 
X_{\boldsymbol{\pi}_1}
\right]
\\
&=&
E
\left[ 
\frac{1}{h^L}	
E
\left[ 	
	\left(			
	\frac{ 
		\varphi_{ \boldsymbol{\pi}_s, \tau}	
	}{ 
		f_{vx,\boldsymbol{\pi}_s}  
	}
	\chi_{\boldsymbol{\pi}_s}
	\right)
\mid 
X_{\boldsymbol{\pi}_s} 
\right]
K_{x,h}
\left( 
X_{\boldsymbol{\pi}_1}-X_{\boldsymbol{\pi}_s}
\right) 
\mid 
\zeta_{\boldsymbol{\pi}_1}
\right]
-
\Xi
\left[ 
X_{\boldsymbol{\pi}_1}
\right]
\\
&=&
\int 
\left\{ 
\Xi
\left[ 
	X_{ \boldsymbol{\pi}_1} + h \boldsymbol{\nu} 
\right]
-
\Xi
\left[ 
	X_{\boldsymbol{\pi}_1}
\right]
\right\}
K_{x}
\left( 
	\boldsymbol{\nu}
\right)
d\boldsymbol{\nu}
\\
&=&
O
\left( 
h^{\overline{M}}	
\right).
\end{eqnarray*}

Using the previous results it follows that 
\begin{eqnarray*}
E
\left[ 
p_{\boldsymbol{\pi}_s, \boldsymbol{\pi}_3}\left( \overline{\sigma} \right) 
\mid 
\zeta_{\boldsymbol{\pi}_1}
\right]	
&=&
E
\left[
D_{\boldsymbol{\pi}_1, \tau}^\ast
\chi_{\boldsymbol{\pi}_1}
\mid 
X_{\boldsymbol{\pi}_1}
\right]
-
E\left[ 
D_{\boldsymbol{\pi}_1, \tau}^\ast
\chi_{\boldsymbol{\pi}_1}
\right],	
\end{eqnarray*}
and thus,  
\begin{eqnarray*}
&&
E
\left[ 
p_{\boldsymbol{\pi}_1, \boldsymbol{\pi}_3}\left( \overline{\sigma} \right) 
-
p_{\boldsymbol{\pi}_1(1)\boldsymbol{\pi}_2(2), \boldsymbol{\pi}_3}\left( \overline{\sigma} \right) 
-
p_{\boldsymbol{\pi}_2(1)\boldsymbol{\pi}_1(2), \boldsymbol{\pi}_3}\left( \overline{\sigma} \right) 
+
p_{\boldsymbol{\pi}_2, \boldsymbol{\pi}_3}\left( \overline{\sigma} \right) 
\mid 
\zeta_{\boldsymbol{\pi}_1}
\right]	
\\
&=&
\left\{ 
D_{ \boldsymbol{\pi}_1}^\ast	
-
E
\left[ 
	D_{ \boldsymbol{\pi}_1}^\ast
	\mid 
	\omega_{\boldsymbol{\pi}_1}
\right]
\right\}
I_{\tau, \boldsymbol{\pi}_1}
\chi_{\boldsymbol{\pi}_1} + o(1)
\end{eqnarray*}

It then follows that the H\'{a}jek projection is given by 
\begin{eqnarray*}
V_{2,n \tau}^\ast
&=&
\frac{1}{n(n-1)}
\sum_{ i_1 =1}^{n}
\sum_{ j_1 \neq i_1 }
	\xi_{i_1j_1, \tau}
+ o(1)	
\end{eqnarray*}
with 
\begin{eqnarray*}
\xi_{i_1j_1, \tau}
&=&
\left\{ 
D^\ast_{ i_1 j_1} 
-
E
\left[ 
	D^\ast_{i_1j_1}
	\mid 
	\omega_{i_1j_1}	
\right]
\right\}
I_{\tau, i_1j_1} 
\overline{\chi}_{i_1j_1}
\\
\overline{\chi}_{i_1j_1}
&=&
\left\{ 
\frac{1}{(n-2)(n-3)}
\sum_{i_2 \neq i_1 , j_1} \sum_{j_2 \neq i_1 , j_1, i_2} 
	E\left[ 
		\tilde{W}_{\sigma\{i_1,i_2;j_1,j_2\}}
		\mid 	
		X_{i_1}, X_{j_1}
	\right]
\right\}.	
\end{eqnarray*}

If follows from a Law of Iterated Expectations that 
\begin{eqnarray*}
E
\left[ 
	V_{2,n \tau}^\ast
\right]
&=&
E
\left[ 
	\xi_{i_1j_1, \tau}	
	\right]
=
0. 
\end{eqnarray*}

\textbf{Step 2. Variance of H\'{a}jek Projection}

As in the proof of Lemma \ref{Lemma:TechAppx:Hajek_Sum1},
the variance of 
$ V_{1,n \tau}^\ast$ is given by 
\begin{eqnarray*}
Var\left( V_{1,n \tau}^\ast \right)
&=&
\left\{ \frac{1}{n(n-1)}\right\}^2
\left\{ 	
\sum_{i_1 =1}^n  \sum_{j_1 \neq i_1} 
E\left[ 
\xi_{i_1j_1, \tau}
\xi_{i_1'j_1', \tau}'
\right] 
\right\}
\\
&=&
\left\{ \frac{1}{n(n-1)}\right\}^2
\left\{ 
\sum_{i_1 =1}^n  \sum_{j_1 \neq i_1} 
\Lambda^\ast_{i_1, j_1}
\right\}
\end{eqnarray*}
where 
\begin{eqnarray*}
\Lambda^\ast_{i_1, j_1}
&=&
E
\left[ 
\left\{ 
	E
	\left[ 
		D^\ast_{i_1j_1}
		D^\ast_{i_1j_1}
		\mid 
		\omega_{i_1j_1}
	\right]	
	-
	E
	\left[ 
		D^\ast_{i_1j_1}
		\mid 
		\omega_{i_1j_1}
	\right]^2		
\right\}
I_{\tau, i_1j_1}^2
\overline{\chi}_{i_1j_1}
\overline{\chi}_{i_1j_1}'
\right].
\end{eqnarray*}

Define 
\begin{eqnarray*}
\Upsilon_{n} &=& n(n-1) Var\left( V_{1,n \tau}^\ast \right)
=
\frac{1}{n(n-1)}
\left\{ 
\sum_{i_1 =1}^n  \sum_{j_1 \neq i_1} 
\Lambda^\ast_{i_1, j_1}
\right\}.
\end{eqnarray*}

\textbf{Step 3. Variance of $S_{2,n \tau}$}

Given two different 6-tuples $\overline{\sigma}\{i_1,i_2,j_1,j_2, l_1,l_2\}$ and $\overline{\sigma}'\{i_1',i_2',j_1',j_2', l_1',l_2'\}$, let 
\begin{eqnarray*}
\Delta_{c,n} 
= 
Cov
\left( 
	s_{2, n} \left( \sigma\{i_1, i_2, j_1, j_2, l_1, l_2\} \right), 
	s_{2, n} \left( \sigma'\{i_1', i_2' , j_1' , j_2', l_1', l_2'\} \right)
\right)	
\end{eqnarray*}
denote the covariance between 
$s_{2,n}(\overline{\sigma})$ 
and 
$s_{2, n}(\overline{\sigma}')$ 
when 
$\overline{\sigma}$ and $\overline{\sigma}'$ have $c=0,1,2,3,4, 5, 6$ indices in common. 

It follows from the conditionally independent formation of links, implied by Assumption \ref{Ass:SR01:distr}, 
and the conditional mean zero, 
$E
\left[ 
s_{2, n}
\left( 
	\sigma\{i_1, i_2, j_1, j_2, l_1, l_2\}
\right)	
\mid \Omega_\sigma
\right]=0$,  
that $\Delta_{0,n}= \Delta_{1,n} =0$. 

Consider 
\begin{eqnarray*}
\Delta_{2,n}
&=&
E\left[ 
s_{2,n}(\overline{\sigma}\{i_1,i_2 ,j_1,j_2, l_1, l_2\} )	
s_{2,n}(\overline{\sigma'}\{i_1,i_2',j_1,j_2', l_1, l_2'\} )'	
\right]
\\
&=&
E
\left[
s_{i_1j_1}\left( \overline{\sigma} \right)
s_{i_1j_1}\left( \overline{\sigma'} \right)'
\right]
+
o(1)	
\\
&=&
E
\left[ 
\left\{ 
E
\left[ 
	\tilde{D}^\ast_{i_1j_1,\tau} 
	\tilde{D}^\ast_{i_1j_1,\tau} 
	\mid 
	\omega_{i_1j_1}
\right]	
-
E
\left[ 
	\tilde{D}^\ast_{i_1j_1,\tau} 
	\mid 
	\omega_{i_1j_1}
\right]^2
\right\}
I_{\tau, i_1j_1}^2
\tilde{W}_{\sigma}
\tilde{W}_{\sigma'}
\right]
+
o(1).
\end{eqnarray*}

Therefore, the variance of $Var(S_{2,n \tau}^\dagger)$ can be expressed as 
\begin{eqnarray*}
&&
\left(  \frac{1}{\overline{m}_{n}} \right)^2
\sum_{ \overline{\sigma} } 
\sum_{ \overline{\sigma}' } 
E
\left[ 
	\left( 
	s_{2,n}(\overline{\sigma})
	s_{2,n}(\overline{\sigma'})'
	\right)
\right]
\\
&& 
+
\left(  
4! \binom{n}{4}
\right)^{-2}
\sum_{ i_1=1}^{n} 
\sum_{ j_1 \neq i_1} 
\left\{ 
\sum_{k_1 \neq i_1 , j_1} \sum_{k_2 \neq i_1 , j_1, k_1}
\sum_{l_1 \neq i_1 , j_1} \sum_{l_2 \neq i_1 , j_1, l_1}
\Delta_{2,n}
\right\}
\\
&&
+
O\left(\frac{1}{n^3}\right)
\Delta_{3,n}
+
O\left(\frac{1}{n^4}\right)
\Delta_{4,n}
+
O\left(\frac{1}{n^5}\right)
\Delta_{5,n}
+
O\left(\frac{1}{n^6}\right)
\Delta_{6,n}
\end{eqnarray*}

Notice that the term inside the brackets scaled by $((n-2)(n-3))^{-2}$ can be written as 
\begin{eqnarray*}
&&
\left(
	\frac{1}{(n-2) (n-3)}
\right )^2	
\sum_{k_1 \neq i_1 , j_1} \sum_{k_2 \neq i_1 , j_1, k_1}
\sum_{l_1 \neq i_1 , j_1} \sum_{l_2 \neq i_1 , j_1, l_1}
\Delta_{2,n}
\\
&=&
E
\left[ 
	\left\{ 
		E
		\left[ 
			D^\ast_{i_1j_1}
			D^\ast_{i_1j_1}
			\mid 
			\omega_{i_1j_1}
		\right]	
		-
		E
		\left[ 
			D^\ast_{i_1j_1}
			\mid 
			\omega_{i_1j_1}
		\right]^2		
	\right\}
	I_{\tau, i_1j_1}^2
	\chi_{i_1j_1}
	\chi_{i_1j_1}'
\right]
\\
&=& 
\Lambda^\ast_{i_1,j_1}.
\end{eqnarray*}

As a result, 
\begin{eqnarray*}
Var
\left[ 
	S_{2,n \tau}^\dagger
\right]
&=&
\left(  \frac{1}{n(n-1)} \right)^2
\left\{ 
\sum_{ i_1=1}^{n} 
\sum_{j_1\neq i_1}
	\Lambda^\ast_{i_1,j_1} 
\right\}
+ o(1),
\end{eqnarray*}
and $Var \left[  V_{2,n \tau}^\ast   \right] - Var \left[  S_{2,n \tau}^\dagger \right] = o_p(1)$.

The asymptotic equivalence results follows from similar arguments as in the proof of Lemma \ref{Lemma:TechAppx:Hajek_Sum1}.
The proof is complete.
\end{proof}
\begin{lemma}\label{Lemma:TechAppx:Hajek_Sum3}
Under the same Assumptions of Theorem \ref{Theo:Inf:AN}, it follows that the H\'{a}jek projection of 
\begin{eqnarray*}
S_{3,n \tau}^\dagger
&=&
S_{3,n \tau}
-
E 
\left[ 
	\tilde{W}_{\sigma}
	\tilde{D}^\ast_{\sigma,\tau} 
	\mid 
	\Omega_{\sigma}
\right]
\\
S_{3,n \tau}
&=&
\frac{1}{m_{n}} 
\sum_{ \sigma \in \mathcal{N}_{m_{n}}} 
\tilde{W}_{\sigma}
\left\{ 
	\left( 
		\frac{ 
			D_{ i_1 j_1, \tau}^\ast	
			\widehat{f}_{vx,i_1 j_1} 
		}{ 
			f_{vx,i_1 j_1} 
		}
		-
		\frac{ 
			D_{ i_1 j_2, \tau}^\ast	
			\widehat{f}_{vx,i_1 j_2} 
		}{ 
			f_{vx,i_1 j_2} 
		}
	\right)
	-
	\left( 
		\frac{ 
			D_{ i_2 j_1, \tau}^\ast
			\widehat{f}_{vx,i_2 j_1} 
		}{ 
			f_{vx,i_2 j_1} 
		}
		-
		\frac{ 
			D_{ i_2 j_2, \tau}^\ast	
			\widehat{f}_{vx,i_2 j_2} 
		}{ 
			f_{vx,i_2 j_2} 
		}
	\right)
\right\}
\end{eqnarray*}
into an arbitrary function $\zeta_{i_1j_1} = \left( X_{i_1},X_{j_1}, A_{i_1},A_{j_1}, v_{i_1j_1},U_{i_1j_1} \right)$ is given by 
\begin{eqnarray*}
V_{3,n\tau}^\ast
\;
=
\;
E \left[  S_{3,n \tau}^\dagger \mid  \zeta_{i_1j_1} \right]	
&=& 
\frac{1}{n(n-1)} 
\sum_{i_1 = 1}^{n} \sum_{j_1 \neq i_1}
\xi_{i_1j_1, \tau}
\end{eqnarray*}
and 
\begin{eqnarray*}
n(n-1)
\Upsilon_{n}^{-1/2}
E 
\left[  
	\left(  
		S_{3,n \tau}^\dagger - 	V_{3,n\tau}^\ast
	\right)^2
\right]
\Upsilon_{n}^{-1/2}
= o(1),
\end{eqnarray*}
where $\Upsilon_{n} = n(n-1) Var(V_{3,n\tau}^\ast)$.
\end{lemma}

\begin{proof}
Consider a fixed 6-tuple $\{i_1,i_2,j_1,j_2,k_1,k_2\}$, and define  
\begin{eqnarray*}
s_{i_1, j_1}(\overline{\sigma})
&=&
\tilde{W}_{i_1i_2,j_1j_2}
\left\{ 
		\frac{1}{h^{L+1}}
		\frac{ 
			D_{ i_1 j_1, \tau}^\ast	
		}{ 
			f_{vx,i_1 j_1} 
		}
		K_{vx,h}\left( v_{k_1k_2}-v_{i_1j_1}, X_{k_1}-X_{i_1}, X_{k_2}-X_{j_1}\right) 
		-
		E\left[ 
			D_{i_1j_1,\tau}^\ast 
			\mid 
			\Omega_{i_1i_2,j_1j_2}
		\right]
\right\}		
\\
s_{i_1,j_2}(\overline{\sigma})
&=&
\tilde{W}_{i_1i_2,j_1j_2}
\left\{ 
		\frac{1}{h^{L+1}}
		\frac{ 
		D_{ i_1 j_2, \tau}^\ast	
		}{ 
			f_{vx,i_1 j_2} 
		}
		K_{vx,h}\left( v_{k_1k_2}-v_{i_1j_2}, X_{k_1}-X_{i_1}, X_{k_2}-X_{j_2}\right) 
		-
		E\left[  
			D_{i_1j_2,\tau}^\ast 
			\mid 
			\Omega_{i_1i_2,j_1j_2}
		\right]	
\right\}
\\
s_{i_2,j_1}(\overline{\sigma})
&=&
\tilde{W}_{i_1i_2,j_1j_2}
\left\{ 
		\frac{1}{h^{L+1}}
		\frac{ 
			D_{ i_2 j_1, \tau^\ast}
			}{ 
				f_{vx,i_2 j_1} 	 
		}
		K_{vx,h}\left( v_{k_1k_2}-v_{i_2j_1}, X_{k_1}-X_{i_2}, X_{k_2}-X_{j_1}\right)
		-
		E
		\left[ 
			D_{i_2j_1,\tau}^\ast 
		\mid 
		\Omega_{i_1i_2,j_1j_2}
		\right]
\right\}		
\\
s_{i_2, j_2}\overline{\sigma})
&=&
\tilde{W}_{i_1i_2,j_1j_2}
\left\{ 
	\frac{1}{h^{L+1}}
	\frac{ 
		D_{ i_2 j_2, \tau}^\ast	
	}{ 
		f_{vx,i_2 j_2} 
	}
	K_{vx,h}\left( v_{k_1k_2}-v_{i_2j_2}, X_{k_1}-X_{i_2}, X_{k_2}-X_{j_2}\right)
	-
	E\left[ 
		D_{i_2j_2,\tau}^\ast 
		\mid 
		\Omega_{i_1i_2,j_1j_2}	
	\right]
\right\},
\end{eqnarray*}
and $s_{3, n}(\overline{\sigma})= 	s_{i_1, j_1}(\overline{\sigma})-s_{i_1,j_2}(\overline{\sigma})- s_{i_2,j_1}(\overline{\sigma})+	s_{i_2,j_2}(\overline{\sigma})$.
It follows then that $S_{3,n \tau}^\dagger$ can be written as 
\begin{eqnarray*}
S_{3,n \tau}^\dagger
&=&
\left[ 
	6!\binom{n}{6}
\right]^{-1}
\sum_{\overline{\sigma} \in \mathcal{N}_{\overline{m}_n}} 
s_{2, n\tau}(\overline{\sigma})
\\
&=&
\left[ 
	6!\binom{n}{6}
\right]^{-1}
\sum_{\overline{\sigma} \in \mathcal{N}_{\overline{m}_n}} 
\left\{ 
	s_{i_1j_1}(\overline{\sigma})
	-
	s_{i_1j_2}(\overline{\sigma})
	-
	s_{i_2j_1}(\overline{\sigma})
	+
	s_{i_2j_2}(\overline{\sigma})
\right\}.
\end{eqnarray*}

\textbf{Step 1. H\'{a}jek Projection}

The rest of the proof makes use of the following index notation for dyads.
Given the total number of ordered dyads $\overline{\textbf{n}} = n(n-1)$, 
let the boldface indeces $\boldsymbol{\pi} = \textbf{1}, \textbf{2}, \cdots$ index the $\overline{\textbf{n}}$ ordered dyads in the sample.
In an abuse of notation, also let $\boldsymbol{\pi}$ denote the set $\left\{ i_1,j_1\right\}$, where $i_1$ and $j_1$ are the indices that comprise dyad $\boldsymbol{\pi}$.
In particular, $\boldsymbol{\pi}(1)=i_1$ and $\boldsymbol{\pi}(2)=j_1$, when $\boldsymbol{\pi}=\left\{ i_1,j_1\right\}$. 

With this notation at hand, $S_{3,n \tau}^\dagger$ can be expressed as 
\begin{eqnarray*}
S_{3,n \tau}^\dagger
&=&
\left[ 
	6!\binom{n}{6}
\right]^{-1}
\sum_{\boldsymbol{\pi}_1=\textbf{1}}^{\overline{\textbf{n}}}
\sum_{\boldsymbol{\pi}_2 \neq \boldsymbol{\pi}_1 }
\sum_{\boldsymbol{\pi}_3 \neq \boldsymbol{\pi}_1 }
\left\{ 
	s_{\boldsymbol{\pi}_1}(\overline{\sigma})
	-
	s_{\boldsymbol{\pi}_1(1)\boldsymbol{\pi}_2(2)}(\overline{\sigma})
	-
	s_{\boldsymbol{\pi}_2(1)\boldsymbol{\pi}_1(2)}(\overline{\sigma})
	+
	s_{\boldsymbol{\pi}_2}(\overline{\sigma})
\right\}
\end{eqnarray*}
where $\overline{\sigma}=\overline{\sigma}\left\{\boldsymbol{\pi}_1, \boldsymbol{\pi}_2, \boldsymbol{\pi}_3 \right\}$.

Let 
\begin{eqnarray*}
p_{\boldsymbol{\pi}_1, \boldsymbol{\pi}_3} 
\left( \overline{\sigma} \right)
&=& 
\frac{1}{h^{L+1}}
\left( 
		\frac{ 
			D_{ \boldsymbol{\pi}_1, \tau}^\ast	
		}{ 
			f_{vx,\boldsymbol{\pi}_1} 
		}
		\tilde{W}_{\boldsymbol{\pi}_1,\boldsymbol{\pi}_2}
		+
		\frac{ 
			D_{ \boldsymbol{\pi}_3, \tau}^\ast	
		}{ 
			f_{vx,\boldsymbol{\pi}_3} 
		}
		\tilde{W}_{\boldsymbol{\pi}_3,\boldsymbol{\pi}_2}
\right)
K_{vx,h}
\left(
	v_{\boldsymbol{\pi}_3}-v_{\boldsymbol{\pi}_1}, 
	X_{\boldsymbol{\pi}_3}-X_{\boldsymbol{\pi}_1}
\right) 
\\
&&
-
E\left[ 
	\tilde{W}_{\boldsymbol{\pi}_1,\boldsymbol{\pi}_2}
	D_{\boldsymbol{\pi}_1,\tau}^\ast 
	\mid 
	\Omega_{\boldsymbol{\pi}_1, \boldsymbol{\pi}_2}
\right]
-
E\left[ 
	\tilde{W}_{\boldsymbol{\pi}_3,\boldsymbol{\pi}_2}
	D_{\boldsymbol{\pi}_3,\tau}^\ast 
	\mid 
	\Omega_{\boldsymbol{\pi}_3, \boldsymbol{\pi}_2}
\right]	 
\\
p_{\boldsymbol{\pi}_2, \boldsymbol{\pi}_3} 
\left( \overline{\sigma} \right)
&=& 
\frac{1}{h^{L+1}} 
\tilde{W}_{\boldsymbol{\pi}_1,\boldsymbol{\pi}_2}
\left\{ 
\frac{ 
	D_{ \boldsymbol{\pi}_2, \tau}^\ast	
}{ 
	f_{vx,\boldsymbol{\pi}_2} 
}
K_{vx,h}
\left(
	v_{\boldsymbol{\pi}_3}-v_{\boldsymbol{\pi}_2},  
	X_{\boldsymbol{\pi}_3}-X_{\boldsymbol{\pi}_2}
\right) 
-
E\left[ 
D_{\boldsymbol{\pi}_2,\tau}^\ast 
\mid 
\Omega_{\boldsymbol{\pi}_1, \boldsymbol{\pi}_2}
\right]
\right\}
\\
&&
\frac{1}{h^{L+1}}
\tilde{W}_{\boldsymbol{\pi}_3,\boldsymbol{\pi}_2}
\left\{ 
\frac{ 
D_{ \boldsymbol{\pi}_2, \tau}^\ast	
}{ 
f_{vx,\boldsymbol{\pi}_2} 
}
K_{vx,h}
\left(
v_{\boldsymbol{\pi}_1}-v_{\boldsymbol{\pi}_2}, 
X_{\boldsymbol{\pi}_1}-X_{\boldsymbol{\pi}_2}
\right) 
-
E\left[ 
D_{\boldsymbol{\pi}_2,\tau}^\ast 
\mid 
\Omega_{\boldsymbol{\pi}_3, \boldsymbol{\pi}_2}
\right]	
\right\}
\\	
p_{\boldsymbol{\pi}_1(1)\boldsymbol{\pi}_2(2), \boldsymbol{\pi}_3} 
\left( \overline{\sigma} \right)
&=& 
\frac{1}{h^{L+1}}
\frac{ 
	D_{ \boldsymbol{\pi}_1(1)\boldsymbol{\pi}_2(2), \tau}^\ast	
}{ 
	f_{vx,\boldsymbol{\pi}_1(1)\boldsymbol{\pi}_2(2)} 
}
\tilde{W}_{\boldsymbol{\pi}_1,\boldsymbol{\pi}_2}
K_{vx,h}
\left( 
v_{\boldsymbol{\pi}_3}-v_{\boldsymbol{\pi}_1(1)\boldsymbol{\pi}_2(2)}, 	
X_{\boldsymbol{\pi}_3}-X_{\boldsymbol{\pi}_1(1)\boldsymbol{\pi}_2(2)}
\right) 
\\
&&
+
\frac{1}{h^{L+1}}		
\frac{ 
	D_{ \boldsymbol{\pi}_3(1)\boldsymbol{\pi}_2(2), \tau}^\ast	
}{ 
	f_{vx,\boldsymbol{\pi}_3(1)\boldsymbol{\pi}_2(2)} 
}
\tilde{W}_{\boldsymbol{\pi}_3,\boldsymbol{\pi}_2}
K_{vx,h}
\left( 
	v_{\boldsymbol{\pi}_1}-v_{\boldsymbol{\pi}_3(1)\boldsymbol{\pi}_2(2)}, 	
	X_{\boldsymbol{\pi}_1}-X_{\boldsymbol{\pi}_3(1)\boldsymbol{\pi}_2(2)}
\right) 
\\
&&
-
E\left[ 
\tilde{W}_{\boldsymbol{\pi}_1,\boldsymbol{\pi}_2}
D_{\boldsymbol{\pi}_1(1)\boldsymbol{\pi}_2(2),\tau}^\ast 
\mid 
\Omega_{\boldsymbol{\pi}_1, \boldsymbol{\pi}_2}
\right]
-
E\left[ 
\tilde{W}_{\boldsymbol{\pi}_3,\boldsymbol{\pi}_2}
D_{\boldsymbol{\pi}_3(1)\boldsymbol{\pi}_2(2),\tau}^\ast 
\mid 
\Omega_{\boldsymbol{\pi}_3, \boldsymbol{\pi}_2}
\right]	
\\
p_{\boldsymbol{\pi}_2(1)\boldsymbol{\pi}_1(2), \boldsymbol{\pi}_3} 
\left( \overline{\sigma} \right)
&=& 
\frac{1}{h^{L+1}}	
\frac{ 
	D_{ \boldsymbol{\pi}_2(1)\boldsymbol{\pi}_1(2), \tau}^\ast	
}{ 
	f_{ \boldsymbol{\pi}_2(1)\boldsymbol{\pi}_1(2), \tau}	
}
\tilde{W}_{\boldsymbol{\pi}_1,\boldsymbol{\pi}_2}
K_{vx,h}
\left( 
	v_{\boldsymbol{\pi}_3}-v_{\boldsymbol{\pi}_2(1)\boldsymbol{\pi}_1(2)}, 
	X_{\boldsymbol{\pi}_3}-X_{\boldsymbol{\pi}_2(1)\boldsymbol{\pi}_1(2)}
\right) 
\\
&&
+	
\frac{1}{h^{L+1}}		
\frac{ 
D_{ \boldsymbol{\pi}_2(1)\boldsymbol{\pi}_3(2), \tau}^\ast	
}{ 
f_{vx,\boldsymbol{\pi}_2(1)\boldsymbol{\pi}_3(2)} 
}
\tilde{W}_{\boldsymbol{\pi}_3,\boldsymbol{\pi}_2}
K_{vx,h}
\left( 
v_{\boldsymbol{\pi}_1}-v_{\boldsymbol{\pi}_2(1)\boldsymbol{\pi}_3(2)}, 
X_{\boldsymbol{\pi}_1}-X_{\boldsymbol{\pi}_2(1)\boldsymbol{\pi}_3(2)}
\right) 
\\
&&
-
E\left[ 
\tilde{W}_{\boldsymbol{\pi}_1,\boldsymbol{\pi}_2}
D_{\boldsymbol{\pi}_2(1)\boldsymbol{\pi}_1(2),\tau}^\ast 
\mid 
\Omega_{\boldsymbol{\pi}_1, \boldsymbol{\pi}_2}
\right]
-
E\left[ 
\tilde{W}_{\boldsymbol{\pi}_3,\boldsymbol{\pi}_2}
D_{\boldsymbol{\pi}_2(1)\boldsymbol{\pi}_3(2),\tau}^\ast 
\mid 
\Omega_{\boldsymbol{\pi}_3, \boldsymbol{\pi}_2}
\right]	
\\		
\end{eqnarray*}
where
$K_{vx,h}
\left( 
v_{\boldsymbol{\pi}_3}-v_{\boldsymbol{\pi}_1},
X_{\boldsymbol{\pi}_3}-X_{\boldsymbol{\pi}_1}
\right)$ 	
denotes 
$K_{vx,h}
\left(
v_{\boldsymbol{\pi}_3} - v_{\boldsymbol{\pi}_1},  
X_{\boldsymbol{\pi}_3(1)} - X_{\boldsymbol{\pi}_1(1)}, 
X_{\boldsymbol{\pi}_3(2)} - X_{\boldsymbol{\pi}_1(2)}
\right)$, 
$\tilde{W}_{\boldsymbol{\pi}_1, \boldsymbol{\pi}_2 }$ 
denotes 
$\tilde{W}_{\boldsymbol{\pi}_1\{i_1i_2\},\boldsymbol{\pi}_2\{j_1j_2\} }$, 
and
\begin{eqnarray*}
\chi_{\boldsymbol{\pi}_1}
&=&
E
\left[ 
\tilde{W}_{\boldsymbol{\pi}_1, \boldsymbol{\pi}_2 }
\mid 
X_{\boldsymbol{\pi}_1}
\right]
\\
\overline{\chi}_{\boldsymbol{\pi}_1}
&=&
\sum_{\boldsymbol{\pi}_2 \neq \boldsymbol{\pi}_1, \boldsymbol{\pi}_3 } 
\chi_{\boldsymbol{\pi}_1}.
\end{eqnarray*}

Using the symmetry of the kernel, it follows that $	S_{3,n \tau}^\dagger$ can be written as
\begin{align*}
\left[ 
	6!\binom{n}{6}
\right]^{-1}
\sum_{\boldsymbol{\pi}_1=\textbf{1}}^{\overline{\textbf{n}}}
\sum_{\boldsymbol{\pi}_3 = \boldsymbol{\pi}_1 + 1 }
\sum_{\boldsymbol{\pi}_2 \neq \boldsymbol{\pi}_1, \boldsymbol{\pi}_3 } 
\left\{ 
	p_{\boldsymbol{\pi}_1, \boldsymbol{\pi}_3}\left( \overline{\sigma} \right) 
	-
	p_{\boldsymbol{\pi}_1(1)\boldsymbol{\pi}_2(2), \boldsymbol{\pi}_3}\left( \overline{\sigma} \right) 
	-
	p_{\boldsymbol{\pi}_2(1)\boldsymbol{\pi}_1(2), \boldsymbol{\pi}_3}\left( \overline{\sigma} \right) 
	+
	p_{\boldsymbol{\pi}_2, \boldsymbol{\pi}_3}\left( \overline{\sigma} \right) 
\right\}
\end{align*}

To compute the H\'{a}jek projection of the above sum into an arbitrary function of $\zeta_{\boldsymbol{\pi}_1}$, 
consider first $ E \left[ p_{\boldsymbol{\pi}_1, \boldsymbol{\pi}_3}\left( \overline{\sigma} \right) \mid \zeta_{\boldsymbol{\pi}_1} \right]$ .
To that end, the following results will be useful. 
\begin{eqnarray*}
E
\left[  
	E\left[ 
		\tilde{W}_{\boldsymbol{\pi}_1, \boldsymbol{\pi}_2}
		D_{\boldsymbol{\pi}_1,\tau}^\ast 
		\mid 
		\omega_{\boldsymbol{\pi}_1} 
	\right]
	\mid 	
	\zeta_{\boldsymbol{\pi}_1}
\right]
&=&
E\left[ D_{\boldsymbol{\pi}_1,\tau}^\ast \mid \omega_{\boldsymbol{\pi}_1} \right]
E
\left[  
	\tilde{W}_{\boldsymbol{\pi}_1, \boldsymbol{\pi}_2}
	\mid 	
	X_{\boldsymbol{\pi}_1}
\right]
=
E\left[ 
	D_{\boldsymbol{\pi}_1,\tau}^\ast 
	\chi_{\boldsymbol{\pi}_1}
	\mid \omega_{\boldsymbol{\pi}_1} 
\right]
\\
E
\left[  
	E
	\left[ 
		\tilde{W}_{\boldsymbol{\pi}_3, \boldsymbol{\pi}_2}
		D_{\boldsymbol{\pi}_3,\tau}^\ast 
		\mid 
		\omega_{\boldsymbol{\pi}_3} 
	\right]
	\mid 	
	\zeta_{\boldsymbol{\pi}_1}
\right]
&=&
E
\left[ 
	E\left[ D_{\boldsymbol{\pi}_3,\tau}^\ast \mid \omega_{\boldsymbol{\pi}_3} \right]
	E
	\left[ 
	\tilde{W}_{\boldsymbol{\pi}_3, \boldsymbol{\pi}_2}
	\mid 
	X_{\boldsymbol{\pi}_3}
	\right]
\right]
=
E
\left[ 
	D_{\boldsymbol{\pi}_3,\tau}^\ast 
	\chi_{\boldsymbol{\pi}_3}
\right].	
\end{eqnarray*}

Moreover, 
\begin{align*}
& 
E
\left[ 
\left( 
	\frac{ 
		D_{ \boldsymbol{\pi}_1, \tau}^\ast	
	}{ 
		f_{vx,\boldsymbol{\pi}_1} 
	}
	\tilde{W}_{\boldsymbol{\pi}_1,\boldsymbol{\pi}_2}
	+
	\frac{ 
		D_{ \boldsymbol{\pi}_3, \tau}^\ast	
	}{ 
		f_{vx,\boldsymbol{\pi}_3} 
	}
	\tilde{W}_{\boldsymbol{\pi}_3,\boldsymbol{\pi}_2}
\right)
\frac{1}{h^{L+1}}
K_{vx,h}
\left(
v_{\boldsymbol{\pi}_3}-v_{\boldsymbol{\pi}_1},  
X_{\boldsymbol{\pi}_3}-X_{\boldsymbol{\pi}_1}
\right) 	
\mid 
\zeta_{\boldsymbol{\pi}_1}
\right]	
\\
=&
E
\left[ 
\left\{ 
\frac{ 
	D_{ \boldsymbol{\pi}_1, \tau}^\ast	
}{ 
	f_{vx,\boldsymbol{\pi}_1} 
}	 
E
\left[ 
	\tilde{W}_{\boldsymbol{\pi}_1,\boldsymbol{\pi}_2}
	\mid 
	X_{\boldsymbol{\pi}_1}
\right]	
+
E
\left[  
		\frac{ 
			D_{ \boldsymbol{\pi}_3, \tau}^\ast	
		}{ 
			f_{vx,\boldsymbol{\pi}_3} 
		}
\mid
v_{\boldsymbol{\pi}_3},  
X_{\boldsymbol{\pi}_3}
\right]
E
\left[ 
	\tilde{W}_{\boldsymbol{\pi}_3,\boldsymbol{\pi}_2}
	\mid 
	X_{\boldsymbol{\pi}_3}
\right]	
\right\}
\right.
\\
&
\qquad
\times 
\left.
\frac{1}{h^{L+1}}
K_{vx,h}
\left(
v_{\boldsymbol{\pi}_3}-v_{\boldsymbol{\pi}_1},  
X_{\boldsymbol{\pi}_3}-X_{\boldsymbol{\pi}_1}
\right) 	
\mid 
\zeta_{\boldsymbol{\pi}_1}
\right]	
\\
=&
\int 
\left\{ 
\frac{ 
	D_{ \boldsymbol{\pi}_1, \tau}^\ast	
}{ 
	f_{vx,\boldsymbol{\pi}_1} 
}	
\chi_{\boldsymbol{\pi}_1}
+
E
\left[  
		\frac{ 
			D_{ \boldsymbol{\pi}_3, \tau}^\ast	
		}{ 
			f_{vx,\boldsymbol{\pi}_3} 
		}
		\chi_{\boldsymbol{\pi}_3}
\mid 
v_{\boldsymbol{\pi}_3}, 
X_{\boldsymbol{\pi}_3}
\right]
\right\}
\frac{1}{h^{L+1}}
K_{vx,h}
\left(
v_{\boldsymbol{\pi}_3}-v_{\boldsymbol{\pi}_1},  
X_{\boldsymbol{\pi}_3}-X_{\boldsymbol{\pi}_1}
\right) 	
f_{vx}(v_{\boldsymbol{\pi}_3}, X_{\boldsymbol{\pi}_3})
d v_{\boldsymbol{\pi}_3}
d X_{\boldsymbol{\pi}_3}
\end{align*}
where the second equality follows from a Law of Iterated Expectations and Assumption \ref{Ass:SR00:iidsampling}.

Let 
\begin{eqnarray*}
\Xi
\left( 
v_{\boldsymbol{\pi}_3},  
X_{\boldsymbol{\pi}_3} 
\right) 
&=&
E
\left[  
D_{ \boldsymbol{\pi}_3, \tau}^\ast	
\chi_{\boldsymbol{\pi}_3}
\mid 
v_{\boldsymbol{\pi}_3}, 
X_{\boldsymbol{\pi}_3}
\right],
\end{eqnarray*}
and consider 
\begin{eqnarray*}
&&
\int 
\left\{ 
\frac{ 
	D_{ \boldsymbol{\pi}_1, \tau}^\ast	
}{ 
	f_{vx,\boldsymbol{\pi}_1} 
}	
\chi_{\boldsymbol{\pi}_1} 
f_{vx}(v_{\boldsymbol{\pi}_3}, X_{\boldsymbol{\pi}_3})
+
\Xi
\left( 
	v_{\boldsymbol{\pi}_3},  
	X_{\boldsymbol{\pi}_3} 
\right) 
\right\}
\frac{1}{h^{L+1}}
K_{vx,h}
\left(
v_{\boldsymbol{\pi}_3}-v_{\boldsymbol{\pi}_1},  
X_{\boldsymbol{\pi}_3}-X_{\boldsymbol{\pi}_1}
\right) 	
d v_{\boldsymbol{\pi}_3}
d X_{\boldsymbol{\pi}_3}
\\
&&
-
\left\{ 
\frac{ 
	D_{ \boldsymbol{\pi}_1, \tau}^\ast	
}{ 
	f_{vx,\boldsymbol{\pi}_1} 
}	
\chi_{\boldsymbol{\pi}_1} 
f_{vx}(v_{\boldsymbol{\pi}_1}, X_{\boldsymbol{\pi}_1})
+
\Xi\left( v_{\boldsymbol{\pi}_1}, X_{\boldsymbol{\pi}_1} \right) 
\right\}
\\
&=&
\int 
\left\{ 
\frac{ 
	D_{ \boldsymbol{\pi}_1, \tau}^\ast	
}{ 
	f_{vx,\boldsymbol{\pi}_1} 
}	
\chi_{\boldsymbol{\pi}_1} 
f_{vx}
(
	v_{\boldsymbol{\pi}_1} + h \boldsymbol{\nu_1}	
	X_{\boldsymbol{\pi}_1} + h \boldsymbol{\nu_2}
)
+
\Xi
\left(
	v_{\boldsymbol{\pi}_1} + h \boldsymbol{\nu_1},  
	X_{\boldsymbol{\pi}_1} + h \boldsymbol{\nu_2}
\right) 
\right\}
K_{vx}\left( \boldsymbol{\nu} \right) 	
d \boldsymbol{\nu}
\\
&&
-
\left\{ 
\frac{ 
	D_{ \boldsymbol{\pi}_1, \tau}^\ast	
}{ 
	f_{vx,\boldsymbol{\pi}_1} 
}	
\chi_{\boldsymbol{\pi}_1} 
f_{vx}(v_{\boldsymbol{\pi}_1}, X_{\boldsymbol{\pi}_1})
+
\Xi\left( v_{\boldsymbol{\pi}_1}, X_{\boldsymbol{\pi}_1} \right) 
\right\}
\\
&=&
\int 
\left( 
\frac{ 
	D_{ \boldsymbol{\pi}_1, \tau}^\ast	
}{ 
	f_{vx,\boldsymbol{\pi}_1} 
}	
\chi_{\boldsymbol{\pi}_1} 
\left\{ 
f_{vx}
(
	v_{\boldsymbol{\pi}_1} + h \boldsymbol{\nu_1}	
	X_{\boldsymbol{\pi}_1} + h \boldsymbol{\nu_2}
)
-
f_{vx}(v_{\boldsymbol{\pi}_1}, X_{\boldsymbol{\pi}_1})
\right\}
\right.
\\
&&
+
\left.
\left\{ 
\Xi
\left(
	v_{\boldsymbol{\pi}_1} + h \boldsymbol{\nu_1},  
	X_{\boldsymbol{\pi}_1} + h \boldsymbol{\nu_2}
\right) 
-
\Xi\left( v_{\boldsymbol{\pi}_1}, X_{\boldsymbol{\pi}_1} \right)
\right\}
\right)
K_{vx}\left( \boldsymbol{\nu} \right) 	
d \boldsymbol{\nu}
\\
&=& 
o (h^{\overline{M}})
\end{eqnarray*}
where the first equality follows from a change of variable
$\boldsymbol{\nu}= (\boldsymbol{\nu}_1, \boldsymbol{\nu}_2)$, with 
$\boldsymbol{\nu}_1=h^{-1}\left( v_{\boldsymbol{\pi}_3}- v_{\boldsymbol{\pi}_1} \right)$,  
and
$\boldsymbol{\nu}_2=h^{-1}\left( X_{\boldsymbol{\pi}_3}- X_{\boldsymbol{\pi}_1} \right)$,  
with Jacobian $h^L$.
The last equality follows Assumptions 
\ref{Ass:Inf01:SamplingMoments},
\ref{Ass:Inf03:SmoothnessDensity}, and \ref{Ass:Inf04:Kernel} 
which guarantee that $f_{vx}(v_{\boldsymbol{\pi}_1}, X_{\boldsymbol{\pi}_1})$ 
and 
$\Xi\left( v_{\boldsymbol{\pi}_1}, X_{\boldsymbol{\pi}_1} \right)$
are continuous and $\overline{M}$-times differentiable with respect to all of its arguments, and 
$K_{vx}$ is a bias-reducing kernel of order $2\overline{M}$.
Observe that 
\begin{eqnarray*}
\frac{ 
	D_{ \boldsymbol{\pi}_1, \tau}^\ast	
}{ 
	f_{vx,\boldsymbol{\pi}_1} 
}	
\chi_{\boldsymbol{\pi}_1} 
f_{vx}(v_{\boldsymbol{\pi}_1}, X_{\boldsymbol{\pi}_1})
&=&0
\end{eqnarray*}
holds for any 
$(v_{\boldsymbol{\pi}_1}, X_{\boldsymbol{\pi}_1})$ 
within a $\tau$ distance of the boundary $\mathbb{S}_{vx}$, 
and having $h/\tau \rightarrow 0$ ensures that the change of variable 
$\boldsymbol{\nu}= (\boldsymbol{\nu}_1, \boldsymbol{\nu}_2)$, with 
$\boldsymbol{\nu}_1=h^{-1}\left( v_{\boldsymbol{\pi}_3}- v_{\boldsymbol{\pi}_1} \right)$,  
and
$\boldsymbol{\nu}_2=h^{-1}\left( X_{\boldsymbol{\pi}_3}- X_{\boldsymbol{\pi}_1} \right)$,  
is not affected by boundary effects.

The previous results yield 
\begin{eqnarray*}
E
\left[ 
p_{\boldsymbol{\pi}_1 , \boldsymbol{\pi}_3  }
\left( \overline{\sigma} \right)	
\mid 
\zeta_{\boldsymbol{\pi}_1}
\right]	
=
D_{ \boldsymbol{\pi}_1, \tau}^\ast	
\chi_{\boldsymbol{\pi}_1}
+
E
\left[ 
D_{ \boldsymbol{\pi}_1, \tau}^\ast
\chi_{\boldsymbol{\pi}_1}
\mid 
X_{\boldsymbol{\pi}_1}
\right]
-
E
\left[ 
D_{ \boldsymbol{\pi}_1, \tau}^\ast
\chi_{\boldsymbol{\pi}_1}
\mid 
\omega_{\boldsymbol{\pi}_1}
\right]
-
E
\left[ 
D_{ \boldsymbol{\pi}_1, \tau}^\ast
\chi_{\boldsymbol{\pi}_1}
\right]
+ o(1).
\end{eqnarray*}

Notice that for $\boldsymbol{\pi}_s \in  \{ \left( \boldsymbol{\pi}_1(1),\boldsymbol{\pi}_2(2)\right), \left( \boldsymbol{\pi}_2(1),\boldsymbol{\pi}_1(2) \right),  \boldsymbol{\pi}_2 \}$,
\begin{eqnarray*}
&&
E
\left[ 
\tilde{W}_{\boldsymbol{\pi}_1,\boldsymbol{\pi}_2}
\left\{ 
\frac{1}{h^{L+1}}
\frac{ 
	D_{ \boldsymbol{\pi}_s, \tau}^\ast
	}{ 
		f_{vx,\boldsymbol{\pi}_s} 	 
}
K_{vx,h}
\left(
v_{\boldsymbol{\pi}_3}-v_{\boldsymbol{\pi}_s},  
X_{\boldsymbol{\pi}_3}-X_{\boldsymbol{\pi}_s}
\right) 
-
E\left[ 
	D_{\boldsymbol{\pi}_s, \tau}^\ast
	\mid 
	\omega_{\boldsymbol{\pi}_s}
\right]
\right\}
\mid 
\zeta_{\boldsymbol{\pi}_1}
\right]	
\\
&=&
E
\left[ 
\tilde{W}_{\boldsymbol{\pi}_1,\boldsymbol{\pi}_2}
\left\{ 
E
\left[ 
\frac{1}{h^{L+1}}
\frac{ 
	D_{ \boldsymbol{\pi}_s, \tau}^\ast
	}{ 
		f_{vx,\boldsymbol{\pi}_s} 	 
}
K_{vx,h}
\left(
v_{\boldsymbol{\pi}_3}-v_{\boldsymbol{\pi}_s},  
X_{\boldsymbol{\pi}_3}-X_{\boldsymbol{\pi}_s}
\right) 
\mid 
\Omega_{\boldsymbol{\pi}_1, \boldsymbol{\pi}_2}
\right]
-
E\left[ 
	D_{\boldsymbol{\pi}_s, \tau}^\ast
	\mid 
	\omega_{\boldsymbol{\pi}_s}
\right]
\right\}
\mid 
\zeta_{\boldsymbol{\pi}_1}
\right]
\\
&=&
O
\left( 
h^{\overline{M}} 
\right)
\end{eqnarray*}
since the expectation 
\begin{eqnarray*}
&&
E
\left[ 
\frac{1}{h^{L+1}}
\frac{ 
	D_{ \boldsymbol{\pi}_s, \tau}^\ast
	}{ 
		f_{vx,\boldsymbol{\pi}_s} 	 
}
K_{vx,h}
\left(
	v_{\boldsymbol{\pi}_3}-v_{\boldsymbol{\pi}_s},  
	X_{\boldsymbol{\pi}_3}-X_{\boldsymbol{\pi}_s}
\right) 
\mid 
\Omega_{\boldsymbol{\pi}_1, \boldsymbol{\pi}_2}
\right]
\\
&=&
\int 
\frac{1}{h^{L+1}}
E
\left[  
\frac{ 
	D_{ \boldsymbol{\pi}_s, \tau}^\ast
	}{ 
		f_{vx,\boldsymbol{\pi}_s} 	 
}
\mid 
\omega_{\boldsymbol{\pi}_s}
\right]
K_{vx,h}
\left(
	v_{\boldsymbol{\pi}_3}-v_{\boldsymbol{\pi}_s},  
	X_{\boldsymbol{\pi}_3}-X_{\boldsymbol{\pi}_s}
\right) 
f_{vx}\left(  v_{\boldsymbol{\pi}_3},  X_{\boldsymbol{\pi}_3} \right)
d v_{\boldsymbol{\pi}_3}
d X_{\boldsymbol{\pi}_3}
\\
&=&
E
\left[  
	D_{ \boldsymbol{\pi}_s, \tau}^\ast
	\mid 
\omega_{\boldsymbol{\pi}_s}
\right]
+
o
\left( 
	h^{\overline{M}} 
\right), 
\end{eqnarray*}
where the second equality follows from 
Assumptions \ref{Ass:SR00:iidsampling}, \ref{Ass:SR01:distr}, 
and properties of the bias-reducing kernel, Assumption \ref{Ass:Inf04:Kernel}.

Similarly, for a given 
$\boldsymbol{\pi}_s \in  \{ \left( \boldsymbol{\pi}_3(1),\boldsymbol{\pi}_2(2)\right), \left( \boldsymbol{\pi}_2(1),\boldsymbol{\pi}_3(2) \right),  \boldsymbol{\pi}_2 \}$, 
it follows from Assumptions \ref{Ass:SR00:iidsampling}, \ref{Ass:SR01:distr}, 
\ref{Ass:Inf03:SmoothnessDensity}, and \ref{Ass:Inf04:Kernel}, that 
\begin{eqnarray*}
&&
E
\left[
\frac{1}{h^{L+1}}	 
\left(			
\frac{ 
	D_{\boldsymbol{\pi}_s, \tau}^\ast	
}{ 
	f_{vx,\boldsymbol{\pi}_s} 
}
\tilde{W}_{\boldsymbol{\pi}_3,\boldsymbol{\pi}_2}
K_{vx,h}
\left(
	v_{\boldsymbol{\pi}_1}-v_{\boldsymbol{\pi}_s},  
	X_{\boldsymbol{\pi}_1}-X_{\boldsymbol{\pi}_s}
\right) 
\right)
\mid 
\zeta_{\boldsymbol{\pi}_1}
\right]
-
\Xi
\left[ 
v_{\boldsymbol{\pi}_1}, 
X_{\boldsymbol{\pi}_1}
\right]
\\
&=&
E
\left[ 
\frac{1}{h^{L+1}}	
E
\left[ 	
	\left(			
	\frac{ 
		D_{ \boldsymbol{\pi}_s, \tau}^\ast	
	}{ 
		f_{vx,\boldsymbol{\pi}_s}  
	}
	\chi_{\boldsymbol{\pi}_s}
	\right)
\mid 
v_{\boldsymbol{\pi}_s} , 
X_{\boldsymbol{\pi}_s} 
\right]
K_{vx,h}
\left( 
v_{\boldsymbol{\pi}_1}-v_{\boldsymbol{\pi}_s}, 
X_{\boldsymbol{\pi}_1}-X_{\boldsymbol{\pi}_s}
\right) 
\mid 
\zeta_{\boldsymbol{\pi}_1}
\right]
-
\Xi
\left[ 
v_{\boldsymbol{\pi}_1}, 
X_{\boldsymbol{\pi}_1}
\right]
\\
&=&
\int 
\left\{ 
\Xi
\left(
	v_{\boldsymbol{\pi}_1} + h \boldsymbol{\nu_1},  
	X_{\boldsymbol{\pi}_1} + h \boldsymbol{\nu_2}
\right) 
-
\Xi\left( v_{\boldsymbol{\pi}_1}, X_{\boldsymbol{\pi}_1} \right)
\right\}
K_{vx}\left( \boldsymbol{\nu} \right) 	
d \boldsymbol{\nu}
\\
&=&
O
\left( 
h^{\overline{M}}	
\right).
\end{eqnarray*}

Using the previous results it follows that 
\begin{eqnarray*}
E
\left[ 
p_{\boldsymbol{\pi}_s, \boldsymbol{\pi}_3}\left( \overline{\sigma} \right) 
\mid 
\zeta_{\boldsymbol{\pi}_1}
\right]	
&=&
E
\left[
D_{\boldsymbol{\pi}_1, \tau}^\ast
\chi_{\boldsymbol{\pi}_1}
\mid 
X_{\boldsymbol{\pi}_1}
\right]
-
E\left[ 
D_{\boldsymbol{\pi}_1, \tau}^\ast
\chi_{\boldsymbol{\pi}_1}
\right],	
\end{eqnarray*}
and thus,  
\begin{eqnarray*}
&&
E
\left[ 
p_{\boldsymbol{\pi}_1, \boldsymbol{\pi}_3}\left( \overline{\sigma} \right) 
-
p_{\boldsymbol{\pi}_1(1)\boldsymbol{\pi}_2(2), \boldsymbol{\pi}_3}\left( \overline{\sigma} \right) 
-
p_{\boldsymbol{\pi}_2(1)\boldsymbol{\pi}_1(2), \boldsymbol{\pi}_3}\left( \overline{\sigma} \right) 
+
p_{\boldsymbol{\pi}_2, \boldsymbol{\pi}_3}\left( \overline{\sigma} \right) 
\mid 
\zeta_{\boldsymbol{\pi}_1}
\right]	
\\
&=&
\left\{ 
D_{ \boldsymbol{\pi}_1}^\ast	
-
E
\left[ 
	D_{ \boldsymbol{\pi}_1}^\ast
	\mid 
	\omega_{\boldsymbol{\pi}_1}
\right]
\right\}
I_{\tau, \boldsymbol{\pi}_1}
\chi_{\boldsymbol{\pi}_1} + o(1)
\end{eqnarray*}

It follows then that the H\'{a}jek projection is given by 
\begin{eqnarray*}
V_{3,n \tau}^\ast
&=&
\frac{1}{n(n-1)}
\sum_{ i_1 =1}^{n}
\sum_{ j_1 \neq i_1 }
	\xi_{i_1j_1, \tau}
+ o(1)	
\end{eqnarray*}
with 
\begin{eqnarray*}
\xi_{i_1j_1, \tau}
&=&
\left\{ 
D^\ast_{ i_1 j_1} 
-
E
\left[ 
	D^\ast_{i_1j_1}
	\mid 
	\omega_{i_1j_1}	
\right]
\right\}
I_{\tau, i_1j_1} 
\overline{\chi}_{i_1j_1}
\\
\overline{\chi}_{i_1j_1}
&=&
\left\{ 
\frac{1}{(n-2)(n-3)}
\sum_{i_2 \neq i_1 , j_1} \sum_{j_2 \neq i_1 , j_1, i_2} 
	E\left[ 
		\tilde{W}_{\sigma\{i_1,i_2;j_1,j_2\}}
		\mid 	
		X_{i_1}, X_{j_1}
	\right]
\right\}.	
\end{eqnarray*}

If follows from a Law of Iterated Expectations that 
\begin{eqnarray*}
E
\left[ 
	V_{3,n \tau}^\ast
\right]
&=&
E
\left[ 
	\xi_{i_1j_1, \tau}	
	\right]
=
0. 
\end{eqnarray*}

\textbf{Step 2. Variance of H\'{a}jek Projection}

As in the proof of Lemma \ref{Lemma:TechAppx:Hajek_Sum1},
the variance of 
$ V_{3,n \tau}^\ast$ is given by 
\begin{eqnarray*}
Var\left( V_{3,n \tau}^\ast \right)
&=&
\left\{ \frac{1}{n(n-1)}\right\}^2
\left\{ 	
\sum_{i_1 =1}^n  \sum_{j_1 \neq i_1} 
E\left[ 
\xi_{i_1j_1, \tau}
\xi_{i_1'j_1', \tau}'
\right] 
\right\}
\\
&=&
\left\{ \frac{1}{n(n-1)}\right\}^2
\left\{ 
\sum_{i_1 =1}^n  \sum_{j_1 \neq i_1} 
\Lambda^\ast_{i_1, j_1}
\right\}
\end{eqnarray*}
where 
\begin{eqnarray*}
\Lambda^\ast_{i_1, j_1}
&=&
E
\left[ 
\left\{ 
	E
	\left[ 
		D^\ast_{i_1j_1}
		D^\ast_{i_1j_1}
		\mid 
		\omega_{i_1j_1}
	\right]	
	-
	E
	\left[ 
		D^\ast_{i_1j_1}
		\mid 
		\omega_{i_1j_1}
	\right]^2		
\right\}
I_{\tau, i_1j_1}^2
\overline{\chi}_{i_1j_1}
\overline{\chi}_{i_1j_1}'
\right].
\end{eqnarray*}

Define 
\begin{eqnarray*}
\Upsilon_{n} &=& n(n-1) Var\left( V_{1,n \tau}^\ast \right)
=
\frac{1}{n(n-1)}
\left\{ 
\sum_{i_1 =1}^n  \sum_{j_1 \neq i_1} 
\Lambda^\ast_{i_1, j_1}
\right\}.
\end{eqnarray*}

\textbf{Step 3. Variance of $S_{3,n \tau}$}

Given two different 6-tuples $\overline{\sigma}\{i_1,i_2,j_1,j_2, l_1,l_2\}$ and $\overline{\sigma}'\{i_1',i_2',j_1',j_2', l_1',l_2'\}$, let 
\begin{eqnarray*}
\Delta_{c,n} 
= 
Cov
\left( 
	s_{3, n} \left( \sigma\{i_1, i_2, j_1, j_2, l_1, l_2\} \right), 
	s_{3, n} \left( \sigma'\{i_1', i_2' , j_1' , j_2', l_1', l_2'\} \right)
\right)	
\end{eqnarray*}
denote the covariance between 
$s_{3, n}(\overline{\sigma})$ 
and 
$s_{3, n}(\overline{\sigma}')$ 
when 
$\overline{\sigma}$ and $\overline{\sigma}'$ have $c=0,1,2,3,4, 5, 6$ indices in common. 

It follows from the conditionally independent formation of links, implied by Assumption \ref{Ass:SR01:distr}, 
and the conditional mean zero, 
$E
\left[ 
s_{3, n}
\left( 
	\sigma\{i_1, i_2, j_1, j_2, l_1, l_2\}
\right)	
\mid \Omega_{\sigma}
\right]=0$,  
that $\Delta_{0,n}= \Delta_{1,n} =0$. 

Consider 
\begin{eqnarray*}
\Delta_{2,n}
&=&
E\left[ 
s_{3, n}(\overline{\sigma}\{i_1,i_2 ,j_1,j_2, l_1, l_2\} )	
s_{3, n}(\overline{\sigma'}\{i_1,i_2',j_1,j_2', l_1, l_2'\} )'	
\right]
\\
&=&
E
\left[
s_{i_1j_1}\left( \overline{\sigma} \right)
s_{i_1j_1}\left( \overline{\sigma'} \right)'
\right]
+
o(1)	
\\
&=&
E
\left[ 
\left\{ 
E
\left[ 
	\tilde{D}^\ast_{i_1j_1,\tau} 
	\tilde{D}^\ast_{i_1j_1,\tau} 
	\mid 
	\omega_{i_1j_1}
\right]	
-
E
\left[ 
	\tilde{D}^\ast_{i_1j_1,\tau} 
	\mid 
	\omega_{i_1j_1}
\right]^2
\right\}
I_{\tau, i_1j_1}^2
\tilde{W}_{\sigma}
\tilde{W}_{\sigma'}
\right]
+
o(1).
\end{eqnarray*}

Therefore, the variance of $Var(S_{2,n \tau}^\dagger)$ can be expressed as 
\begin{eqnarray*}
&&
\left(  \frac{1}{\overline{m}_{n}} \right)^2
\sum_{ \overline{\sigma} } 
\sum_{ \overline{\sigma}' } 
E
\left[ 
	\left( 
	s_{3, n}(\overline{\sigma})
	s_{3, n}(\overline{\sigma'})'
	\right)
\right]
\\
&& 
+
\left(  
4! \binom{n}{4}
\right)^{-2}
\sum_{ i_1=1}^{n} 
\sum_{ j_1 \neq i_1} 
\left\{ 
\sum_{k_1 \neq i_1 , j_1} \sum_{k_2 \neq i_1 , j_1, k_1}
\sum_{l_1 \neq i_1 , j_1} \sum_{l_2 \neq i_1 , j_1, l_1}
\Delta_{2,n}
\right\}
\\
&&
+
O\left(\frac{1}{n^3}\right)
\Delta_{3,n}
+
O\left(\frac{1}{n^4}\right)
\Delta_{4,n}
+
O\left(\frac{1}{n^5}\right)
\Delta_{5,n}
+
O\left(\frac{1}{n^6}\right)
\Delta_{6,n}
\end{eqnarray*}

Notice that the term inside the brackets scaled by $((n-2)(n-3))^{-2}$ can be written as 
\begin{eqnarray*}
&&
\left(
	\frac{1}{(n-2) (n-3)}
\right )^2	
\sum_{k_1 \neq i_1 , j_1} \sum_{k_2 \neq i_1 , j_1, k_1}
\sum_{l_1 \neq i_1 , j_1} \sum_{l_2 \neq i_1 , j_1, l_1}
\Delta_{2,n}
\\
&=&
E
\left[ 
	\left\{ 
		E
		\left[ 
			D^\ast_{i_1j_1}
			D^\ast_{i_1j_1}
			\mid 
			\omega_{i_1j_1}
		\right]	
		-
		E
		\left[ 
			D^\ast_{i_1j_1}
			\mid 
			\omega_{i_1j_1}
		\right]^2		
	\right\}
	I_{\tau, i_1j_1}^2
	\overline{\chi}_{i_1j_1}
	\overline{\chi}_{i_1j_1}'
\right]
\\
&=& 
\Lambda^\ast_{i_1,j_1}.
\end{eqnarray*}

As a result, 
\begin{eqnarray*}
Var
\left[ 
	S_{3,n \tau}^\dagger 
\right]
&=&
\left(  \frac{1}{n(n-1)} \right)^2
\left\{ 
\sum_{ i_1=1}^{n} 
\sum_{j_1\neq i_1}
	\Lambda^\ast_{i_1,j_1} 
\right\}
+ o(1),
\end{eqnarray*}
and $Var \left[  V_{3,n \tau}^\ast   \right] - Var \left[  S_{3,n \tau}^\dagger \right] = o_p(1)$.

The asymptotic equivalence results follows from similar arguments as in the proof of Lemma \ref{Lemma:TechAppx:Hajek_Sum1}.
The proof is complete.

\end{proof}	
\newpage
\section{Simulations: alternative designs}
\label{Appx:simulations}
\begin{table}[h!]       
	\centering
	\caption{Simulation results for the semiparametric estimator $\widehat{\theta}_n$ with kernel estimator $\widehat{f}_{v}(v_{ij})$}
	\begin{threeparttable}
	   \begin{tabular}{l c c c c c}
\toprule
                & mean &median & std & MSE & Degree \\ \midrule
\multicolumn{6}{c}{$n=50$} \\                
$\log(\log(n))$     &1.6047&1.6164&1.1253&1.2772&0.4237\\
$\log(n)^{1/2}$     &1.6630&1.6179&1.2352&1.5522&0.3963\\
$\log(n)$           &1.6444&1.6643&1.5801&2.5176&0.3125\\
\multicolumn{6}{c}{$n=100$} \\               
$\log(\log(n))$     &1.5373&1.5011&0.4911&0.2425&0.4214\\
$\log(n)^{1/2}$     &1.5955&1.5778&0.5547&0.3168&0.3859\\
$\log(n)$           &1.5415&1.5197&0.7317&0.5371&0.2907\\
\bottomrule
\end{tabular}
\begin{tablenotes}
    \item[1] \footnotesize{Total number of Monte Carlo simulations $=500$.}
    \item[2] \footnotesize{Bandwith parameter $h=0.025$.}
\end{tablenotes}

	\end{threeparttable}   
\end{table}
\begin{table}[h!]       
	\centering
	\caption{Simulation results for the semiparametric estimator $\widehat{\theta}_n$ with kernel estimator $\widehat{f}_{v}(v_{ij})$}
	\begin{threeparttable}
	   \begin{tabular}{l c c c c c}
\toprule
                & mean &median & std & MSE & Degree \\ \midrule
\multicolumn{6}{c}{$n=50$} \\                
$\log(\log(n))$     &1.6296&1.5640&1.1280&1.2893&0.4252\\
$\log(n)^{1/2}$     &1.6236&1.5961&1.1864&1.4229&0.3979\\
$\log(n)$           &1.6308&1.6379&1.5430&2.3981&0.3127\\
\multicolumn{6}{c}{$n=100$} \\               
$\log(\log(n))$     &1.5944&1.5782&0.4999&0.2588&0.4218\\
$\log(n)^{1/2}$     &1.5603&1.5563&0.5452&0.3009&0.3863\\
$\log(n)$           &1.5009&1.5244&0.7059&0.4983&0.2896\\
\bottomrule
\end{tabular}
\begin{tablenotes}
    \item[1] \footnotesize{Total number of Monte Carlo simulations $=500$.}
    \item[2] \footnotesize{Bandwith parameter $h=0.05$.}
\end{tablenotes}

	\end{threeparttable}   
\end{table}
\newpage
\begin{table}[h!]       
	\centering
	\caption{Simulation results for the semiparametric estimator $\widehat{\theta}_n$ with kernel estimator $\widehat{f}_{v}(v_{ij})$}
	\begin{threeparttable}
	   \begin{tabular}{l c c c c c}
\toprule
                & mean &median & std & MSE & Degree \\ \midrule
\multicolumn{6}{c}{$n=50$} \\                
$\log(\log(n))$     &1.7149&1.7235&1.1336&1.3313&0.4252\\
$\log(n)^{1/2}$     &1.6486&1.6280&1.2045&1.4729&0.3973\\
$\log(n)$           &1.5690&1.5839&1.5592&2.4358&0.3116\\
\multicolumn{6}{c}{$n=100$} \\               
$\log(\log(n))$     &1.5394&1.5478&0.4973&0.2488&0.4212\\
$\log(n)^{1/2}$     &1.5443&1.5336&0.5533&0.3081&0.3855\\
$\log(n)$           &1.5662&1.6033&0.7749&0.6049&0.2905\\
\bottomrule
\end{tabular}
\begin{tablenotes}
    \item[1] \footnotesize{Total number of Monte Carlo simulations $=500$.}
    \item[2] \footnotesize{Bandwith parameter $h=0.1$.}
\end{tablenotes}

	\end{threeparttable}   
\end{table}
\begin{table}[h!]       
	\centering
	\caption{Simulation results for the semiparametric estimator $\widehat{\theta}_n$ with kernel estimator $\widehat{f}_{v}(v_{ij})$}
	\begin{threeparttable}
	   \begin{tabular}{l c c c c c}
\toprule
                & mean &median & std & MSE & Degree \\ \midrule
\multicolumn{6}{c}{$n=50$} \\                
$\log(\log(n))$     &1.6675&1.6378&1.0617&1.1552&0.4250\\
$\log(n)^{1/2}$     &1.6453&1.6463&1.2179&1.5044&0.3974\\
$\log(n)$           &1.6594&1.6162&1.5833&2.5321&0.3113\\
\multicolumn{6}{c}{$n=100$} \\               
$\log(\log(n))$     &1.5577&1.5512&0.5305&0.2848&0.4208\\
$\log(n)^{1/2}$     &1.5739&1.5653&0.5594&0.3184&0.3852\\
$\log(n)$           &1.5653&1.5637&0.7064&0.5033&0.2898\\
\bottomrule
\end{tabular}
\begin{tablenotes}
    \item[1] \footnotesize{Total number of Monte Carlo simulations $=500$.}
    \item[2] \footnotesize{Bandwith parameter $h=0.2$.}
\end{tablenotes}

	\end{threeparttable}   
\end{table}
\end{document}